%% file: InformationError.tex
\documentclass[reqno,11pt]{article} 
\usepackage[margin=1in]{geometry}
\usepackage{amsmath,amsthm,amssymb,amscd,verbatim,framed,mathtools}
\usepackage{graphicx}
\usepackage{stmaryrd}
\input{stdinc}

\input{header}

\usepackage{amsmath,amsfonts,amssymb,amsthm}
\usepackage{color}

\let\Ex\relax

\newif\ifdraft
\drafttrue

\ifdraft
\providecommand{\yuvalcomment}[1]{{\color{red} #1---YF}}
\providecommand{\ydcomment}[1]{{\color{brown} #1 --- YD}}
\providecommand{\yl}[1]{{\color{blue} #1 --- Yaqiao}}
\providecommand{\hamed}[1]{{\color{purple} #1 --- HH}}
\else
\providecommand{\yuvalcomment}[1]{}
\providecommand{\ydcomment}[1]{}
\providecommand{\yl}[1]{}
\providecommand{\hamed}[1]{}
\fi

\providecommand{\ext}{\mathit{ext}}
\DeclareMathOperator{\IW}{IW}
\DeclareMathOperator{\SIM}{SIM}

\DeclareMathOperator*{\Ex}{\mathbb{E}}
\DeclareMathOperator{\AND}{AND}
\DeclareMathOperator{\XOR}{XOR}
\DeclareMathOperator{\DISJ}{DISJ}

\DeclareMathOperator{\IC}{IC}
\DeclareMathOperator{\CI}{CI}
\providecommand{\ICD}{\IC^D}
\providecommand{\ICND}{\IC}        
\providecommand{\ICNDz}{\IC^{0}}   
\providecommand{\charf}[1]{1_{#1}}
\providecommand{\entf}{h}

\providecommand{\leafp}{\mathbf{p}}
\providecommand{\leafq}{\mathbf{q}}
\providecommand{\leafl}{\boldsymbol{\ell}}
\providecommand{\rmd}{\mathrm{d}}
\providecommand{\cX}{\mathcal{X}}
\providecommand{\cY}{\mathcal{Y}}

\providecommand{\ICD}{\IC^D}
\providecommand{\ICND}{\IC^{ND}}
\providecommand{\ICNDz}{\IC^{ND,0}}
\providecommand{\charf}[1]{1_{#1}}
\providecommand{\hc}{\overline{h}}
\providecommand{\omu}{{\overline{\mu}}}
\providecommand{\oy}{\overline{y}}

\providecommand{\proofpar}[1]{\textbf{#1} \quad}

\providecommand{\defeq}{\vcentcolon=}
\renewcommand{\defeq}{\vcentcolon=}

\begin{document}

\title{Trading information complexity for error}


\author{Yuval Dagan~\thanks{Technion --- Israel Institute of Technology. \texttt{yuval.dagan@cs.technion.ac.il}} \and Yuval Filmus \thanks{Technion --- Israel Institute of Technology. \texttt{yuvalfi@cs.technion.ac.il}} \and Hamed Hatami \thanks{McGill University. \texttt{hatami@cs.mcgill.ca}. Supported by an NSERC grant.} \and Yaqiao Li \thanks{McGill University. \texttt{yaqiao.li@mail.mcgill.ca}}}

\maketitle

\begin{abstract}

We consider the standard two-party communication model. The central problem studied in this article is how much one can save in  information complexity by allowing an error of $\epsilon$.

\begin{itemize}
\item For arbitrary functions, we obtain lower bounds and upper bounds indicating a gain that is of  order $\Omega(h(\epsilon))$ and $O(h(\sqrt{\epsilon}))$. Here $h$ denotes the binary entropy function.  
\item We analyze the case of the two-bit AND function in detail to show that for this function the gain is  $\Theta(h(\epsilon))$. This answers a question of Braverman et al.~\cite{MR3210776}.
\item We obtain sharp bounds for the set disjointness function of order $n$. For the case of the distributional error, we introduce a new protocol that achieves a gain of $\Theta(\sqrt{h(\epsilon)})$ provided that $n$ is sufficiently large.  We apply these results to answer another of question of Braverman et al. regarding the randomized communication complexity of the set disjointness function.
\item Answering a question of  Braverman~\cite{MR2961528}, we apply our analysis of the set disjointness function to establish a gap between the two different notions of the prior-free information cost. In the light of~\cite{MR2961528}, this implies that amortized randomized communication complexity is not necessarily equal to the amortized distributional communication complexity with respect to the hardest distribution.
\end{itemize}
\end{abstract}

\section{Introduction}

In recent years, a focus on the applications of the information theoretic methods to the area of communication complexity has resulted in a new and deeper understanding of some of the classical problems in this area. These developments have given rise to a new field of complexity theory called information complexity, which was first formally defined in \cite{MR1948715, MR2059642, MR2743255}.  While communication complexity is concerned with minimizing the amount of communication required for two players to evaluate a function, information complexity, on the other hand, is concerned with the amount of information that the communicated bits reveal about the inputs of the players.

The study of information complexity is  motivated by fundamental questions regarding compressing communication~\cite{MR2743255,MR3265014,MR2961528,MR3388235} that extend the seminal work of Shannon~\cite{MR0026286} to the setting where interaction is allowed. Moreover it has important applications to communication complexity, and in particular to the study of the direct-sum problem~\cite{MR2059642,MR1948715,MR3366999,MR3246278,MR3109074}, a problem that has been studied extensively in the past~\cite{MR1342989, MR1948715, MR2080709, MR2589281, MR2743255, MR2743256, MR3366999, MR3186603, MR3246278, MR3109074}. For example, the only known direct-sum result for general randomized communication complexity is proven via information theoretic techniques in~\cite{MR2743255}.

Arguably, the randomized communication complexity of a function $f$, often denoted by $R_\eps(f)$, is the most important object of study in the area of communication complexity. This quantity corresponds to the smallest number of bits that two players need to exchange so that they can compute the value of $f(x,y)$ correctly with probability at least $1-\eps$ on the worst input $(x,y)$. Determining the asymptotics of $R_\epsilon(f)$ can be very difficult, and as a result the  focus of the area has mainly been on determining the growth rate of this function in the Big-O sense. Note that for $\epsilon<1/2$, the players can decrease the probability of error by running the same protocol multiple times. Hence for constant values of $\eps \in (0,1/2)$, the quantities $R_\eps(f)$ are within constant multiples of each other, and as a result, as long as one is concerned only with the Big-O asymptotics, it is possible to fix the error parameter to a constant such as $1/3$.

Set disjointness is one of the most important functions in communication complexity, and as a result it has been studied extensively in the past four decades (see the surveys~\cite{chattopadhyay2010story,MR3253040} and the references therein).  In this communication problem, which is denoted  by $\DISJ_n$,  Alice and Bob each receive a subset of $\{1,\ldots,n\}$ and their goal is to determine whether their sets are disjoint or not. The correct asymptotic $R_\eps(\DISJ_n) = \Theta(n)$ was first proved by Kalyanasundaram and Schnitger~\cite{MR1186822}. Although later Razborov~\cite{MR1192778}  gave a shorter proof, still despite several decades of research in this area, all known proofs for this fact are intricate and sophisticated. It was thus a great breakthrough when  a recent article~\cite{MR3210776}  determined the exact constant in the asymptotics of $R_\eps(\DISJ_n)$ as $\eps \to 0$ by employing several recent results from the area of information complexity. They proved that as the error-parameter $\eps$ tends to $0$, the quantity $\lim_{n \rightarrow \infty} R_\epsilon(\DISJ_n)/n$ tends to a constant $\approx 0.4827$.

This and similar recent results show that the area of communication complexity has developed to a point where even for difficult functions such as set disjointness, analyzing  the asymptotic of $R_\eps(f)$ in a precision beyond the Big-O approximation has become possible. This might be an indication that some problems regarding the dependency of the  randomized communication complexity and the information complexity on the error parameter $\eps$ might now be within reach. The purpose of this article is to conduct a systematic study of such problems. In doing so we answer several open problems and conjectures that were raised previously in the literature.

\paragraph{Information complexity:} Consider finite sets $\cX,\cY, \cZ$,  a function $f\colon \cX \times \cY \to \cZ$, and a two-party communication protocol $\pi$ for the task of computing the value of $f$.  In order to define information complexity, one needs to assume that $\cX \times \cY$ is endowed with a probability distribution $\mu$, and  Alice and Bob's inputs $X \in \cX$ and $Y \in \cY$ are sampled according to this joint distribution. Now using Shannon's notion of information, we can consider the amount of information that the players learn  about each other's inputs from the exchanged bits in $\pi$. The amount of this leaked information is called the information complexity of $\pi$, and it is denoted by $\IC_\mu(\pi)$. The formal definition is given below in Definition~\ref{def:infocost}.  Let $\IC_\mu(f,\epsilon)$ denote the infimum of $\IC_\mu(\pi)$ among all randomized protocols $\pi$ that compute $f$ with probability of error at most $\epsilon$ on every input. Similarly define the  information complexity $\IC_\mu(f,\mu,\epsilon)$ under distributional error as the infimum of $\IC_\mu(\pi)$ among all randomized protocols $\pi$ that compute $f(X,Y)$ with probability of error at most $\epsilon$ when $(X,Y)$ is sampled according to $\mu$.

One can define the  prior-free information complexity of a function in two different ways. The first notion is $ \IC(f,\eps)= \max_{\mu} \IC_\mu(f, \eps)$
which has been proven very useful, as it captures the amortized communication complexity required for calculating multiple copies of $f$ with $\epsilon$ error on each copy. The second natural definition is to consider $\ICD(f,\eps)= \max_{\mu} \IC_\mu(f, \mu, \epsilon)$, which trivially satisfies $\ICD(f,\eps)  \le \IC(f,\eps)$. These two notions and their relation to each other are first studied by Braverman in~\cite{MR2961528}.

%

Finally let us mentioned it is  also natural to consider the amount of the information that an external observer learns about the players' inputs by observing the exchanged bits in a protocol $\pi$. This leads to the notion of external information complexity $\IC_\mu^\ext(\pi)$ that is formally defined below in Definition~\ref{def:infocost}.

\subsection{Our contributions}
In this section we briefly describe our main results. A more detailed description is given later in Section~\ref{sec:main_results} after we introduce some preliminary facts and definitions in Section~\ref{sec:Prelim}.

The central problem studied in this article is how much one can save in  information complexity by allowing an error of $\epsilon$. We start by considering the point-wise error case, and proving upper bounds and lower bounds for $\IC_\mu(f,\epsilon)$. Then we move to the case of the distributional error and study $\IC_\mu(f,\mu,\epsilon)$. Afterwards, we study these  parameters for two special cases, the AND function and the set disjointness function, in great detail. As we shall see, this will have important implications regarding the prior-free information complexity.

\paragraph{Information complexity with point-wise error:} 

In Theorems~\ref{thm:upper-bd-IC-mu-eps} and~\ref{thm:non-distri-error-lowerbd} we consider the point-wise error case and prove that for every distribution $\mu$  with $\IC_\mu(f, 0) > 0$,
\begin{equation}
\label{eq:pointwise-error}
\IC_\mu(f,0) - O(h(\sqrt{\epsilon})) \le \IC_\mu(f,\epsilon) \le \IC_\mu(f,0) - \Omega(h(\epsilon)).
\end{equation}
Note that the lower bound implies the continuity of  $\IC_\mu(f,\epsilon)$ with respect to $\eps$ at $\eps=0$. Theorem~\ref{thm:AND-gap} shows that the upper bound in (\ref{eq:distributional-error}) can be tight for the AND function, while for the lower bound, to the best of our knowledge, it might be possible that $\IC_\mu(f,\eps)= \IC_\mu(f, 0) - \Theta(h(\eps))$ for every $f,\mu$ with $\IC_\mu(f, 0)>0$.

It is worth noting that, maybe surprisingly, the upper bound does not hold for external information complexity. Indeed in Proposition~\ref{prop:ext_XOR}, we show that for certain distributions for the two-bit $\XOR$ function, the gain in the external information complexity is only of order $\Theta(\epsilon)$.

\paragraph{Information complexity with distributional error:}
It is shown in~\cite{MR3210776} that  $\IC_\mu(f,\mu,\epsilon)$ is continuous with respect to $\epsilon$.  The continuity at $\eps>0$ is  easy to show, and it was proven earlier in~\cite{MR2961528},  however the case $\eps=0$ is more subtle and it is  established in~\cite{MR3210776} through the inequality
\[ \IC_\mu(f,\mu,0) -O(h(\eps^{1/8})) \le \IC_\mu(f,\mu,\eps), \]
where $h(\cdot)$ denotes the binary entropy function (See Section~\ref{sec:estimates}). In Theorem~\ref{thm:bd-IC-mu--distributional} we improve this result by showing that for every distribution $\mu$ with $\IC_\mu(f, \mu, 0) > 0$, 
\begin{equation}
\label{eq:distributional-error}
\IC_\mu(f,\mu,0)-O(h(\sqrt{\eps})) \le \IC_\mu(f,\mu,\eps) \le \IC_\mu(f,\mu,0)- \Omega(h(\epsilon)).
\end{equation}
Later in Theorem~\ref{thm:AND-gap}, we prove that for the two-bit AND function, the upper bound in (\ref{eq:distributional-error}) is tight provided that $\mu$ is of full support. We do not know whether the lower bound is sharp. In fact we are not aware of any example that does not satisfy $\IC_\mu(f,\mu,\eps)= \IC_\mu(f, \mu, 0) - \Theta(h(\eps))$.

\paragraph{Prior-free information complexity:}

In~\cite{MR2961528} Braverman defined the two notions of the prior-free information complexity. Combining our aforementioned results regarding $\IC_\mu(f,\mu,\eps)$ and $\IC_\mu(f,\eps)$, we show that for every function $f$, we have
	\[
		\IC(f,0) - O(h(\sqrt{\epsilon})) \le \IC(f,\epsilon) \le \IC(f,0) - \Omega(h(\epsilon)),
	\]
and	
	\[
		\ICD(f,0) - O(h(\sqrt{\epsilon})) \le \ICD(f,\epsilon) \le \ICD(f,0) - \Omega(h(\epsilon)).
	\]
The upper bounds are both tight for $f=\AND$.

Braverman~\cite{MR2961528}  showed that for $\eps=0$, the two notions of the prior-free information complexity coincide: $\IC(f,0)=\ICD(f,0)$. Moreover he proved  that interestingly, for every $0<\alpha<1$,
\[ \IC(f,\eps/\alpha) \le \frac{\ICD(f,\eps)}{1-\alpha}. \]
In particular  setting $\alpha=1/2$ yields  $\IC(f,2\eps) \le2 \ICD(f,\eps)$.

In~\cite{MR2961528} he asks  whether  there is a gap between  $\IC(f,\eps)$ and $\ICD(f,\eps)$ for $\eps>0$. As we  explain below, our analysis of the set disjointness function answers this question in the affirmative, and shows that the inequality $\ICD(f,\eps)  \le \IC(f,\eps)$ can be strict. Indeed we show that Braverman's analysis is tight, and  for $\alpha=\sqrt{\eps \log(1/\eps)}$, we have
\begin{equation}
\label{eq:priorfree_separation_thightness}
\frac{\ICD(\DISJ_n,\eps)}{1-\Theta(\alpha)}=\IC(\DISJ_n,\eps/\alpha) < \IC(\DISJ_n,\eps),
\end{equation}
provided that $n$ is sufficiently large.

Let $R_\epsilon^n(f^n)$ denote the randomized communication complexity of computing $n$ copies of $f$ such that on each set of $n$ inputs the probability of failure on each of the inputs is at most $\epsilon$, and let $D_\epsilon^{\mu,n}(f^n)$ denote the corresponding distributional communication complexity where each of the $n$ independent pairs of inputs drawn from $\mu$. In~\cite{MR2961528} it was proven that $\IC(f,\epsilon) = \lim_{n\to \infty} R_\epsilon^n(f^n)/n$ and $\IC_\mu(f,\mu,\epsilon) = \lim_{n\to \infty} D_\epsilon^{\mu,n}(f^n)/n$; the latter implies $\ICD(f,\epsilon) = \lim_{n\to \infty} \max_\mu D_\epsilon^{\mu,n}(f^n)/n$. As \eqref{eq:priorfree_separation_thightness} shows that $\IC(f,\epsilon)$ and $\ICD(f,\epsilon)$ do not necessarily coincide, we conclude that $\max_\mu D_\epsilon^{\mu,n}(f^n) = R_\epsilon^n(f^n)$ does not always hold either.

\paragraph{Information complexity of the AND function with error:}

In \cite[Problem 1.1]{MR3210776} Braverman et al. posed the problem of determining the prior-free information cost of the two-bit $\AND$ function with error of at most $\epsilon$. In particular, they conjecture that $\IC(\AND,\epsilon) \le \IC(\AND,0)-  \Omega(h(\epsilon))$. Our general upper bounds for prior-free information complexity settles this conjecture in the affirmative. Furthermore, in Corollary~\ref{cor:AND-gap}, we will prove a matching lower bound to show that indeed
\begin{equation}
\label{eq:AND_ND_priorfree_eps}
 \IC(\AND,\epsilon) = \IC(\AND,0) - \Theta(h(\epsilon)).
\end{equation}
In order to achieve this, we first start by proving a lower bound for the point-wise error case. This proof turns out to be rather involved and it contains  new components and ideas. First we introduce a potential function, and use it to show that the optimal protocol for the AND function is stable in the sense that if a protocol for the AND function has almost optimal information cost, then it has to share certain similarities with the optimal protocol (i.e. the so called buzzer protocol). We use this to show that even an $\epsilon$-error protocol with small information cost has to look somewhat similar to the buzzer protocol, and from that we obtain a lower bound on its information cost. More precisely we show that
\[
 \IC_\mu(\AND,\eps) = \IC_\mu(\AND,0) - \Theta(h(\epsilon)),
\]
where the asymptotic notation holds as $\epsilon \to 0$, and the hidden constant depends on $\mu$. However under certain conditions on $\mu$ one can obtain a uniform constant that does not depend on $\mu$. We combine this with some earlier results of \cite{MR3210776} and \cite{MR2961528} to obtain \eqref{eq:AND_ND_priorfree_eps} and its distributional analogue 
\[
 \ICD(\AND,\epsilon) = \ICD(\AND,0) - \Theta(h(\epsilon)).
\]

\paragraph{The communication complexity of the disjointness:} As we mentioned earlier Braverman et al.~\cite{MR3210776} determined the asymptotics of the randomized communication complexity of the set disjointness problem as $n \to \infty$. They showed that
\[ \lim_{\eps \to 0} \lim_{n \to \infty} \frac{R_\eps(\DISJ_n)}{n}=C_{\DISJ}, \]
where $C_{\DISJ}=\ICNDz(\AND,0)\approx 0.4827$. Here $\ICNDz(\AND,0)$ is defined similar to $\IC(\AND,0)$  but allowing only distributions that put a $0$ mass on the point $(1,1)$. That is $\ICNDz(\AND,0)=\max \IC_\mu(\AND,0)$, where the maximum is over all $\mu$ with $\mu(1,1)=0$.

Regarding the asymptotics of $R_\eps(\DISJ_n)$  with respect to $\eps$, they conjectured that for every $\eps>0$, the limit
\[C_{\DISJ_\eps} \defeq \lim_{n \to \infty} \frac{R_\eps(\DISJ_n)}{n}\]
is strictly smaller than $C_{\DISJ}$. Moreover they posed determining the asymptotics of $C_{\DISJ}-C_{\DISJ_\eps}$ as an open problem.  We resolve these questions by proving that $ C_{\DISJ_\eps} = C_{\DISJ} - \Theta(h(\epsilon))$.
	
In Theorem~\ref{thm:setDisj_distrib}, we  introduce a protocol for the set disjointness function  that has  distributional error $\eps$, and has small information cost. This will show
\[
 \ICD(\DISJ_n,\epsilon) =n[\ICNDz(\AND,0) - \Theta(\sqrt {h( \epsilon)})]+O(\log n).
\]

 On the other hand in Corollary~\ref{cor:disjointness_IC_ND}, we will show 
\[
 \IC(\DISJ_n,\epsilon) \geq n [\ICNDz(\AND,0) - \Theta(h(\epsilon))].
\]
This will separate the two notions of prior-free information complexity, and yield \eqref{eq:priorfree_separation_thightness}.

\paragraph{Characterization of trivial measures:}
In order to prove \eqref{eq:pointwise-error}, we will first need to characterize all measures that satisfy $\IC_\mu(f,0)=0$, and analogously for the external case $\IC^\ext_\mu(f,0)=0$. We call such measures, respectively, internal-trivial and external-trivial for $f$. In Theorems~\ref{thm:internal-trivial}~and~\ref{thm:external-trivial}, we obtain a characterization of such measures.

\subsection{Preliminaries}\label{sec:Prelim}
In this section we introduce some basic notation and facts, and review the necessary background for the paper.

\subsubsection{Notation and basic estimates}  \label{sec:estimates}

We typically denote the random variables by capital letters (e.g $A,B,C,\Pi$). For the sake of brevity, we shall write $A_1\ldots A_n$ to denote the random variable $(A_1,\ldots,A_n)$ and \emph{not} the product of the $A_i$'s. We use $[n]$ to denote the set $\{1,\ldots,n\}$, and $\supp\mu$ to denote the support of a measure $\mu$.

For a finite set $\Omega$, we denote by  $\Delta(\Omega)$, the set of all  discrete probability distributions on $\Omega$.
For $\mu,\nu \in \Delta(\Omega)$, we denote their \emph{total variation distance} with
\[ |\mu-\nu| \defeq \frac{1}{2} \sum_{x \in \Omega} |\mu(x)-\nu(x)|. \]

For every $\epsilon \in [0,1]$, $h(\epsilon) = -\epsilon \log\epsilon - (1-\epsilon) \log(1-\epsilon)$ denotes the \emph{binary entropy}, where here and throughout the paper $\log(\cdot)$ is in base $2$, and $0 \log 0 = 0$.

\subsubsection{Communication complexity}
The notion of two-party communication complexity was introduced by Yao~\cite{Yao:1979} in 1979. In this model there are two players (with unlimited computational power), often called Alice and Bob, who wish to collaboratively perform a task such as computing a given function $f\colon \cX \times \cY \to \cZ$. Alice receives an input $x \in \cX$ and Bob receives $y \in \cY$. Neither of them knows the other player's input, and they wish to communicate in accordance with an agreed-upon protocol $\pi$ to  compute $f(x,y)$.  The protocol $\pi$ specifies as a function of (only) the transmitted bits  whether the communication is over, and if not, who sends the next bit. Furthermore $\pi$ specifies what the next bit must be as a function of the transmitted bits, and the input of the player who sends the bit.  We will assume that when the protocol terminates Alice and Bob agree on a value as the output of the protocol. We denote this value by $\pi(x,y)$. The \emph{communication cost} of $\pi$ is the total number of bits transmitted on the worst case input. The \emph{transcript} of an execution of $\pi$ is a string $\Pi$ consisting of a list of all the transmitted bits during the execution of the protocol.  As protocols are defined using protocol trees, transcripts are in one-to-one correspondence with the leaves of this tree.

In the randomized communication model, the players might have access to a shared random string (\emph{public randomness}), and their own private random strings (\emph{private randomness}).  These random strings are independent, but they can have any desired distributions individually. In the randomized model the \emph{transcript} also includes the public random string  in addition to the transmitted bits. Similar to the case of deterministic protocols, the \emph{communication cost} is the total number of bits transmitted on the worst case input and random strings. The \emph{average communication cost} of the protocol is the expected number of bits transmitted on the worst case input.

For a function $f\colon \cX \times \cY \to \cZ$ and a parameter $\epsilon>0$, we denote by $R_\epsilon(f)$ the communication cost of the best randomized protocol that computes the value of $f(x,y)$ correctly with probability at least $1-\eps$ for \emph{every} $(x,y)$.

\subsubsection{Information complexity}  \label{sec:IC-definition}

The setting is the same as  in communication complexity, where Alice and Bob (having infinite computational power) wish to mutually compute a function $f\colon \cX \times \cY \to \cZ$. To be able to measure information, we also need to assume that there is a prior distribution $\mu$ on $\cX \times \cY$.

 For the purpose of communication complexity, once we allow public randomness, it makes no difference whether we permit the players to have private random strings or not. This is because the private random strings can be simulated by parts of the public random string. On the other hand, for information complexity, it is crucial to  permit private randomness, and once we allow private randomness, public randomness becomes inessential. Indeed, one of the players can use her private randomness to generate the public random string, and then transmit it to the other player. Although this might have very large communication cost, it has no information cost, as it does not reveal any information about the players' inputs.

Probably the most natural way to define the information cost of a protocol is to consider the amount of information that is revealed about the inputs $X$ and $Y$ to an external observer who sees the transmitted bits and the public randomness. This is called  the \emph{external information cost} and is formally defined as the mutual information between $XY$ and the transcript of the protocol (recall that the transcript also contains the public random string). While this notion is interesting and useful, it turns out there is a different way of defining the  information cost that   enjoys  certain desirable properties that the external information cost lack. This is called the \emph{internal information cost} or just the \emph{information cost} for short, and is equal to  the amount of information that Alice and Bob learn about each other's inputs from the communication. Note that Bob knows $Y$, the public randomness $R$, and his own private randomness $R_B$, and thus what he learns about $X$ \emph{from the communication} can be measured by the conditional mutual information  $I(X;\Pi|YRR_B)$. Similarly, what Alice learns about $Y$ from the communication can be measured by  $I(Y;\Pi|XRR_A)$ where $R_A$ is Alice's private random string. It is not  difficult to see~\cite{MR2743255} that  conditioning on the public and private randomness does not affect these quantities. In other words $I(X;\Pi|YRR_B)=I(X;\Pi|Y)$ and $I(Y;\Pi|XRR_A)= I(Y;\Pi|X)$. We summarize these in the following definition.

\begin{definition}
\label{def:infocost}
The \emph{internal information cost} and the \emph{external information cost}  of a protocol $\pi$ with respect to a distribution $\mu$ on inputs from $\cX \times \cY$ are defined as
\[ \IC_\mu(\pi) = I(\Pi; X|Y)+I(\Pi; Y|X), \]
and
\[ \IC_\mu^\ext(\pi) = I(\Pi; XY), \]
respectively, where $\Pi=\Pi_{XY}$ is the transcript of the protocol when it is executed on  $XY$.
\end{definition}

We will be interested in certain \emph{communication tasks}. Let $[f,\eps]$ denote the task of computing  the value of $f(x,y)$ correctly with probability at least $1-\eps$ for \emph{every} $(x,y)$. Thus a protocol $\pi$ performs this task  if
\[
\Pr[\pi(x,y) \neq f(x,y)] \le \epsilon, \quad \forall\ (x,y) \in \cX \times \cY.
\]
Given another distribution $\nu$ on $\cX \times \cY$, let $[f, \nu, \epsilon]$ denote the task of computing  the value of $f(x,y)$ correctly with probability at least $1-\eps$ if the input $(x,y)$ is sampled from the distribution $\nu$. A protocol $\pi$ performs this task  if
\[
\Pr_{(x,y) \sim \nu}[\pi(x,y) \neq f(x,y)] \le \epsilon.
\]
Note that a protocol $\pi$ performs $[f, 0]$ if it  computes $f$ correctly on \emph{every} input while performing $[f, \nu,0]$  means computing $f$ correctly on the inputs that belong to support of $\nu$. 

We will also need a one-sided version of the task $[f,\eps]$. Let $[f,\eps,z_1\to z_0]$ denote the task of computing the value of $f(x,y)$ correctly with probability at least $1-\eps$ for \emph{every} $(x,y)$, allowing the protocol to err only if it outputs $z_0$ instead of $z_1$. Thus a protocol $\pi$ performs this task if it performs the task $[f,\eps]$, and additionally
\[
 \pi(x,y) \neq f(x,y) \Longrightarrow f(x,y) = z_1 \text{ and } \pi(x,y) = z_0.
\]

The \emph{information complexity}  of a communication task $T$ with respect to a measure $\mu$ is defined as
\[ \IC_\mu(T) = \inf_{\pi :\ \pi \text{\  performs\ } T} \IC_\mu(\pi). \]
It is essential here that we use infimum rather than minimum as there are tasks for which there is no protocol that achieves $\IC_\mu(T)$ while there is a sequence of protocols whose information cost converges to $\IC_\mu(T)$.  The \emph{external information complexity}  of a communication task $T$ is defined similarly. We will abbreviate $\IC_\mu(f,\eps)=\IC_\mu([f,\eps])$,  $\IC_\mu(f,\nu,\eps)=\IC_\mu([f,\nu,\eps])$, etc. 
It is important to note that when $\mu$ does not have full support, $\IC_\mu(f,\mu,0)$ can be strictly smaller than  $\IC_\mu(f,0)$.

\begin{remark}[A warning regarding our notation]
In the literature of information complexity it is common to use ``$\IC_\mu(f,\eps)$'' to denote the distributional error case, i.e. what we denote by $\IC_\mu(f,\mu,\eps)$. Unfortunately this has become the source of some confusions in the past, as sometimes  ``$\IC_\mu(f,\eps)$'' is used to denote both of the distributional error   and the point-wise error cases. To avoid ambiguity we distinguish the two cases by using the different notations $\IC_\mu(f,\mu,\eps)$ and $\IC_\mu(f,\eps)$.
\end{remark}

Similar to the fact that the maximal distributional communication complexity over all measures equals the public coin randomized communication complexity (see e.g., \cite[Section 3.4]{KS}), below we prove a lemma that establishes a similar relation between $\IC_\mu(f,\nu,\epsilon)$ and $\IC_\mu(f,\epsilon)$.

\begin{lemma}  \label{lem:sup-over-nu}
$\IC_\mu(f,\epsilon) = \max_{\nu} \IC_\mu(f,\nu,\epsilon)$ holds for all $\epsilon \ge 0$.
\end{lemma}

Note that the maximum exists due to continuity of $\IC_\mu(f,\nu,\epsilon)$ with respect to $\nu$, a fact that is discussed later in Section~\ref{sec:IC-continuity} (For  $\epsilon = 0$ one can take any full-support $\nu$).

\begin{proof}
We only need to show $\IC_\mu(f,\epsilon) \le \max_{\nu} \IC_\mu(f,\nu,\epsilon)$ as the other direction is obvious. The proof is an application of von Neumann's minimax theorem. 

Pick a small $\delta >0$, let $C_{\delta} = \{\pi: \IC_\mu(\pi) \le \IC_\mu(f,\epsilon) - \delta\}$. Although $C_{\delta}$ is an infinite set, we can approximate it by a  finite set by considering only  the protocols with bounded communication cost that use only a bounded number of unbiased random bits.  This process does not affect the validity of the proof, and hence the minimax theorem is still applicable.

Consider a two-player zero-sum game in which Alice chooses a protocol $\pi \in C_{\delta}$ and Bob chooses an input $(x,y) \in \cX \times \cY$, and define the utility for Alice to be $\Pr[\pi(x,y)=f(x,y)]$. Note that a mixed strategy for Alice is still just a protocol, and a mixed strategy for Bob corresponds to a probability measure on $\cX \times \cY$. By our definition of $C_{\delta}$ and the minimax theorem, we have
\[  \nonumber
\min_{\nu} \max_{\pi} \Ex_{(x,y)\sim \nu} \Pr[\pi(x,y)=f(x,y)]
= \max_{\pi} \min_{\nu} \Ex_{(x,y)\sim \nu} \Pr[\pi(x,y)=f(x,y)]
=  1-\epsilon - t(\delta) < 1 - \epsilon,
\]
where $t(\delta) > 0$ is a positive quantity. This means that there exists a measure $\nu_\delta^*$ such that for all $\pi \in C_{\delta}$, $\Ex_{(x,y)\sim \nu_\delta^*} \Pr[\pi(x,y) \neq f(x,y)] > \epsilon$. Letting $\delta \to 0$ gives $\max_{\nu} \IC_\mu(f,\nu,\epsilon) \ge \IC_\mu(f,\epsilon)$ as desired.
\end{proof}

Finally let us recall the two definitions of the prior-free notions of information complexity introduced in \cite{MR2961528}. The \emph{max-distributional information complexity} of a function $f\colon \cX \times \cY \to \cZ$ is defined as
\[ \ICD(f,\eps) = \max_{\mu} \IC_\mu(f,\mu,\eps). \]
The information complexity of  $f$ with error $\eps$ is defined as
\[ \IC(f,\eps)= \inf_{\pi} \max_{\mu} \IC_\mu(\pi), \]
where the infimum is over all protocols $\pi$ that perform the task $[f,\eps]$.
It is possible~\cite{MR2961528} to use a minimax argument and the concavity of $\IC_\mu(\pi)$ with respect to $\mu$ to show that
\[ \IC(f,\eps)= \inf_{\pi} \max_{\mu} \IC_\mu(\pi) = \max_{\mu} \inf_{\pi} \IC_\mu(\pi) = \max_\mu \IC_\mu(f,\epsilon)= \max_{\mu,\nu} \IC_\mu(f,\nu,\eps), \]
where the last equality  follows from  Lemma~\ref{lem:sup-over-nu}.

\subsubsection{The continuity of information complexity}   \label{sec:IC-continuity}
It is shown in~\cite[Lemma 4.4]{SelfRed} that for every communication task $T$, $\IC_\mu(T)$ is uniformly continuous with respect to $\mu$. More precisely, for every two measures $\mu_1$ and $\mu_2$ with $|\mu_1-\mu_2| \le \delta$ (the distance is in total variation distance), we have
\begin{equation}
\label{eq:continMu}
|\IC_{\mu_1}(T)-\IC_{\mu_2}(T)| \le 2 \log(|\cX \times \cY|) \delta + 2 h(2\delta).
\end{equation}

The information complexity functions $\IC_\mu(f,\epsilon)$ and $\IC_\mu(f,\nu,\epsilon)$ are both continuous with respect to $\epsilon$. The following simple lemma from~\cite{MR2961528} proves continuity for $\epsilon \in (0,1]$. The continuity at $0$ is more complicated and is proven in~\cite{MR3210776} (See also Theorem~\ref{thm:non-distri-error-lowerbd} and Theorem~\ref{thm:bd-IC-mu--distributional}  below).

\begin{lemma} \label{lem:continuity} \cite{MR2961528}
For every $f\colon\mathcal{X} \times \mathcal{Y} \to \mathcal{Z}$, $\epsilon_2 >\epsilon_1>0$ and measures $\mu,\nu$ on $\mathcal{X} \times \mathcal{Y}$, we have
\begin{equation}
\label{eq:continuity_Dist_Eps}
\IC_\mu(f,\nu,\epsilon_1)- \IC_\mu(f,\nu,\epsilon_2) \le (1-\epsilon_1/\epsilon_2) \log |\mathcal{X} \times \mathcal{Y}|,
\end{equation}
and
\begin{equation}
\label{eq:continuity_Pw_Eps}
\IC_\mu(f,\epsilon_1)- \IC_\mu(f,\epsilon_2)  \le (1-\epsilon_1/\epsilon_2) \log |\mathcal{X} \times \mathcal{Y}|.
\end{equation}
\end{lemma}
\begin{proof}
Consider a protocol $\pi$ with information cost $I$, and error $\epsilon_2>0$. Here we can consider the distributional error as in \eqref{eq:continuity_Dist_Eps} or the point-wise error as in~\eqref{eq:continuity_Pw_Eps}. Set $\delta= 1-\epsilon_1/\epsilon_2$, and let $\tau$ be the protocol that with probability $1-\delta$ runs $\pi$, and with probability $\delta$ Alice and Bob exchange their inputs and compute $f(x,y)$ correctly. The theorem follows as the new protocol has error at most $(1-\delta)\epsilon_2=\epsilon_1$, and information cost at most $I + \delta \log |\mathcal{X}  \times \mathcal{Y}|$.
\end{proof}

Note that  $\IC_\mu(f,\mu,0)$ is not always continuous with respect to $\mu$. For example, let the matrices
\begin{equation}
  \label{eq:measure-delta}
\mu_\epsilon =
\begin{pmatrix}
\frac{1-\epsilon}{3} & \frac{1-\epsilon}{3} \\
\frac{1-\epsilon}{3} & \epsilon
\end{pmatrix},  \qquad
\mu = \lim_{\epsilon \to 0} \mu_\epsilon = \begin{pmatrix}
\frac{1}{3} & \frac{1}{3} \\
\frac{1}{3} & 0
\end{pmatrix}.
\end{equation}
represent distributions on $\{0,1\}^2$. Here the entry at the $i$-th row and $j$-th column corresponds to the measure of the point $(i-1,j-1) \in \{0,1\}^2$.
Now for the $2$-bit $\AND$ function, we have $\IC_{\mu}(\AND,\mu,0)=0$, while $\IC_{\mu_\epsilon}(\AND,\mu_\epsilon,0)=\IC_{\mu_\epsilon}(\AND,0)$ as $\mu_\epsilon$ has full support. Thus
\[ \lim_{\epsilon \to 0} \IC_{\mu_\epsilon}(\AND, \mu_\epsilon,0) =\lim_{\epsilon \to 0} \IC_{\mu_\epsilon}(\AND,0) = \IC_{\mu}(\AND,0), \]
which is known to be bounded away from $0$.

Finally, note that Lemma~\ref{lem:continuity} also implies the continuity of  $\IC_{\mu}(f,\nu,\epsilon)$ with respect to $\nu$ when $\epsilon>0$.  Indeed if $|\nu_1-\nu_2| \le \delta \le \epsilon$, then a protocol that has distributional error $\eps$ with respect to $\nu_2$, will have error at most $\eps+\delta$ and  at least $\eps-\delta$ with respect to $\nu_1$. Thus
\begin{equation}
\label{eq:continNu}
\IC_{\mu}(f,\nu_1,\epsilon+\delta) \le \IC_{\mu}(f,\nu_2,\epsilon) \le  \IC_{\mu}(f,\nu_1,\epsilon-\delta).
\end{equation}
which establishes the desired continuity. A similar example to \eqref{eq:measure-delta} shows that $\IC_{\mu}(f,\nu,0)$ is not necessarily continuous with respect to $\nu$.

\subsubsection{Communication protocols as random walks on $\Delta(\cX \times \cY)$}   \label{sec:randomwalk}

Recall that $\Delta(\cX \times \cY)$ denotes the set of probability distributions on $\cX \times \cY$.  Consider a protocol $\pi$ and a prior distribution $\mu$ on the set of inputs $\cX \times \cY$. Suppose that in the first round Alice  sends a random signal $B$ to Bob. We can interpret this as  a random update of  the prior distribution $\mu$ to a new distribution $\mu_0 = \mu|_{B=0}$ or $\mu_1 = \mu|_{B=1}$ depending on the value of $B$. It is not difficult to see that $\mu_b(x,y) =p_b(x) \mu(x,y)$ for $b=0,1$, where $p_b(x)=\frac{\Pr[B=b|x]}{\Pr[B=b]}$. In other words, $\mu_b$ is obtained by multiplying the rows of $\mu$ by non-negative numbers. 	From the law of total expectation,
	\begin{equation} \label{eq:no-drift}
		\mu = \Ex_B [ \mu | B] = \Pr[B=0] \mu_0 + \Pr[B=1] \mu_1.
	\end{equation}

Similarly if Bob is sending a message, then  $\mu_b$ is obtained by multiplying the columns of $\mu$ by the numbers $p_b(y)=\frac{\Pr[B=b|y]}{\Pr[B=b]}$. That is  $\mu_b(x,y) = \mu(x,y)p_b(y)$.

The opposite direction is also true: given a distribution $\mu$, distributions $\mu_0$, $\mu_1$, and $0 \leq p_0, p_1 \leq 1$ such that
\begin{itemize}
	\item
		$p_0 + p_1 = 1$,
	\item
		$\mu_0$ and $\mu_1$ are obtained from $\mu$ by scaling its rows,
	\item
		$\mu = p_0 \mu_0 + p_1 \mu_1$,
\end{itemize}
one can define a random bit $B$ that can be sent by Alice such that $\mu_b$ is $\mu$ conditioned on $B=b$ for $b\in \{0,1\}$, and $p_b = \Pr[B=b]$.
A similar statement holds for the case where $\mu_0$ and $\mu_1$ are obtained from $\mu$ by scaling its columns and $B$ is a signal that will be sent by Bob.

Therefore, we can think of a protocol as a random walk on $\Delta(\cX \times \cY)$ that starts at $\mu$, and every time that a player sends a message, it moves to a new distribution.  Equation~\eqref{eq:no-drift} implies that this random walk is without drift.

Let $\Pi$ denote the transcript of the protocol.  Note that when the protocol terminates, the random walk stops at $\mu_\Pi \defeq \mu|_\Pi$. Since $\Pi$ itself is a random variable, $\mu_\Pi$ is a random variable that takes values in  $\Delta(\cX \times \cY)$.  Interestingly, both the internal and external information costs of the protocol depend only on the distribution of $\mu_\Pi$ (this is a distribution on the set $\Delta(\cX \times \cY)$, which itself is a set of distributions)~\cite{MarkComputable}. It does not matter how different the steps of two protocol are, and as long as they both  yield the same distribution on $\Delta(\cX \times \cY)$, they have the same internal and external information cost. Consequently, one can directly work with this random walk, instead of working with the actual  protocols.

In order to study the relation between the information complexity and the distribution of $\mu_\Pi$, define the \emph{concealed information} and \emph{external concealed information} of a protocol $\pi$ with respect to $\mu$, respectively, as
\begin{equation} \label{eq:def-CI}
	\CI_\mu(\pi) = H(X|\Pi Y)+H(Y|\Pi X) = H(X|Y)+H(Y|X)-\IC_\mu(\pi),
\end{equation}
and
\[ \CI^\ext_\mu(\pi) = H(XY|\Pi) = H(XY)-\IC^\ext_\mu(\pi). \]

With this definition it is easy to see that the information cost of a protocol $\pi$ with transcript $\Pi$ only depends on the distribution of $\mu_\Pi$. Indeed
\[
	\CI_\mu(\pi)
	= H_{XY \sim \mu}(X|\Pi Y) + H_{XY \sim \mu}(Y | \Pi X)
	= \Ex_{\Pi} H_{XY \sim \mu_\Pi}(X | Y) + \Ex_\Pi H_{XY \sim \mu_\Pi}(Y | X).
\]
Another nice property of concealed information is that if $\pi_0$ and $\pi_1$ are the two branches of the protocol $\pi$ corresponding respectively to $B=0$ and $B=1$ where $B$ is the first bit sent, then
\[
	\CI_\mu(\pi)
	= \Pr[B=0] \CI_{\mu|B=0}(\pi_0) + \Pr[B=1] \CI_{\mu|B=1}(\pi_1).
\]
Thus, the expected value of $\CI$ is preserved throughout the execution of the protocol. Similar results hold for $\CI^\ext_\mu(\pi)$.

%

\section{Main Results}\label{sec:main_results}
In this section, we state and discuss our main results in full detail. Simpler proofs are presented in this section, but the proofs of the more involved results  are postponed to later sections.

We will use the following simple estimate:
\begin{equation} \label{eq:h-estimate}
 x  \in [0,1/2] \Longrightarrow x \log \frac{1}{x} \leq h(x) \leq 2x \log \frac{1}{x},
\end{equation}
which holds since in that range $-x \log x \geq -(1-x) \log (1-x)$. 

Denote
\begin{equation}\label{eq:h-convex}
\hc(x) = h(\min(x,1/2)).
\end{equation}
It satisfies $\hc(x) \geq h(x)$ and $x \le  \hc(x)$.
It is easy to see that $h$ is concave. Therefore, $\hc$ is also concave as a minimum of two concave functions. 
Additionally, $h(0) = \hc(0) = 0$. We will next show how to utilize these two properties of $h$ and $\hc$:
for any concave function $g\colon \mathbb{R}^+ \to \mathbb{R}$ for which $g(0)=0$, and for any $x>0$ and $0<q<1$, it holds
\begin{equation} \label{eq:conv-factor}
	g(qx) 
	\geq q g(x) + (1-q) g(0)
	= q g(x).
\end{equation}
This implies the subadditivity of $g$: for all $a_1,a_2>0$, $g(a_1+a_2) \leq g(a_1) + g(a_2)$, as
$g(a_i) \geq \frac{a_i}{a_1+a_2} g(a_1+a_2)$, for all $i=1,2$.

\subsection{Information complexity with point-wise error} \label{sec:pointwiseError}

Consider a communication problem $f\colon \cX \times \cY \to \cZ$,  and a distribution $\mu$. How close can $\IC_\mu(f,\epsilon)$ be to $\IC_\mu(f,0)$? A simple argument shows that $\IC_\mu(f,\epsilon) \leq \IC_\mu(f,0) - \Omega(\epsilon)$.

\begin{proposition} \label{prop:upper-bd-IC-mu-eps-trivial}
 Let $f\colon \cX \times \cY \to \cZ$, and let $\mu$ be a measure on $\cX \times \cY$. Denoting $c=\IC_\mu(f,0)$, we have
\[
 \IC_\mu(f,\epsilon) \leq (1-\epsilon) \IC_\mu(f,0)=\IC_\mu(f,0) -  c \eps.
\]
\end{proposition}
\begin{proof}
 Let $\pi$ be a zero-error protocol for $f$. Consider a protocol $\pi'$ in which Alice and Bob use their public randomness to run with probability $1-\epsilon$  the protocol $\pi$, or to terminate with an arbitrary output with probability $\epsilon$.  Let $\Pi$ and $\Pi'$  be  respectively the transcripts of $\pi$ and $\pi'$ on the random input $(X,Y)$. We have
\[
 I(X;\Pi'|Y) = H(X|Y) - H(X|\Pi' Y) = H(X|Y) -\epsilon H(X|Y) - (1-\epsilon) H(X | \Pi Y) = (1-\epsilon) I(X;\Pi | Y).
\]
The same holds for $I(Y;\Pi'|X)$, and the statement follows.
\end{proof}

Our first major theorem shows that this trivial bound can be improved to $\IC_\mu(f,\epsilon) \leq \IC_\mu(f,0) - \Omega(\entf(\epsilon))$.

\begin{theorem} \label{thm:upper-bd-IC-mu-eps}
Consider a function $f\colon \cX \times \cY \to \cZ$ and a probability measure $\mu$ on $\cX \times \cY$ such that $\IC_\mu(f,0) > 0$. There exist positive constants $\tau,\epsilon_0$, depending on $f$ and $\mu$, such that for every $\epsilon \leq \epsilon_0$,
\[
 \IC_\mu(f,\epsilon) \le \IC_\mu(f,0) - \tau \entf(\epsilon).
\]
Moreover:
\begin{description}
 \item[Non-constant case:] Suppose that $f(a) \neq f(b)$ for two points $a,b$ in the support of $\mu$, and on the same row or column. Then one can take $\tau \geq \mu(a)^2 \mu(b)/64$, and $\epsilon_0$ depends only on $\min(\mu(a),\mu(b))$ and $|\cX \times \cY|$.
 \item[AND case:] Let $x_0,x_1 \in \cX$ and $y_0,y_1 \in \cY$. Suppose that $f(x_0y_0) = f(x_0y_1) = f(x_1y_0) = z_0$ and $f(x_1y_1) = z_1 \neq z_0$, and that $x_0y_0,x_0y_1,x_1y_0 \in \supp\mu$. Then one can take $\tau \geq \frac{\mu(x_0y_0)^2}{128} \min(\mu(x_0y_1),\mu(x_1y_0))$, and $\epsilon_0$ depends only on $|\cX \times \cY|$ and the minimum of $\mu(x_0y_0)$, $\mu(x_0y_1)$, $\mu(x_1y_0)$.
\end{description}
\end{theorem}
\begin{proof}
See Section~\ref{sec:proof:upper-bd-IC-mu-eps}.
\end{proof}

\begin{remark}
\label{rem:upper-bd-IC-mu-eps}
We prove  Theorem~\ref{thm:upper-bd-IC-mu-eps} by taking a zero-error protocol for $f$, and turning it into an  $\epsilon$-error protocol that has an $\Omega(\entf(\epsilon))$ gain in the information cost over the original protocol. The high-level idea is that one of the players checks  her/his input and if it is equal to a certain value $x_1$, then with probability $\eps$ changes to a different value $x_0$. This obviously creates an error of at most $\eps$. In the Non-constant case of Theorem~\ref{thm:upper-bd-IC-mu-eps}, the points $a$ and $b$ are used to determine $x_0$ and $x_1$, and in the AND case, the same $x_0$ and $x_1$ as they are described in the statement of the theorem can be used. Note that this modification can only create errors that erroneously output $f(x_0,y)$ instead of $f(x_1,y)$ for some values of $y$. This allows us to obtain a one-sided error for many functions. We shall use this later in Corollary~\ref{cor:ANDzero-gap} to obtain an upper bound on the information complexity of the AND function when only one-sided error is allowed.
\end{remark}

Despite the simplicity of the idea described in Remark~\ref{rem:upper-bd-IC-mu-eps}, the proof is rather involved, and uses some of our other results such as characterization of internal-trivial measures.  The heart of the proof is of course showing the existence of  appropriate values of $x_0$ and $x_1$ that can lead to the desired  gain of $\Omega(\entf(\epsilon))$.  

Let $\XOR$ denote the $2$-bit $\XOR$ function. The next result shows that the analogue of Theorem~\ref{thm:upper-bd-IC-mu-eps} does not hold for the external information complexity.

\begin{proposition}
\label{prop:ext_XOR}
Let $\mu$ be the distribution defined as
\[
 \mu = \begin{array}{|c|c|} \hline 1/2 & 0 \\\hline 0 & 1/2 \\\hline \end{array} \, .
\]
Then $\IC_\mu^\ext(\XOR,\epsilon) \geq \IC_\mu^\ext(\XOR,0) - 3\epsilon$.
\end{proposition}
\begin{proof}
See Section~\ref{sec:proof:ext_XOR}.
\end{proof}

For the lower bound we prove the following theorem.
\begin{theorem}
\label{thm:non-distri-error-lowerbd}
For all $f, \mu, \epsilon$, we have
\[
\IC_\mu(f,\epsilon) \ge \IC_\mu(f,0) - 4|\cX| |\cY| \hc(\sqrt{\epsilon}).
\]
\end{theorem}
\begin{proof}
See Section~\ref{proof:non-distri-error-lowerbd}.
\end{proof}

Theorem~\ref{thm:non-distri-error-lowerbd} is obtained by taking an $\eps$-error protocol and completing it to a zero-error protocol. Here Alice and Bob first run the protocol that performs $[f,\epsilon]$,  but when this protocol terminates, instead of returning the output, they continue their interaction to verify that the value that they have obtained is correct. We will be able to show that these additional interactions can be performed at a small information cost, and thus the total information complexity of the new protocol is not going to be much larger than that of the original protocol. This method, that we call \emph{protocol completion}, is used in the proofs of other results such as Theorem~\ref{thm:AND-gap} as well.

Finally let us remark that we do not know whether the bound in Theorem~\ref{thm:non-distri-error-lowerbd} is tight. In fact we are not aware of any examples of $f$ and $\mu$ that refutes the possibility that $\IC_\mu(f,\epsilon) = \IC_\mu(f,0) - \Theta(h(\epsilon))$ for every $f$ and $\mu$ satisfying $\IC_\mu(f,0)>0$.

\subsection{Information complexity with distributional error}   \label{sec:distriError}
In Section~\ref{sec:pointwiseError} we considered the amount of gain one can obtain by allowing point-wise error. Next we turn  to distributional error. How much can one gain in information cost by allowing a distributional error of $\epsilon$? Small modifications in the proofs of Theorem~\ref{thm:upper-bd-IC-mu-eps} and Theorem~\ref{thm:non-distri-error-lowerbd} imply the following bounds.

\begin{theorem}
\label{thm:bd-IC-mu--distributional}
Let $\mu$ be a probability measure on $\cX \times \cY$, and let $f\colon \cX \times \cY \to \cZ$ satisfy $\IC_\mu(f,\mu,0)>0$.  We have 
\[
 \IC_\mu(f,\mu,0) - 4 |\cX| |\cY| \hc(\sqrt{\epsilon/\alpha}) \le \IC_\mu(f,\mu,\epsilon) \le  \IC_\mu(f,\mu,0) - \frac{\alpha^2}{4} \entf\left(\epsilon \alpha/4 \right) + 3\epsilon \log |\cX \times \cY|,
\]
where $\alpha = \min_{xy \in \supp \mu} \mu(x,y)$.
\end{theorem}
\begin{proof}
See Section~\ref{sec:proof:distributional-error}.
\end{proof}

It is also possible to prove the upper bound of  Theorem~\ref{thm:bd-IC-mu--distributional}  using a different approach by ``truncating'' a zero-error protocol. Unfortunately this approach requires some assumptions on the support of $\mu$. Nevertheless we sketch this proof, as the idea seems to be new, and it might have other applications.

Let $\Delta_0 \subseteq \Delta(\cX \times \cY)$ be the set of all measures $\nu$ such that $\IC_\nu(f,\nu,\eps)=0$.  Consider a protocol $\pi$ that performs $[f,\mu,0]$. First we simulate $\pi$ with another protocol $\pi'$ such that no signal of $\pi'$ jumps from outside of $\Delta_0$ to the interior of $\Delta_0$. In other words if  some partial transcript $t$ satisfies $\mu_t \not\in \Delta_0$, then when the next signal $B$ is sent,  $\mu_{tB}$ is either still outside of $\Delta_0$ or it is on the boundary $\partial \Delta_0$. The simulation can be done in a prefect manner so that if $\Pi$ and $\Pi'$ denote, respectively, the transcripts of $\pi$ and $\pi'$, then $\mu_{\Pi'}$ has the same distribution as $\mu_\Pi$. The new protocol $\pi'$ might  not necessarily have bounded communication, but it will terminate with probability $1$. We refer the reader to \cite[Signal Simulation Lemma]{multipartyAND} and \cite[Claim 7.14]{MR3210776} for more details on such simulations. 

We will truncate $\pi'$ in the following manner to obtain a new protocol $\pi_0$ that performs $[f,\mu,\eps]$. Whenever the corresponding random walk of $\pi'$ reaches a distribution $\nu$ that is on the boundary $\partial\Delta_0$, the two players  stop the random walk, and use $\IC_\nu(f,\nu,\eps)=0$ to output a value that creates a distributional error of at most $\eps$ with respect to $\nu$ at no information cost. Obviously the distributional error of the protocol $\pi_0$ is at most $\eps$. To analyze its information cost, denote the transcript of $\pi_0$ by $P$, and note that $P$ is a partial transcript for $\pi'$. Let $\pi'_P$ be the continuation of $\pi'$ when one starts at this partial transcript. It is not difficult to see that
$$\IC_\mu(\pi) = \IC_\mu(\pi') = \IC_\mu(\pi_0) + \Ex_{P} [\IC_{\mu_P}(\pi'_P)].$$
Since $\pi'$ performs $[f,\mu,0]$, the tail protocol $\pi_P$ must perform $[f,\mu_P,0]$. Hence in order to finish the proof, it suffices to show that $\IC_{\nu}(f,\nu,0)=\Omega(h(\eps))$ for every $\nu \in \partial \Delta_0$, as this would imply the desired $\IC_\mu(\pi) \ge  \IC_\mu(\pi_0) + \Omega(h(\eps))$. This can be proven with some work when $\mu$ is of full support, however it is not true for general measures. For example, consider the AND function, and let $\mu$ be the distribution on $\{0,1\}^2$ defined as $\mu(0,0)=1-2\eps$ and  $\mu(1,0)=\mu(1,1)=\eps$. Note that although $\mu$ is on the boundary of $\Delta_0$, we have $\IC_\mu(\AND,\mu,0) \le 2\eps$. Indeed, since $\mu(0,1)=0$,  Bob with probability $1$ knows the correct output by looking at his own input $Y$, and so if he sends his bit to Alice,  they will both know the correct output.  This will have information cost at most  $H(Y|X) = \Pr[X=1] H(Y|X=1) =  2\epsilon$.

\subsection{Information complexity of the AND function with error}\label{sec:ANDfunction}

Building upon the previous works of Ma and Ishwar~\cite{MaIshwar2011,MaIshwar2013}, Braverman et al. \cite{MR3210776} developed a method for proving the optimality of information complexity and applied it to determine the internal and external information complexity of the two-bit AND function. They introduced a ``continuous-time'' protocol for this task, and proved that it has optimal internal and external information cost for any underlying distribution. Although this protocol is not a conventional communication protocol as it has access to a continuous clock, it can be approximated by conventional  communication protocols through dividing the time into
finitely many discrete units.   Then in~\cite[Problem 1.1]{MR3210776} they considered the case where error is allowed, and conjectured a gain of $\IC(\AND) - \IC(\AND,\eps)=\Theta(h(\eps))$.   In this section, we conduct a thorough analysis of the information complexity of the AND function when error is permitted, and  among other results, prove the aforementioned conjecture.

Applying our general bounds from in Section~\ref{sec:pointwiseError} and Section~\ref{sec:distriError} (i.e. Theorems~\ref{thm:upper-bd-IC-mu-eps},~\ref{thm:non-distri-error-lowerbd},~and~\ref{thm:bd-IC-mu--distributional}) we already obtain that for small enough $\epsilon \ge 0$,
\begin{itemize}
\item[(i).] For every distribution $\mu$ satisfying $\IC_\mu(\AND,0) > 0$, we have
\[
\IC_\mu(\AND,0) - O_\mu(h(\sqrt{\epsilon})) \le \IC_\mu(\AND,\epsilon) \le \IC_\mu(\AND,0) - \Omega_\mu(h(\epsilon));
\]
\item[(ii).] For every distribution $\mu$  satisfying $\IC_\mu(\AND, \mu, 0) > 0$, we have
\[
\IC_\mu(\AND,\mu,0) - O_\mu(h(\sqrt{\epsilon})) \le \IC_\mu(\AND,\mu,\epsilon) \le \IC_\mu(\AND, \mu, 0) - \Omega_\mu(h(\epsilon)).
\]
\end{itemize}

We show that under some conditions on the support of $\mu$, the above lower bounds can be improved to match the upper bounds. 

\begin{theorem}   \label{thm:AND-gap}
For small enough $\epsilon \ge 0$, the following hold,
\begin{itemize}
\item[(i).] For every distribution $\mu$ which is full support, except perhaps for $\mu(1,1)$, we have
\[
\IC_\mu(\AND,\epsilon) = \IC_\mu(\AND,0) -  \Theta(\hc(\epsilon)),
\]
where the hidden constants can be fixed if $\omu(0,0), \omu(0,1), \omu(1,0)$ are bounded away from  $0$.
\item[(ii).] In particular for every distribution $\mu$ of full support, we have
\[
\IC_\mu(\AND,\mu,\epsilon) = \IC_\mu(\AND, \mu, 0) - \Theta(\hc(\epsilon)).
\]
\end{itemize}
\end{theorem}
Note that for every distribution $\mu$ of full support, we have $\IC_\mu(\AND, \mu, 0)=\IC_\mu(\AND,0)>0$,  and $\IC_\mu(\AND,\epsilon/\alpha) \le \IC_\mu(\AND,\mu,\epsilon) \le \IC_\mu(\AND,\epsilon)$ where $\alpha=\min_{xy} \mu(xy)$. Thus Theorem~\ref{thm:AND-gap}~(ii) follows from (i).

From a technical point of view,  Theorem~\ref{thm:AND-gap} is perhaps our most  involved result in this article, and its proof occupies the bulk of Section~\ref{sec:AND}. The first idea that  facilitates the proof substantially is  developed by the first two authors in~\cite{DaganFilmus}. They showed that it is possible to parametrize the space of the distributions $\Delta(\cX \times \cY)$ so that the changes that occur in the prior distribution by the  players' interactions can be captured by product measures. This idea, that is discussed in details in Section~\ref{sec:parametrization}, allows us to first prove the lower bound of Theorem~\ref{thm:AND-gap} for the product measures, and then add minor adjustments to adopt it for non-product distributions. The second component of the proof is a stability result. Recall from Section~\ref{sec:randomwalk} that the information cost of every protocol $\pi$ depends only on its ``leaf distribution'', i.e. the distribution of $\mu_\Pi$, where $\Pi$ is the transcript of $\pi$ or equivalently $\mu_{\leafl}$ where $\leafl$ is a random leaf of the protocol tree. Our stability result, Theorem~\ref{thm:stability}, shows that the leaf distribution of  every almost optimal protocol $\pi$ for $[\AND,0]$  shares certain similarities with that of the buzzer protocol. Note that since $\pi$ does not make any errors, by the end of the protocol, either both players know that the input is $(1,1)$, or one of them has revealed that her input is $0$. Theorem~\ref{thm:stability} formalizes the intuition that in this latter case, the other player  must not have revealed that his input is very likely to be $0$. This is achieved through defining a potential function that depends only on the distribution of $\mu_\Pi$ and proving that it is bounded by the so called information wastage $\IC_\mu(\pi)-\IC_\mu(\AND,0)$. With these results in hand, in order to complete the lower bound of Theorem~\ref{thm:AND-gap}, we start with a protocol $\pi$ performing $[\AND,\epsilon]$ with almost optimal information complexity. First we show that $\pi$ can be completed to a protocol that performs $[\AND,0]$ at a small additional information cost, though possibly larger than the desired $O(h(\eps))$. Then we apply the stability result to deduce certain properties for the leaf distribution of $\pi$. This will imply that one indeed needs only an additional cost of $O(h(\eps))$ to extend $\pi$ to a protocol that solves $[\AND,0]$.

Braverman et al.~\cite{MR3210776} showed that $\IC(\AND,0) = \max_\mu \IC_\mu(\AND,0)$ is attained on a distribution having full support.  This enables us to derive the following corollary on prior-free information complexity.

\begin{corollary} \label{cor:AND-gap}
When $\epsilon \ge 0$ is sufficiently small, we have
\begin{itemize}
\item[(i).] $\IC(\AND,\epsilon) = \IC(\AND,0) - \Theta(\hc(\epsilon))$;
\item[(ii).] $\ICD(\AND,\epsilon) = \IC(\AND,0) - \Theta(\hc(\epsilon))$;
\end{itemize}
\end{corollary}

\begin{proof}
The measure $\mu$ that maximizes $\IC_\mu(\AND,0)$ has full support~\cite{MR3210776}, and thus $\IC(\AND,0) = \IC_\mu(\AND,0) = \IC_\mu(\AND,\mu,0)$. By Theorem~\ref{thm:AND-gap}~(ii),
\[
\IC(\AND,\epsilon)  \ge \ICD(\AND,\epsilon) \ge \IC_\mu(\AND,\mu,\epsilon) \ge \IC_\mu(\AND,\mu,0) - O(\hc(\epsilon)) = \IC(\AND,0) - O(\hc(\epsilon)).
\]
Moreover by a general upper bound that we prove later in  Theorem~\ref{thm:non-distributional-ub}, we have
\[
\ICD(\AND,\epsilon) \le \IC(\AND,\epsilon) \le \IC(\AND,0) - \Omega(\hc(\epsilon)).
\]
Both items in the corollary follow. 
%
\end{proof}

Since the difficult distributions for the set disjointness function are the ones in which the inputs typically have small or no intersections at all, the distributions for the AND function that assign a very small or $0$ mass to the point $(1,1)$ are of particular importance. Let 
	\[ \IC^\delta(\AND, \epsilon, 1 \to 0) = \sup_{\mu \colon \mu(1,1) \leq \delta} \IC_\mu(\AND,\epsilon,1\to 0). \]
The following corollary is used in Section~\ref{sec:DISJ-result} to analyze the information complexity of the set disjointness problem. 
	
\begin{corollary}
\label{cor:ANDzero-gap}
When $\epsilon \ge 0$ is sufficiently small, we have
\begin{itemize}
\item[(i).] $\IC^0(\AND,\epsilon) = \IC^0(\AND,0) - \Theta(\hc(\epsilon))$;
\item[(ii).] $\IC^0(\AND,\epsilon, 1 \to 0) = \IC^0(\AND,0) - \Theta(\hc(\epsilon))$.
\item[(iii).] There exist universal constants $C_1$ and $C_2$ such that for every $\eps,\delta>0$, 
	\[
	\IC^\delta(\AND, \epsilon, 1 \to 0) \leq \IC^0(\AND,0) - C_1 \hc(\epsilon) + C_2 \hc(\delta).
	\]
\end{itemize}	
\end{corollary}
\begin{proof}
	Let $\mu$ be the distribution maximizing $\IC_\mu(\AND,0)$ under the constraint $\mu(1,1) = 0$; This measure, which is described in~\cite{MR3210776}, has full support except for $\mu(1,1) = 0$. Thus by Theorem~\ref{thm:AND-gap}~(i), 
	\[
	\IC^0(\AND,\epsilon)
	\geq\IC_\mu(\AND,\epsilon)
	\geq \IC_\mu(\AND,0) - O(\hc(\epsilon))
	= \IC^0(\AND,0) - O(h(\epsilon)). 
	\]	
 	Consequently, since $\IC^0(\AND,\epsilon) \le \IC^0(\AND,\epsilon, 1 \to 0)$, both (i) and (ii) will follow if we prove   $\IC^0(\AND,\epsilon, 1 \to 0) \le \IC^0(\AND,0) - \Omega(\hc(\epsilon))$.  To prove this, we would like to apply the AND case of  Theorem~\ref{thm:upper-bd-IC-mu-eps}, however to be able to obtain a uniform upper bound on  $\IC^0(\AND,\epsilon, 1 \to 0)$,  we need to have a uniform lower bound on the probabilities  $\mu(0,0),\mu(0,1),\mu(1,0)$.  Let $\alpha > 0$ to be determined later, and consider any distribution $\mu$ with $\mu(1,1)=0$ and $\mu(a)<\alpha$ for some input $a \neq (1,1)$. Pick $b \in \{0,1\}^2 \setminus \{a,(1,1)\}$, and obtain the distribution $\mu'$ from $\mu$ by transferring all the probability mass on $a$ to $b$. That is $\mu'(b) = \mu(a)+\mu(b)$ and $\mu'(a)=0$, and otherwise $\mu$ and $\mu'$ are identical. Obviously $| \mu - \mu' | = \alpha$.  Now  \eqref{eq:continMu} and \eqref{eq:conv-factor} imply
 	\begin{equation}
 	\label{eq:contin_applied}
 	\IC_\mu(\AND,\epsilon, 1 \to 0)  \le \IC_\mu(\AND,0)
 	\leq \IC_{\mu'}(\AND,0) + 4 \alpha + 2h(2\alpha) =4 \alpha + 2h(2\alpha) 
 	\leq  4h(2\alpha),
 	\end{equation}
 	where we used the fact that $\IC_{\mu'}(\AND,0)=0$ as $\supp\mu'$ contains only two points. Setting $\alpha=0.001$ for example yields $\IC_\mu(\AND,0) \le 4 h(2\alpha) < 0.1 < \IC^0(\AND,0) \approx 0.4827$. It remains to prove the statement for  the distributions $\mu$ with $\mu(0,0),\mu(0,1),\mu(1,0) \geq \alpha$.  In this case Theorem~\ref{thm:upper-bd-IC-mu-eps} (See Remark~\ref{rem:upper-bd-IC-mu-eps} regarding the one-sidedness) implies that  exists a constant $C>0$ such that   $\IC_\mu(\AND,\epsilon, 1 \to 0) \le \IC^0(\AND,0) - C \hc(\epsilon)$. This finishes the proof (i) and (ii). 
	
        To prove (iii), consider an arbitrary distribution $\mu$ with $\mu(1,1) \leq \delta$, and let $\mu'$ be the distribution that is obtained from $\mu$ by moving the probability mass on $(1,1)$ to a different point so that  $\mu'(1,1) = 0$ and $| \mu - \mu'| = \delta$. Similar to \eqref{eq:contin_applied}, we obtain
        $$\IC_\mu(\AND, \epsilon, 1 \to 0) \le \IC_{\mu'}(\AND, \epsilon, 1 \to 0)+ 4h(2\delta) \le \IC^0(\AND, \epsilon, 1 \to 0)+ 4h(2\delta),$$
and thus (iii) follows from (ii).
\end{proof}

\subsection{Set disjointness function with error}  \label{sec:DISJ-result}
In this section we focus on the set disjointness function. Firstly it is not hard to obtain the following result.

\begin{corollary} \label{cor:disjointness_IC_ND}
For $\epsilon \ge 0$ small enough,
\[
 \IC(\DISJ_n,\epsilon) \geq n [\ICNDz(\AND,0) - \Theta(h(\epsilon))],
\]
where the hidden constant is independent of $n$. 
\end{corollary}

\begin{proof}
By the argument that proves the additivity of information complexity (see e.g. \cite{MR3265014}), one can prove that $\IC(\DISJ_n,\epsilon) \ge n \IC^0(\AND,\epsilon)$. Then apply Corollary~\ref{cor:AND-gap}. The essential idea is the following. Consider a distribution $\mu$ on $\{0,1\}^2$ with $\mu(1,1)=0$, and let $(a,b) \in \{0,1\}^2$ be an input for the AND function. Let $XY \in \{0,1\}^n \times \{0,1\}^n$ be such that for some randomly selected $J \in \{1,\ldots,n\}$ we have $(X_j,Y_j)=(a,b)$, and for $i \in \{1,\ldots,n\} \setminus \{J\}$, the pairs $(X_i,Y_i)$ are i.i.d. random variables, each with distribution $\mu$. Since $\mu(1,1)=0$,  we have  $\DISJ_n(X,Y)=1-\AND(a,b)$ with probability $1$. Thus one can take a protocol  $\pi$ for $\DISJ_n$ and use it to solve $\AND(a,b)$ correctly for every $(a,b)$. By  sampling $XY$ in a clever way, using both public and private randomness, one can guarantee that the information cost of the new protocol that solves $\AND(a,b)$ will be the information cost of $\pi$ divided by $n$. 
\end{proof}

As a result one also obtains that $R_\epsilon(\DISJ_n) \ge n [\ICNDz(\AND,0) - \Theta(h(\epsilon))]$.  It turns out that by using techniques from~\cite{MR3210776} and~\cite{MR2961528}, one can prove the following theorem.
\begin{theorem}
\label{thm:set_disj_cc}
For the set disjointness function $\DISJ_n$ on inputs of length $n$, we have
\[
 R_\epsilon(\DISJ_n) = n[\IC^0(\AND,0) - \Theta(h(\epsilon))].
\]
\end{theorem}
\begin{proof}
See Section~\ref{sec:proof:set_disj_cc}.
\end{proof}

We conjecture that in fact the exact constant is given by $\IC^0(\AND,\epsilon,1\to 0)$. In other words:
\begin{conjecture}
\label{conj:SetDisjointExact}
For the set disjointness function $\DISJ_n$ on inputs of length $n$, we have
\[
 R_\epsilon(\DISJ_n) = n \IC^0(\AND,\epsilon,1\to 0) + o(n).
\]
\end{conjecture}
Braverman~\cite{MR2961528} proved that for all $0< \alpha < 1$ and for all functions $f$,
\[
 \ICD(f,\epsilon) \geq (1-\alpha) \ICND(f, \frac{\epsilon}{\alpha}).
\]
When $f = \DISJ_n$, Corollary~\ref{cor:disjointness_IC_ND} gives 
\[
 \frac{\ICD(\DISJ_n,\epsilon)}{n} \geq (1-\alpha)(\IC^0(\AND,0) - \Theta(h(\epsilon/\alpha))) \geq \IC^0(\AND,0) - \Theta(\alpha + h(\epsilon/\alpha)).
\]
Substituting $\alpha=\sqrt{\epsilon \log (1/\epsilon)}$ yields 
\begin{equation}
\label{eq:disj_distrib_lowerbnd}
 \frac{\ICD(\DISJ_n,\epsilon)}{n} \geq \IC^0(\AND,0) - \Theta(\sqrt{h(\epsilon)}).
\end{equation}
In Theorem~\ref{thm:setDisj_distrib} below, which is one of our main contributions, we show that this bound is sharp. The proof relies on introducing a new protocol for set disjointness problem, and analyzing its information cost. 
\begin{theorem}
\label{thm:setDisj_distrib}
For the set disjointness function $\DISJ_n$ on inputs of length $n$, we have
\[
\ICD(\DISJ_n,\epsilon) =n[\IC^0(\AND,0) - \Theta(\sqrt{h(\epsilon)})] + O(\log n).
\]
\end{theorem}
\begin{proof}
See Section~\ref{sec:proof:setDisj_distrib}.  
\end{proof}

\subsection{Prior-free Information Cost} \label{sec:statements:Prior-free}
Theorem~\ref{thm:setDisj_distrib} shows that for $\alpha=\sqrt{\epsilon \log (1/\epsilon)}=\Theta(\sqrt{h(\epsilon)})$, and sufficiently large $n$, we have
\[
\frac{\ICD(\DISJ_n,\eps)}{1-\Theta(\alpha)}=\IC(\DISJ_n,\eps/\alpha) < \IC(\DISJ_n,\eps),
\]
and thus proves the separation between distributional and non-distributional prior-free information complexity as it was promised in \eqref{eq:priorfree_separation_thightness}. As we discussed in the introduction this has the important implication that amortized randomized communication complexity is not necessarily equal to the amortized distributional communication complexity with respect to the hardest distribution. More precisely, there are examples for which  $\max_\mu D_\epsilon^{\mu,n}(f^n) \neq R_\epsilon^n(f^n)$.

Next we turn to proving general lower bounds and upper bounds for the prior-free information complexity. Theorem~\ref{thm:non-distri-error-lowerbd} immediately implies a lower bound for non-distributional prior-free information complexity.
\begin{corollary}[corollary of Theorem~\ref{thm:non-distri-error-lowerbd}]  \label{cor:non-distri-error-lowerbd-Prior-Free}
For every function $f$ and $0\le \epsilon \le 1$, we have
\[
\ICND(f,\epsilon) \ge \ICND(f,0) - 4 |\cX \times \cY| \hc(\sqrt{\epsilon}).
\]
\end{corollary}

Since unless $\mu$ satisfies certain conditions, Theorem~\ref{thm:upper-bd-IC-mu-eps} does not provide an upper bound on $\IC_\mu(f,\eps)$ that is uniform on $\mu$, we cannot apply it directly to bound $\ICND(f,\epsilon)$. However, we will get around this problem by proving that the ``difficult distributions'' satisfy these conditions and hence we obtain the desired upper bound.

\begin{theorem} \label{thm:non-distributional-ub}
 If $f\colon \cX \times \cY \to \cZ$ is non-constant, then
\[ \IC(f,\epsilon) \leq \IC(f,0) - \Omega(h(\epsilon)), \]
 where the hidden constant depends on $f$.
\end{theorem}
\begin{proof}
See Section~\ref{sec:proof:non-distributional-ub}. 
\end{proof}

The same upper bound and lower bound hold for $\ICD(f,\epsilon)$.  

\begin{theorem} \label{thm:ICD-eps-bound}
 If $f\colon \cX \times \cY \to \cZ$ is non-constant, then
\[ 
\ICD(f,0) - O(\hc(\sqrt{\epsilon}))
\le \ICD(f,\epsilon) \leq \ICD(f,0) - \Omega(h(\epsilon)),
\]
 where the hidden constants depend on $f$.
\end{theorem}

\begin{proof}
It is shown in~\cite{MR2961528} that $\ICD(f,0)=\IC(f,0)$, and thus the upper bound follows from Theorem~\ref{thm:non-distributional-ub} as $\ICD(f,\epsilon) \le \IC(f,\epsilon)$.

To prove  the lower bound, choose a measure $\mu$ that maximizes $\IC_\mu(f,\mu,0)$, and let $\alpha = \min_{xy \in \supp\mu} \mu(x,y)$. Applying Theorem~\ref{thm:bd-IC-mu--distributional}, we get
\[
\ICD(f,\epsilon) 
\ge \IC_{\mu}(f,\mu, \epsilon) 
\ge \IC_{\mu}(f,\mu,0) - 4|\cX||\cY| \hc(\sqrt{\epsilon/\alpha})
= \ICD(f,0) - O(\hc(\sqrt{\epsilon})).   \qedhere
\]
\end{proof}

\subsection{A characterization of trivial measures}

We start with a few of definitions. Let $f\colon \cX \times \cY \to \cZ$ be an arbitrary function, and $\mu$ a distribution on $\cX \times \cY$. We say that $\mu$ is \emph{external-trivial} if $\IC^\ext_\mu(f,0) = 0$.  We say that $\mu$ is \emph{strongly external-trivial} if there exists a protocol $\pi$ computing $f$ correctly on all inputs satisfying $\IC^\ext_\mu(\pi) = 0$.  We say that $\mu$ is \emph{structurally external-trivial} if $f$ is constant on $S_A \times S_B$, where $S_A$ is the support of the marginal of $\mu$ on Alice's input and $S_B$ is the support of the marginal of $\mu$ on Bob's input.

Similarly  we say that $\mu$ is \emph{internal-trivial} if $\IC_\mu(f,0) = 0$.  We say that $\mu$ is \emph{strongly internal-trivial} if there exists a protocol $\pi$ computing $f$ correctly on all inputs satisfying $\IC_\mu(\pi) = 0$.  We say that $\mu$ is \emph{structurally internal-trivial} if the marginals of $\mu$ can be partitioned as $S_A = \bigcup_i \cX_i$ and $S_B = \bigcup_i \cY_i$ so that the support of $\mu$ is contained in $\bigcup_i \cX_i \times \cY_i$, and $f$ is constant on each $\cX_i \times \cY_i$.

Theorem~\ref{thm:internal-trivial} below shows that all our definitions of internal triviality are equivalent. In particular, if $\IC_\mu(f,0) = 0$, then the infimum in the definition of $\IC_\mu$ is achieved by a finite protocol.

\begin{theorem} \label{thm:internal-trivial}
 Let $f\colon \cX \times \cY \to \cZ$ be an arbitrary function, and $\mu$ a distribution on $\cX \times \cY$.

 The distribution $\mu$ is internal-trivial iff it is strongly internal-trivial iff it is structurally internal-trivial.	
\end{theorem}
\begin{proof}
See Section~\ref{sec:proof-trivial-measure}.
\end{proof}

In order to prove Theorem~\ref{thm:internal-trivial}, we first obtain a characterization of measures that are not structurally internal-trivial, by defining a graph $G_\mu$ on the support of every distribution $\mu$ on $\cX \times \cY$.
\begin{definition}   \label{def:graph-on-inputs}
 Let $G$ be the graph whose vertex set is $\cX \times \cY$, and two vertices are connected if they agree on one of their coordinates. That is, $(x,y),(x,y')$ are connected for every $x \in \cX$ and $y \neq y' \in \cY$, and $(x,y),(x',y)$ are connected for every $x \neq x' \in \cX$ and $y \in \cY$. In short, $G$ is the Cartesian product of the complete graphs $K_\cX$ and $K_\cY$. Let $G_\mu$ be the subgraph of $G$ induced by the support of $\mu$. For every connected component $C$ of $G_\mu$, define
 \begin{align*}
 C_A &= \{ x \in \cX : xy \in C \text{ for some } y \in \cY \}, \\
 C_B &= \{ y \in \cY : xy \in C \text{ for some } x \in \cX \}.
 \end{align*}
\end{definition}

The following lemma shows that if $\mu$ is not structurally internal-trivial, then there exists a connected component $C$ of $G_\mu$ such that $f$ is not constant on $C_A \times C_B$. We will use this fact later  in Section~\ref{sec:proof:upper-bd-IC-mu-eps}  in the proof of Theorem~\ref{thm:upper-bd-IC-mu-eps}.

\begin{lemma} \label{lem:structurally-internal-trivial}
 Let $f\colon \cX \times \cY \to \cZ$ be an arbitrary function, and $\mu$ a distribution on $\cX \times \cY$. Then the distribution $\mu$ is not structurally internal-trivial iff there exists a connected component $C$ of $G_\mu$ such that $f$ is not constant on $C_A \times C_B$.
\end{lemma}
\begin{proof}
 Suppose first that $\mu$ is structurally internal-trivial. Thus there exist partitions $S_A = \bigcup_i \cX_i$ and $S_B = \bigcup_i \cY_i$ such that the support of $\mu$ is contained in $\bigcup_i \cX_i \times \cY_i$ and $f$ is constant on $\cX_i \times \cY_i$ on each $i$. Any connected component $C$ of $G_\mu$ must lie in some $\cX_i \times \cY_i$. Indeed, if (for example) $x_jy_j,x_jy_k \in C$ where $x_j \in \cX_j$, $y_j \in \cY_j$, $y_k \in \cY_k$, then $x_jy_k \notin \bigcup_i \cX_i \times \cY_i$.
 As $C \subseteq \cX_i \times \cY_i$, we must have $C_A \times C_B \subseteq \cX_i \times \cY_i$, hence $f$ is constant on $C_A \times C_B$ for every connected component $C$.

 Conversely, suppose that for every connected component $C$ of $G_\mu$, the function $f$ is constant on $C_A \times C_B$. If $C,C'$ are two different connected components then $C_A,C'_A$ are disjoint: otherwise, if (say) $(x,y) \in C$ and $(x,y') \in C'$ then $(x,y)$ is connected to $(x,y')$ and so $C = C'$. Thus $\{ C_A : C \text{ a connected component of } G_\mu \}$ partitions a subset $\cX'$ of $\cX$. Similarly, $\{ C_B : C \text{ a connected component of } G_\mu \}$ partitions a subset $\cY'$ of $\cY$. We can obtain partitions of $\cX$ and $\cY$ by adding the parts $\cX \setminus \cX'$ and $\cY \setminus \cY'$. These partitions serve as a witness that $\mu$ is structurally internal-trivial.
\end{proof}

Finally we note that the analogue of Theorem~\ref{thm:internal-trivial}  holds for the external case as well.

\begin{theorem} \label{thm:external-trivial}
 Let $f\colon \cX \times \cY \to \cZ$ be an arbitrary function, and $\mu$ a distribution on $\cX \times \cY$.

 The distribution $\mu$ is external-trivial iff it is strongly external-trivial iff it is structurally external-trivial.	
\end{theorem}
\begin{proof}
See Section~\ref{sec:proof-trivial-measure}.
\end{proof}

\section{Proofs for general functions} \label{sec:Proofs}
In this section we present the proofs of the main results on general functions presented in Section~\ref{sec:main_results}.

\subsection{Information complexity with point-wise error}
\label{sec:pointwise-proofs}

\subsubsection{Proof of Theorem~\ref{thm:upper-bd-IC-mu-eps}}\label{sec:proof:upper-bd-IC-mu-eps}
We discuss some notation before the proof. Consider a protocol $\pi$. For an input $xy$, let $\Pi_{xy}$ denote the random variable corresponding to the transcript of $\pi$ when it is executed on the input $xy$.  Let $\Pi$ denote the random variable for transcripts of $\pi$, whose distribution is given as
\[
\Pr[\Pi = t] = \Ex_{xy} \Pr[\Pi_{xy} = t] = \sum_{xy} \Pr[xy] \Pr[\Pi_{xy} = t],
\]
where $\Pr[\Pi_{xy} = t] = \Pr[\Pi = t| XY=xy]$. As usual we abbreviate $\Pr[xy] = \Pr[XY=xy]$, and $\Pr[x|y] = \Pr[X=x | Y=y]$, and so on.

The next lemma shows that under some conditions, if we modify a protocol $\pi$ to a new protocol  $\pi'$  according to Figure~\ref{fig:modified_protocol}, then the information cost will have a significant drop. 

\begin{figure}[ht]
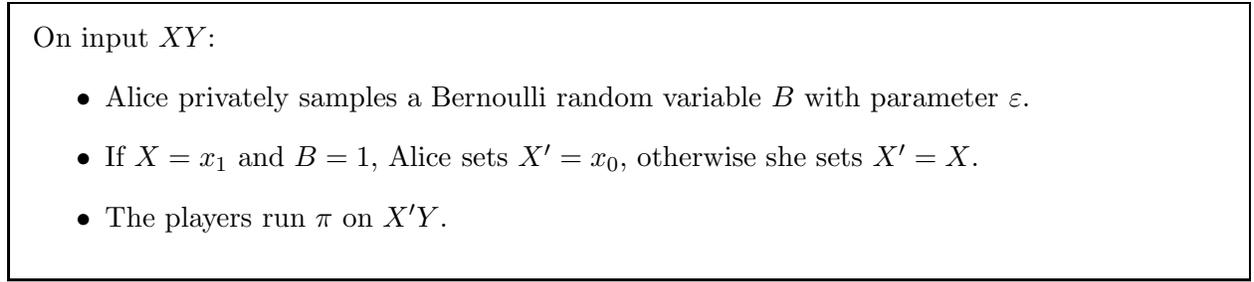

\begin{framed}
On input $XY$:
\begin{itemize}
\item Alice privately samples a Bernoulli random variable $B$ with parameter $\epsilon$.
\item If $X=x_1$ and $B=1$, Alice sets $X'=x_0$, otherwise she sets $X'=X$.
\item The players run $\pi$ on $X'Y$.
\end{itemize}
\end{framed}
\caption{The protocol $\pi'$ is obtained from a protocol $\pi$ using $x_0,x_1 \in \cX$. \label{fig:modified_protocol}}
\end{figure}

\begin{lemma} \label{lem:Given-good-transcripts}
Let  $\mu$ be a distribution on $\cX \times \cY$, and $\pi$ be a protocol with input set $\cX \times \cY$. Suppose there is a set $\cL$ of transcripts of $\pi$ that satisfies, for some $C_1 \in [0,1]$,
\begin{itemize}
 \item[(1)] $\Pr[\Pi \in \cL] \ge C_1$;
\end{itemize}
and there are $x_0 \oy, x_1\oy$, both in the support of $\mu$, and $C_2 \in (0,1], \delta \in [0,1]$ with $C_2 > 2\delta$, such that for every $t \in \cL$,
\begin{itemize}
 \item[(2)] $\Pr[XY=x_0 \oy|\Pi=t] \ge C_2$;
 \item[(3)] $\Pr[XY=x_1 \oy|\Pi=t] \le \delta$.
\end{itemize}
Let $K = \log |\cX \times \cY|$. Then for sufficiently small $\epsilon > 0$ (depending on $\mu, C_2, \delta$), the protocol $\pi'$ defined in Figure~\ref{fig:modified_protocol} satisfies
\[
\IC_\mu(\pi')
\le
\IC_\mu(\pi)
- C_1 C_2 \entf\left( \frac{\epsilon}{2}\min\left\{1, C_2 \frac{\Pr[x_1 \oy]}{\Pr[x_0 \oy]}\right\}  \right) + 3\epsilon K + \hc(\delta/C_2).
\]
Explicitly, the upper bound holds as long as $\frac{\Pr[x_1 \oy]}{\Pr[x_0 \oy]} \epsilon + (1-\epsilon) \delta/C_2 \le 1/2$.
\end{lemma}

Intuitively, this condition  says that $\pi$ has a set of transcripts $\cL$ that happen with significant probability, and every transcript in $\cL$  probabilistically differentiates between $x_0\oy$ and $x_1\oy$. In other words, if we we see a transcript in $\cL$, then we know that the input was much more likely to be $x_0\oy$ than to be $x_1 \oy$. One point to note here is that we require the two points $x_0\oy$ and $x_1\oy$ to be in the same column. By symmetry, if there are two points in the same row satisfying the same properties, then the claim of Lemma~\ref{lem:Given-good-transcripts} also holds.

\begin{proof}
Consider the protocols $\pi$ and $\pi'$ as described in Figure~\ref{fig:modified_protocol}. Note that $\Pi_{X'Y}$ is the transcript of $\pi'$. We shorthand $\Pi'=\Pi_{X'Y}$. The information cost of $\pi'$ is given by
\[
\IC_\mu(\pi') = I(X,\Pi' | Y) + I(Y,\Pi' | X) =H(X|Y)+H(Y|X)-H(X|\Pi'Y)-H(Y|\Pi'X),
\]
while
\[
\IC_\mu(\pi)=I(X,\Pi  | Y) + I(Y,\Pi  | X) =H(X|Y)+H(Y|X)-H(X|\Pi Y)-H(Y|\Pi X).
\]
Hence
\[
\IC_\mu(\pi) - \IC_\mu(\pi') = H(X|\Pi'Y) - H(X|\Pi Y) + H(Y|\Pi'X) -H(Y|\Pi X).
\]
Note that
\begin{equation}    \label{eq:bound-X}
H(Y|\Pi' X) \ge H(Y|\Pi' XB) \ge \left(1-\epsilon \right) H(Y|\Pi' X, (B=0)) = \left(1-\epsilon \right) H(Y|\Pi  X)  \ge H(Y|\Pi  X) - \epsilon K.
\end{equation}
Similarly,  for every $y \in \cY$ and every possible transcript $t$, we have
\begin{equation}   \label{eq:bound-Y}
H(X|\Pi' Y=ty) \ge  H(X|\Pi  Y=ty) -  \epsilon K.
\end{equation}
We will show that for $Y=\oy$ and every transcript $t \in \cL$,
\begin{equation}  \label{eq:to-bound}
H(X|\Pi' Y=t\oy)
\ge H(X|\Pi Y=t\oy) + \entf\left( \frac{\epsilon}{2}\min\left\{1, C_2 \frac{\Pr[x_1 \oy]}{\Pr[x_0 \oy]}\right\}  \right) - \hc(\delta/C_2) - \epsilon K.
\end{equation}
Note that Condition (2) implies that for $t \in \cL$,
\[
\Pr[\Pi Y = t \oy] \ge \Pr[\Pi XY = t x_0 \oy] = \Pr[XY=x_0 \oy | \Pi=t] \Pr[\Pi=t] \ge C_2 \Pr[\Pi=t].
\]
Hence
\[
\Pr[\Pi \in \cL, Y = \oy] \ge C_2 \Pr[\Pi \in \cL] \ge C_1 C_2.
\]
This  together with \eqref{eq:bound-Y} and \eqref{eq:to-bound} would show that
\begin{align*}
H(X|\Pi'Y)
&=  \sum_{t} \sum_{y \in \cY} \Pr[\Pi'Y=ty] H(X|\Pi'Y=ty)
 \ge \sum_{t} \sum_{y \in \cY} (1-\epsilon)\Pr[\Pi Y =t y] H(X|\Pi'Y=ty) \\
&\ge \sum_{t} \sum_{y \in \cY} \Pr[\Pi Y=t y] H(X|\Pi Y=ty) \\
 &\qquad + \Pr[\Pi \in \mathcal{L},Y=\oy]  \left( \entf\left( \frac{\epsilon}{2}\min\left\{1, C_2 \frac{\Pr[x_1 \oy]}{\Pr[x_0 \oy]}\right\}  \right) - \hc(\delta/C_2) \right) - 2\epsilon K
 \\
&\ge  H(X|\Pi Y)+  C_1 C_2 \entf\left( \frac{\epsilon}{2}\min\left\{1, C_2 \frac{\Pr[x_1 \oy]}{\Pr[x_0 \oy]}\right\}  \right) - 2\epsilon K - \hc(\delta/C_2).
\end{align*}
Applying \eqref{eq:bound-X} would immediately give the claimed bound.

Our aim, then, is to show \eqref{eq:to-bound}. From now on we consider exclusively $t \in \mathcal{L}$.

The idea is to consider the indicator variable $C \defeq 1_{[X \neq x_1]}$. Since $C$ is a deterministic function of $X$, we have
\begin{equation}   \label{eq:bd-main-term}
H(X|\Pi' Y=t\oy) = H(XC|\Pi'Y=t\oy) = H(X|C, (\Pi'Y=t\oy)) + H(C| \Pi'Y=t\oy).
\end{equation}
Since $\Pr[XY = x_0 \oy | \Pi=t] = \Pr[Y=\oy | \Pi = t] \Pr[X = x_0| \Pi Y=t \oy]$, by Condition (2) we obtain
\begin{equation}    \label{eq:new-Condition-2}
\Pr[X = x_0| \Pi Y=t \oy] \ge \Pr[XY = x_0 \oy | \Pi=t] \ge C_2,
\end{equation}
and $\Pr[Y=\oy | \Pi = t] \ge C_2$. Similarly, as $\Pr[XY = x_1 \oy | \Pi=t] = \Pr[Y=\oy | \Pi = t] \Pr[X = x_1| \Pi Y=t \oy]$, we obtain by Condition (3) that
\begin{equation}    \label{eq:new-Condition-3}
\Pr[X = x_1| \Pi Y=t \oy] = \frac{\Pr[XY = x_1 \oy | \Pi=t]}{\Pr[Y=\oy | \Pi = t]} \le \frac{\delta}{C_2}.
\end{equation}
Hence using \eqref{eq:new-Condition-3}, the first term in \eqref{eq:bd-main-term} can be bounded as
\begin{align}
H(X|C, (\Pi'Y=t\oy))
&\ge  (1-\epsilon) H(X|C, (B\Pi'Y=0t\oy)) \notag \\
&\geq H(X|C, (\Pi Y=t\oy))- \epsilon K \notag \\
&= H(XC|\Pi Y=t\oy) - H(C|\Pi Y=t\oy) -  \epsilon K  \notag \\
&\ge H(X|\Pi Y=t\oy) - \hc(\delta/C_2) - \epsilon K. \label{eq:bd-main-term-1}
\end{align}

To bound the second term $H(C| \Pi'Y=t\oy )$ in \eqref{eq:bd-main-term}, we must study $\Pr[X=x_1|\Pi'Y=t\oy]$. We will use
\begin{equation}   \label{eq:fraction-expression}
\Pr[C=0|\Pi'Y=t\oy] = \Pr[X=x_1|\Pi'Y=t\oy] = \frac{\Pr[\Pi' XY=tx_1 \oy]}{\Pr[\Pi' Y=t\oy]}.
\end{equation}
Consider the numerator first. By the definition of $\pi'$,
\begin{align}
\Pr[\Pi'XY=tx_1 \oy]
&= \Pr[\Pi'=t | XY=x_1 \oy] \Pr[x_1 \oy] \notag \\
&= (\epsilon \Pr[\Pi=t | XY=x_0 \oy] + (1-\epsilon) \Pr[\Pi=t | XY=x_1 \oy]) \Pr[x_1 \oy] \notag \\
&= \epsilon \Pr[\Pi XY=t x_0 \oy] \frac{\Pr[x_1 \oy]}{\Pr[x_0 \oy]} + (1-\epsilon) \Pr[\Pi XY=t x_1 \oy]. \label{eq:numerator}
\end{align}
For the denominator of \eqref{eq:fraction-expression}, we have
\begin{equation}   \label{eq:ub-temp3}
\Pr[\Pi' Y=t\oy]
\ge \Pr[\Pi' XY=t x_0 \oy] = \Pr[\Pi XY=t x_0 \oy].
\end{equation}
By Conditions (2) and (3),
\begin{equation}   \label{eq:ub-temp4}
\frac{\Pr[\Pi XY=t x_1 \oy]}{\Pr[\Pi XY=t x_0 \oy]}
= \frac{\Pr[XY=x_1 \oy | \Pi = t]}{\Pr[XY=x_0 \oy | \Pi=t]} \le \delta/C_2.
\end{equation}
Combining \eqref{eq:fraction-expression}, \eqref{eq:numerator}, \eqref{eq:ub-temp3} and \eqref{eq:ub-temp4}, we obtain the following upper bound on \eqref{eq:fraction-expression}:
\begin{equation}   \label{eq:upperbd}
\Pr[X=x_1|\Pi'Y=t\oy]
\le   \frac{\Pr[x_1 \oy]}{\Pr[x_0 \oy]} \epsilon + (1-\epsilon) \delta/C_2.
\end{equation}
To obtain a lower bound for \eqref{eq:fraction-expression} note
\begin{align}
\Pr[\Pi' Y=t\oy]
&= \sum_{x} \Pr[\Pi' XY=tx\oy] = \sum_{x \neq x_1}\Pr[\Pi' XY=tx\oy] + \Pr[\Pi' XY=tx_1 \oy] \notag \\
&= \sum_{x \neq x_1}\Pr[\Pi XY=tx\oy] + \epsilon \Pr[\Pi XY=t x_0 \oy] \frac{\Pr[x_1 \oy]}{\Pr[x_0 \oy]} + (1-\epsilon) \Pr[\Pi XY=t x_1 \oy] \notag \\
&\le \sum_{x}\Pr[\Pi XY=tx\oy] + \epsilon \Pr[\Pi XY=t x_0 \oy] \frac{\Pr[x_1 \oy]}{\Pr[x_0 \oy]}  \notag \\
&=\Pr[\Pi Y=t\oy] + \epsilon \Pr[\Pi XY=t x_0 \oy] \frac{\Pr[x_1 \oy]}{\Pr[x_0 \oy]}   \notag \\
&\le 2 \max\left\{\Pr[\Pi Y=t\oy], \Pr[\Pi XY=t x_0 \oy] \frac{\Pr[x_1 \oy]}{\Pr[x_0 \oy]}\right\}.  \label{eq:denomenator-term}
\end{align}
Hence by \eqref{eq:fraction-expression}, \eqref{eq:numerator} and \eqref{eq:denomenator-term},
\begin{align}
\Pr[X=x_1|\Pi'Y=t\oy]
&\ge \frac{ \epsilon \Pr[\Pi XY=t x_0 \oy] \frac{\Pr[x_1 \oy]}{\Pr[x_0 \oy]} }{2 \max\{\Pr[\Pi Y=t\oy], \Pr[\Pi XY=t x_0 \oy] \frac{\Pr[x_1 \oy]}{\Pr[x_0 \oy]}\}}   \notag \\
&\ge \frac{\epsilon}{2} \min\left\{1, C_2 \frac{\Pr[x_1 \oy]}{\Pr[x_0 \oy]}\right\}.    \label{eq:lowerbd}
\end{align}
where we used $\Pr[\Pi XY=t x_0 \oy] / \Pr[\Pi Y=t\oy] = \Pr[X=x_0 | \Pi Y=t\oy] \ge C_2$ by \eqref{eq:new-Condition-2}. Thus we have shown that
\begin{equation}   \label{eq:both-bounds}
 \frac{\epsilon}{2}\min\left\{1, C_2 \frac{\Pr[x_1 \oy]}{\Pr[x_0 \oy]}\right\}
\le
\Pr[X=x_1|\Pi'Y=t\oy]
\le \frac{\Pr[x_1 \oy]}{\Pr[x_0 \oy]} \epsilon + \frac{(1-\epsilon) \delta}{C_2}.
\end{equation}
This together with \eqref{eq:bd-main-term} and \eqref{eq:bd-main-term-1} gives \eqref{eq:to-bound} as desired, as long as $\epsilon > 0$ is small enough such that the upper bound in \eqref{eq:both-bounds} is at most $1/2$.
\end{proof}

\restate{Theorem~\ref{thm:upper-bd-IC-mu-eps}}{
Consider a function $f\colon \cX \times \cY \to \cZ$ and a probability measure $\mu$ on $\cX \times \cY$ such that $\IC_\mu(f,0) > 0$. There exist positive constants $\tau,\epsilon_0$, depending on $f$ and $\mu$, such that for every $\epsilon \leq \epsilon_0$,
\[
 \IC_\mu(f,\epsilon) \le \IC_\mu(f,0) - \tau \entf(\epsilon).
\]
Moreover:
\begin{description}
 \item[Non-constant case:] Suppose that $f(a) \neq f(b)$ for two points $a,b$ in the support of $\mu$, and on the same row or column. Then one can take $\tau \geq \mu(a)^2 \mu(b)/64$, and $\epsilon_0$ depends only on $\min(\mu(a),\mu(b))$ and $|\cX \times \cY|$.
 \item[AND case:] Let $x_0,x_1 \in \cX$ and $y_0,y_1 \in \cY$. Suppose that $f(x_0y_0) = f(x_0y_1) = f(x_1y_0) = z_0$ and $f(x_1y_1) = z_1 \neq z_0$, and that $x_0y_0,x_0y_1,x_1y_0 \in \supp\mu$. Then one can take $\tau \geq \frac{\mu(x_0y_0)^2}{128} \min(\mu(x_0y_1),\mu(x_1y_0))$, and $\epsilon_0$ depends only on $|\cX \times \cY|$ and the minimum of $\mu(x_0y_0)$, $\mu(x_0y_1)$, $\mu(x_1y_0)$.
\end{description}}

\begin{proof}
In order to apply the assumption  $\IC_\mu(f,0) > 0$, we will need to use our characterization of internal-trivial measures. Consider the graph $G_\mu$ defined on $\cX \times \cY$ as given in Definition~\ref{def:graph-on-inputs}.  By Theorem~\ref{thm:internal-trivial} and Lemma~\ref{lem:structurally-internal-trivial}, the assumption $\IC_\mu(f,0) > 0$ implies the existence of  a connected component $C$ of $G_\mu$ such that $f$ is not constant on $C_A \times C_B$. Note that $C \subseteq \supp \mu$, and $C_A \times C_B$ is the corresponding rectangle given by $C$.

\paragraph{Case I:} $f$ is not constant on $C$.

As $C$ is connected, there must be two adjacent points $a, b\in C$ such that $f(a) \neq f(b)$. By our definition of adjacency in Definition~\ref{def:graph-on-inputs}, without loss of generality we can assume that $a, b$ are in the same column. Now consider any protocol $\pi$ that solves $[f, 0]$. Let $\cL_0$ be the set of the transcripts that can occur when $\pi$ runs with input $a$; formally,
\[
\cL_0 = \{t: \Pr[\Pi_a = t] > 0 \}.
\]
Clearly $\Pr[\Pi \in \cL_0] \ge \mu(a)$. As $f(a) \neq f(b)$ and $\pi$ has no error, for every $t \in \cL_0$,
\begin{equation}   \label{eq:Case-1-Condition3}
\Pr[XY = b | \Pi = t] = 0.
\end{equation}
Let
\begin{equation}   \label{eq:Case-1-Condition2}
\cL = \{t \in \cL_0 : \Pr[XY=a | \Pi =t] \ge \mu(a)/2 \}.
\end{equation}
We claim
\begin{equation}   \label{eq:Case-1-Condition1}
\Pr[\Pi \in \cL] \ge \mu(a)/2.
\end{equation}
Indeed, note
\[
 \sum_{t \in \cL_0} \Pr[\Pi=t] \Pr[XY=a | \Pi =t] = \sum_t \Pr[\Pi=t] \Pr[XY=a | \Pi =t] = \mu(a),
\]
use the trivial bound $\Pr[XY=a | \Pi =t] \le 1$, we have
\begin{align*}
\mu(a)
&=
\sum_{t \in \cL} \Pr[\Pi=t] \Pr[XY=a | \Pi =t] + \sum_{t \in \cL_0 \setminus \cL} \Pr[\Pi=t] \Pr[XY=a | \Pi =t]  \\
&\le \sum_{t \in \cL} \Pr[\Pi=t] +  \frac{\mu(a)}{2} \sum_{t \in \cL_0 \setminus \cL} \Pr[\Pi=t]
= \Pr[\Pi \in \cL] +  \frac{\mu(a)}{2} (1 - \Pr[\Pi \in \cL]),
\end{align*}
which gives $\Pr[\Pi \in \cL] \ge \mu(a)/2$, as claimed. For small enough $\epsilon$, the set $\cL$ and the points $a, b$ satisfy the three conditions in Lemma~\ref{lem:Given-good-transcripts} with $C_1 = C_2 = \mu(a)/2$ and $\delta = 0$, respectively from \eqref{eq:Case-1-Condition1}, \eqref{eq:Case-1-Condition2} and \eqref{eq:Case-1-Condition3}.  We conclude that
\[
 \IC_\mu(f,\epsilon)
 \leq
 \IC_\mu(f,0) - \frac{\mu(a)^2}{4} h\left(\frac{\mu(b)}{4} \epsilon\right) + 3\epsilon K
 \text{ whenever } \frac{\mu(b)}{\mu(a)} \epsilon \leq 1/2,
\]
where $K = \log |\cX \times \cY|$. Hence when $\epsilon \leq 1/2$, by \eqref{eq:conv-factor} we  have
\[
 \IC_\mu(f,\epsilon) \leq
 \IC_\mu(f,0) - \frac{\mu(a)^2 \mu(b)}{32} h(\epsilon) + 3\epsilon K.
\]
We can thus find $\epsilon_0 > 0$, depending only on $\mu(a),\mu(b),K$, such that for $\epsilon \leq \epsilon_0$,
\[
 \IC_\mu(f,\epsilon)
 \leq \IC_\mu(f,0) - \frac{\mu(a)^2 \mu(b)}{64} h(\epsilon).
\]

\paragraph{Case II:} $f$ is constant on $C$ but not on $C_A \times C_B$.

We first make a simple observation:
\begin{description}
\item[Property A:] For any protocol $\pi$ that performs $[f, 0]$, and for every transcript $t$ of $\pi$, there exists at least one point $b \in C$ (which can depend on $t$) such that $\Pr[XY=b | \Pi = t] = 0$.
\end{description}
Indeed, otherwise $f$ would be constant on $C_A \times C_B$ by the rectangle property of protocols (i.e. $\Pr[\Pi=t|x_1y_1]\Pr[\Pi=t|x_2y_2]=\Pr[\Pi=t|x_1y_2]\Pr[\Pi=t|x_2y_1]$ for all $x_1,x_2,y_1,y_2$).

Given a protocol $\pi$ that performs $[f,0]$ and a point $a \in C$, let the set $\cL(\pi,a)$ of transcripts be defined as
\[
\cL(\pi,a) = \{t : \Pr[XY=a | \Pi =t] \ge \mu(a)/2 \}.
\]
The same argument as in \textbf{Case I} shows that $\Pr[\Pi \in \cL(\pi,a)] \ge \mu(a)/2$. For any other point $b \in C$, define
\[
\cL(\pi,a, b) = \{t \in \cL(\pi,a): \Pr[XY=b| \Pi =t]= 0\}.
\]
Let $k \defeq |C|$; necessarily $k \ge 3$. By \textbf{Property A}, we have
\[
\cL(\pi,a) = \bigcup_{b \in C} \cL(\pi,a, b).
\]
This implies the existence of a point $b \in C$ with $\Pr[\Pi \in \cL(\pi,a,b)] \ge \Pr[\Pi \in \cL(\pi,a)] / k \ge \mu(a)/2k$. To sum up, we have shown that there exist two different points $a, b \in C \subseteq \supp \mu$ such that the set of transcripts $\cL(\pi,a,b)$ satisfies the following properties:
\begin{itemize}
\item[(1')] $\Pr[\Pi \in \cL(\pi,a,b)] \ge \mu(a)/2k$;
\item[(2')] $\Pr[XY=a | \Pi =t] \ge \mu(a)/2$ for every $t \in \cL(\pi,a,b)$;
\item[(3')] $\Pr[XY=b | \Pi =t] =0$ for every $t \in \cL(\pi,a,b)$.
\end{itemize}

Now consider a sequence of protocols $\pi_n$ that all perform $[f,0]$ and $\lim_{n\to \infty} \IC_\mu(\pi_n) = \IC_\mu(f,0)$. Fix (arbitrarily) a point $a \in C$. For every protocol $\pi_n$ we construct $\cL(\pi_n, a, b_{\pi_n})$ as above. Since there are only $k-1$ different values of $b$, by picking a subsequence of $\pi_n$ if necessary, without loss of generality, we may assume that for some point $b \in C$, $b_{\pi_n} = b$ for all $\pi_n$. Hence for every $\pi_n$ we have a set of transcripts $\cL(\pi_n, a,b)$ such that properties (1'), (2') and (3') are all satisfied.

If we compare these three conditions with the conditions in Lemma~\ref{lem:Given-good-transcripts}, we find that the only issue is that we do not know whether $a$ and $b$ are in the same row or column (in terms of the graph $G_\mu$, whether $a$ and $b$ are adjacent). 

\subparagraph{Case IIa:} $a,b$ are adjacent in $G_\mu$. As we expand on below, we can guarantee that this case happens in the \emph{AND case} (see theorem statement) by choosing $a = x_0y_0$.

For small enough $\epsilon$, the set $\cL(\pi,a,b)$ and the points $a,b$ satisfy the three conditions in Lemma~\ref{lem:Given-good-transcripts} with $C_1 = \mu(a)/2k$, $C_2 = \mu(a)/2$ and $\delta = 0$, respectively from (1'), (2') and (3'). We conclude that
\[
 \IC_\mu(f, \epsilon) \leq \IC_\mu(f,0) - \frac{\mu(a)^2}{4k} h\left(\frac{\mu(b)}{4} \epsilon\right) + 3\epsilon K \text{ whenever } \frac{\mu(b)}{\mu(a)} \epsilon \leq 1/2,
\]
where $K = \log |\cX \times \cY|$. Repeating the calculations of Case~I, we can find $\epsilon_0 > 0$, depending only on $\mu(a),\mu(b),K$, such that for $\epsilon \leq \epsilon_0$,
\[
 \IC_\mu(f,\epsilon)
 \le \IC_\mu(f,0) - \frac{\mu(a)^2 \mu(b)}{64k} h(\epsilon).
\]

Suppose now that we are in the AND case. Choosing $a = x_0y_0$, we see that Property~A must hold for some $b \in \{x_0y_1, x_1y_0\}$, since a transcript having positive probability on both $x_0y_1$ and $x_1y_0$ also has positive probability on $x_1y_1$, whereas $f(x_0y_1) \neq f(x_1y_1)$ by assumption. Property~(1') thus holds with $k = 2$, and we conclude that for $\epsilon \leq \epsilon_0$,
\[
 \IC_\mu(f,\epsilon)
 \le \IC_\mu(f,0) - \frac{\mu(a)^2 \mu(b)}{128} h(\epsilon).
\]

\subparagraph{Case IIb:} $a,b$ are not adjacent in $G_\mu$. To handle this case, we run a \emph{binary search} along a shortest path connecting $a$ and $b$ in $C$.

Pick an arbitrary point $c \in C$ in some shortest path connecting $a$ and $b$. For every $\pi_n$, sort the transcripts in $\cL(\pi_n, a, b)$ according to $p_{n,t,c} \defeq \Pr[XY=c | \Pi_n=t]$ in increasing order, where $\Pi_n$ is the random variable representing the transcript of $\pi_n$. Let $m_n$ be the median of the sequence $p_{n,t,c}$ according to the conditional probability measure $\nu_n (t) \defeq \Pr[\Pi_n=t | t \in \cL(\pi_n,a,b)]$, i.e.,
\begin{equation}   \label{eq:median}
\nu_n(\{t\in \cL(\pi_n, a,b): p_{n,t,c} \le m_n \}), \nu_n(\{t\in \cL(\pi_n, a,b): p_{n,t,c} \geq m_n \}) \geq 1/2.
\end{equation}
Such a median always exists: if $m_n$ is the smallest value such that $\nu_n(\{t\in \cL(\pi_n, a,b): p_{n,t,c} \le m_n \}) \geq 1/2$ then $\nu_n(\{t\in \cL(\pi_n, a,b): p_{n,t,c} \ge m_n \}) = 1 -\nu_n(\{t\in \cL(\pi_n, a,b): p_{n,t,c} < m_n \}) \geq 1/2$.

As trivially $m_n \in [0,1]$, the sequence $m_n$ must have a convergent subsequence. Again by picking a subsequence from $m_n$ if necessary, we may assume that the sequence $m_n$ itself is convergent, say $\lim_{n\to \infty} m_n = m$; moreover, if $m > 0$, by picking another subsequence we can assume that $m_n \geq m/2$ for all $n$. The \emph{binary search} algorithm is then given as:
\begin{itemize}
\item If $m =0$, update the set of transcripts to
\begin{equation}   \label{eq:new-L-0}
\cL(\pi_n, a, c) \defeq \{t \in \cL(\pi_n, a, b): p_{n,t,c} \le m_n\},
\end{equation}
and continue the algorithm with $b$ replaced by $c$;
\item If $m >0$, update the set of transcripts to
\begin{equation}  \label{eq:new-L-positive}
\cL(\pi_n, c, b) \defeq \{t \in \cL(\pi_n, a, b): p_{n,t,c} \ge m_n\},
\end{equation}
and continue the algorithm with $a$ replaced by $c$.
\end{itemize}
We argue that the three properties are roughly preserved. In the case $m=0$, Property~(2') is kept, while Properties~(1') and~(3') change to
\[
\Pr[\Pi_n \in \cL(\pi_n,a,c)] \ge \mu(a)/4k \quad \text{and\ }\quad
\Pr[XY=c | \Pi_n =t] \le m_n, \ \forall\ t \in \cL(\pi_n,a,c),
\]
respectively. In the case $m > 0$, Property~(3') is preserved while Properties~(1') and~(2') change to
\[
\Pr[\Pi_n \in \cL(\pi_n,c,b)] \ge \mu(a)/4k \quad \text{and\ }\quad
\Pr[XY=c | \Pi_n =t] > m/2, \ \forall\ t \in \cL(\pi_n,c,b).
\]
In either case, we have seen that the new set of transcripts $\cL(\pi_n, a, b)$ together with the new two points $a$ and $b$ satisfy Condition (1), (2) and (3) in Lemma~\ref{lem:Given-good-transcripts} with proper constants (e.g., $\delta_n$ in Condition (3) is at most $m_n$ for protocol $\pi_n$, and $m_n \to 0$). After finitely many steps, the binary search algorithm has to stop and return two adjacent points $a$ and $b$. Suppose that it stops after $s$ steps; note that $s \leq \lceil \log k \rceil$. Lemma~\ref{lem:Given-good-transcripts} then gives the upper bound
\begin{equation} \label{eq:caseII-ub}
 \IC_\mu(f,\epsilon) \leq \IC_\mu(\pi_n) - \frac{\mu(a)}{2^{s+1} k} C_2 h \left( \frac{\epsilon}{2}\min\{1, C_2 R\}  \right) + 3\epsilon K + \hc(\delta_n/C_2).
\end{equation}
for some $C_2,R,K > 0$ (where $C_2, R$ depend on $\mu$) and a sequence $\delta_n$ tending to zero, assuming that
\[
 R \epsilon + (1-\epsilon) \delta_n/C_2 \leq 1/2 \text{ and } \delta_n/C_2 \leq 1/2.
\]
By picking a subsequence, we can assume that $\delta_n \leq C_2/4$ for all $n$. Lemma~\ref{lem:Given-good-transcripts} then applies for all $\epsilon \leq 1/(4R)$. Taking the limit of the right-hand side of~\eqref{eq:caseII-ub} as $n\to\infty$, we obtain
\[
 \IC_\mu(f,\epsilon) \leq \IC_\mu(f,0) - \frac{\mu(a)}{2^{s+1}k} C_2  h \left( \frac{\epsilon}{2}\min\{1, C_2 R\}  \right) + 3\eps K = \IC_\mu(f,0) - \Omega(h(\eps)). \qedhere
\]
\end{proof}

\subsubsection{Proof of Theorem~\ref{thm:non-distri-error-lowerbd}}\label{proof:non-distri-error-lowerbd}

\restate{Theorem~\ref{thm:non-distri-error-lowerbd}}{
For all $f, \mu, \epsilon$, we have
\[
\IC_\mu(f,\epsilon) \ge \IC_\mu(f,0) - 4 |\cX| |\cY| \hc(\sqrt{\epsilon}).
\]
}
\begin{proof}[Proof of Theorem \ref{thm:non-distri-error-lowerbd}]
Without loss of generality assume that $\mu$ is  a full-support distribution as otherwise we can approximate it by a sequence of full-support distributions and appeal to the continuity of $\IC_\nu(f,\epsilon)$ with respect to $\nu$. Consider a protocol $\pi$ that performs $[f, \epsilon]$. For every leaf $\ell$ of $\pi$, let $z_\ell$ and $\mu_\ell$ respectively denote the output of the leaf, and the distribution of the inputs conditioned on  the leaf $\ell$.  We will complete it into a protocol $\pi'$ that performs $[f, 0]$, as follows.

\begin{framed}
On input $(X,Y)$:
\begin{itemize}
\item Alice and Bob run the protocol $\pi$  and reach a leaf $\ell$;
\item For every $(x,y) \in \Omega_\ell \defeq \{(x,y): f(x,y)\neq z_\ell\}$, Alice and Bob verify whether $XY = xy$, as follows:
\begin{itemize}
\item If $\mu_\ell(x) \le \mu_\ell(y)$,  Alice reveals whether $X = x$ to Bob, and if yes, Bob reveals whether $Y=y$ to Alice. If $XY=xy$, they terminate.

\item If $\mu_\ell(x) > \mu_\ell(y)$,  Bob initiates the verification process.
\end{itemize}
\end{itemize}
\end{framed}

Clearly, in the end, either both Alice and Bob already revealed their inputs to each other, or otherwise they know $XY \notin \Omega_\ell$, and hence $z_\ell$ is the correct output.  Therefore $\pi'$ performs the task $[f, 0]$.

Next we analyze $\IC_\mu(\pi')$. Let $\pi_{\ell, xy}$ denote the sub-protocol that starts with the distribution $\mu_\ell$ and verifies whether $XY = xy$.  In the case when Alice initiates the verification procedure, we have
\[
\IC_{\mu_\ell}(\pi_{\ell,xy})
= h(\mu_\ell(x)) + \mu_\ell(x)  h\left(\frac{\mu_\ell(x,y)}{\mu_\ell(x)}\right)
\le h(\mu_\ell(x)) + \mu_\ell(x) \le  2 \hc(\mu_\ell(x)),
\]
where by an abuse of notation we are denoting by $\mu_\ell(x)$ the marginal of $\mu_\ell$ on $x$. We can obtain a similar bound for the case where Bob initiates the process, and hence
\begin{align*}
\IC_{\mu_\ell}(\pi_{\ell,xy})
	&\le 2 \min\{\hc(\mu_\ell(x)), \hc(\mu_\ell(y)) \}\\
	&= 2\hc\Big(\mu_\ell(x, y) + \min\{\Pr_{\mu_\ell}[X \neq x, Y=y], \Pr_{\mu_\ell}[X =x, Y \neq y]\}\Big)   \\
	&\le  2\hc(\mu_\ell(x, y))  + 2 \hc\Big( \min\{\Pr_{\mu_\ell}[X \neq x, Y=y], \Pr_{\mu_\ell}[X =x, Y \neq y]\}\Big)
\end{align*}
by the subadditivity of $\hc$.  Using the monotonicity of $\hc$ together with  $\min\{a,b\} \leq \sqrt{ab}$, we obtain that
\begin{equation} \label{eq:NDLB-IC-leaf}
\IC_{\mu_\ell}(\pi_{\ell,xy})
\le 2 \hc(\mu_\ell(x, y))  +   2\hc\Big( \sqrt{\Pr_{\mu_\ell}[X=x, Y\neq y] \Pr_{\mu_\ell}[X \neq x, Y=y]} \Big)
\end{equation}
holds for every leaf $\ell$ and $(x,y) \in \Omega_\ell$. Let $\Pi_{\ell,xy}$ denote the transcript of  $\pi_{\ell,xy}$. Since $\pi_{\ell,xy}$  is a deterministic protocol, we have
$H_{\mu_\ell}(\Pi_{\ell,xy}|XY)=0$, and thus
\[
	\IC_{\mu_\ell}(\pi_{\ell,xy})
	= I(\Pi_{\ell,xy} ; Y | X) + I(\Pi_{\ell,xy} ; X | Y)
	= H_{\mu_\ell}(\Pi_{\ell,xy} | X) + H_{\mu_\ell}(\Pi_{\ell,xy} | Y).
\]
Thus the sub-additivity of entropy implies that the information cost of running all the protocols $\pi_{\ell,xy}$ (for all $x,y \in \Omega_\ell$) is bounded by the sum of their individual information cost. Let $\ell$ be a leaf of $\pi$ sampled by running $\pi$ on a random input.  By \eqref{eq:NDLB-IC-leaf},
\begin{align}
\IC_\mu(\pi') - \IC_\mu(\pi)
&\le	\Ex_\ell  \sum_{xy \in \Omega_\ell}   \IC_{\mu_\ell}(\pi_{\ell,xy})
=	\sum_{(x,y) \in \cX \times \cY} \Ex_\ell \charf{z_\ell \neq f(x,y)}  \IC_{\mu_\ell}(\pi_{\ell,xy}) \notag\\
&\le  \sum_{(x,y) \in \cX \times \cY} 2 \Ex_\ell \charf{z_\ell \neq f(x,y)} \hc(\mu_\ell(x, y)) \ +    \notag\\
&\phantom{=} \sum_{(x,y) \in \cX \times \cY} 2 \Ex_\ell \charf{z_\ell \neq f(x,y)} \hc\Big( \sqrt{\Pr_{\mu_\ell}[X=x, Y\neq y] \Pr_{\mu_\ell}[X \neq x, Y=y]} \Big)    \notag\\
&\le  \sum_{(x,y) \in \cX \times \cY}  2\hc\Big(\Ex_\ell \charf{z_\ell \neq f(x,y)} \mu_\ell(x, y) \Big)  \ +  \notag \\
&\phantom{=} \sum_{(x,y) \in \cX \times \cY}  2\hc\Big(\Ex_\ell \sqrt{\charf{z_\ell \neq f(x,y)} \Pr_{\mu_\ell}[X=x, Y\neq y] \Pr_{\mu_\ell}[X \neq x, Y=y]} \Big) \label{eq:NDLB-diff}
\end{align}
where we used the concavity of $\hc$ in the last step.

For the first summand, we have that for every $(x, y)$,
\begin{align}
\Ex_\ell \charf{z_\ell \neq f(x,y)} \mu_\ell(x, y) &= \sum_\ell \Pr[XY=xy, \pi \text{\ reaches\ }\ell] \charf{z_\ell \neq f(x,y)}  \notag \\ 
&= \sum_\ell \Pr[\pi \text{\ reaches\ }\ell\ |\ XY=xy] \mu(xy) \charf{z_\ell \neq f(x,y)}   \notag \\
&= \mu(xy) \sum_\ell \Pr[\pi_{x,y} \text{\ reaches\ }\ell] \charf{z_\ell \neq f(x,y)} = \mu(xy) \Pr[\pi(x,y) \neq f(x,y)] \notag \\
&\le \mu(xy) \epsilon \le \epsilon, \label{eq:NDLB-summand-1}
\end{align}
where we used that by definition $\mu_\ell(xy) = \Pr[(X,Y)=(x,y)\ |\ \pi \text{\ reaches\ }\ell]$, and the fact that the protocol $\pi$ performs the task $[f, \epsilon]$.

For the second summand in (\ref{eq:NDLB-diff}), since $\mu_\ell$ is obtained by scaling rows and columns of $\mu$, we have
\[
	\frac{\Pr_\mu[X=x,Y=y] \Pr_\mu[X\neq x,Y\neq y] }{\Pr_\mu[X=x,Y\neq y] \Pr_\mu[X\neq x,Y= y]}=\frac{\Pr_{\mu_\ell}[X=x,Y=y] \Pr_{\mu_\ell}[X\neq x,Y\neq y] }{\Pr_{\mu_\ell}[X=x,Y\neq y] \Pr_{\mu_\ell}[X\neq x,Y= y]}
\]
Define (recall that we assumed $\mu$ is of full support)
\[ a_\ell = \charf{z_\ell \neq f(x,y)} \frac{\Pr_{\mu_\ell}[X=x, Y=y]}{\Pr_{\mu}[X=x,Y=y]},
\quad
b_\ell = \frac{\Pr_{\mu_\ell}[X\neq x, Y\neq y]}{\Pr_{\mu}[X\neq x,Y\neq y]},
\]
and note that
\begin{align}   \label{eq:NDLB-summand-2-estimate}
\begin{split}
\charf{z_\ell \neq f(x,y)} \Pr_{\mu_\ell}[X=x, Y\neq y] \Pr_{\mu_\ell}[Y=y, X \neq x]
= a_\ell b_\ell \Pr_{\mu}[X=x,Y\neq y] \Pr_{\mu}[X\neq x,Y= y]  \le a_\ell b_\ell.
\end{split}
\end{align}
Since
\[
\Ex_\ell a_\ell = \frac{1}{\mu(xy)} \Ex_\ell \charf{z_\ell \neq f(x,y)} \mu_\ell(x, y) = \Pr[\pi(x,y) \neq f(x,y)] \le \epsilon
\]
by \eqref{eq:NDLB-summand-1}, and $\Ex_\ell b_\ell = 1$, we can bound the second summand in (\ref{eq:NDLB-diff}) using the Cauchy-Schwarz inequality by
\begin{equation}   \label{eq:NDLB-summand-2-bd}
\Ex_\ell \sqrt{a_\ell b_\ell}
\le \sqrt{\Ex_\ell a_\ell \Ex_\ell b_\ell}  \le \sqrt{\epsilon}.
\end{equation}

Using \eqref{eq:NDLB-diff}, \eqref{eq:NDLB-summand-1},   \eqref{eq:NDLB-summand-2-bd}, and the monotonicity of $\hc$,  we have
 \begin{align*}
 \IC_\mu(f,0) - \IC_\mu(\pi)
 &\le \IC_\mu(\pi') - \IC_\mu(\pi) \le  2 |\cX \times \cY| \hc(\epsilon) + 2 |\cX \times \cY| \hc(\sqrt{\epsilon}) \le 4|\cX \times \cY| \hc(\sqrt{\epsilon}).  \qedhere
 \end{align*}
\end{proof}

\subsubsection{Proof of Proposition~\ref{prop:ext_XOR}}\label{sec:proof:ext_XOR}
\restate{Proposition~\ref{prop:ext_XOR}}{
Let $\mu$ be the distribution defined as
\[
 \mu = \begin{array}{|c|c|} \hline 1/2 & 0 \\\hline 0 & 1/2 \\\hline \end{array} \, .
\]
Then $\IC_\mu^\ext(\XOR,\epsilon) \geq \IC_\mu^\ext(\XOR,0) - 3\epsilon$.
}
\begin{proof}[Proof of Proposition~\ref{prop:ext_XOR}]
 The distribution $\mu$ is supported on the inputs $(0,0),(1,1)$, on which the output is~$0$. It is easy to check (and follows from the analysis below) that $\IC_\mu^\ext(\XOR,0) = 1$, since at the end of any protocol that performs $[\XOR,0]$, we know whether the input is $(0,0)$ or $(1,1)$.

Consider a protocol $\pi$ having at most $\epsilon$~error on every input, where $\epsilon \leq 1/3$. Let $\cL_z$ be the set of transcripts  on which the output is~$z$; Every transcript is either in $\cL_0$ or $\cL_1$.

For each transcript $t$   achievable from the initial distribution, the distribution of $XY|t$ is of the form $\begin{array}{|c|c|} \hline p & 0 \\\hline 0 & 1-p \\\hline \end{array}$ for some $p = p(t)$. Bayes' law shows that
\[
 \Pr[t|00] = \frac{\Pr[00|t] \Pr[t]}{\Pr[00]} = 2p(t)\Pr[t], \quad
 \Pr[t|11] = \frac{\Pr[11|t] \Pr[t]}{\Pr[11]} = 2(1-p(t))\Pr[t].
\]
For each transcript $t$, the rectangle property says $\Pr[t|00] \Pr[t|11] = \Pr[t|10] \Pr[t|01]$. Hence
\[
 \frac{\Pr[t|01] + \Pr[t|10]}{2} \geq \sqrt{\Pr[t|01]\Pr[t|10]} = \sqrt{\Pr[t|00]\Pr[t|11]} = 2 \sqrt{p(t)(1-p(t))}\Pr[t].
\]

The protocol $\pi$ has distributional error at most~$\epsilon$, and so
\[
\Pr[\cL_1] =\sum_{t \in \cL_1} \Pr[t] \le\epsilon, \quad{and}\quad \Pr[\cL_0] = \sum_{t \in \cL_0} \Pr[t] \geq 1-\epsilon.
\]
On the other hand, since $\pi$ has point-wise error at most $\eps$, we have
\begin{equation}
\label{eq:xor_L0_AMGM}
 \sum_{t \in \cL_0} \sqrt{p(t)(1-p(t))} \Pr[t] \leq
 \frac{1}{2} \sum_{t \in \cL_0} \frac{\Pr[t|01] + \Pr[t|10]}{2} \leq \frac{\epsilon}{2}.
\end{equation}
Finally,
\[
 I(XY;\Pi) = H(XY) - H(XY|\Pi) = 1 - \sum_t \Pr[t] h(p(t)).
\]

Let $T$ be a random transcript conditioned on belonging to $\cL_0$, and consider the random variable $P \defeq p(T)$. On the one hand,
\[
 1 - I(XY;\Pi) = \sum_t \Pr[t] h(p(t)) \leq \Pr[\cL_0] \Ex[h(P)] + \Pr[\cL_1] \leq \Ex[h(P)] + \epsilon.
\]
On the other hand, by (\ref{eq:xor_L0_AMGM})
\[
 \Ex[\sqrt{P(1-P)}] \leq \frac{\epsilon}{2\Pr[\cL_0]} \leq \frac{\epsilon}{2(1-\epsilon)} \leq \epsilon,
\]
as we assumed $\epsilon \le 1/3$. Thus it suffices to verify that $\Ex[h(P)] \le 2 \eps$ for any random variable $P$ that takes values in $[0,1]$ and satisfies $\Ex[\sqrt{P(1-P)}]\le \eps$. Indeed this  would imply
\[
 1 - I(XY;\Pi) \leq \Ex[h(P)] + \epsilon \le 3\epsilon,
\]
alternatively, $\IC^\ext_\mu(\XOR,\epsilon) \geq 1 - 3\epsilon$ for all $\epsilon \le 1/3$, which in turn shows that $\IC^\ext_\mu(\XOR,0) = 1$.

Apply the change of variable $Q = \sqrt{P(1-P)}$, so that the assumption  simplifies to  $\Ex[Q] \leq \epsilon$; note that $0 \leq Q \leq 1/2$, and $P = (1 \pm \sqrt{1-4Q^2})/2$. Since $h(P) = h(1-P)$, we conclude that
\[
 \Ex[h(P)] = \Ex[\phi(Q)], \text{ where } \phi(Q) = h\left( \frac{1+\sqrt{1-4Q^2}}{2} \right).
\]
It is routine to check that the function $\phi$ is monotonously  increasing and strictly convex. Since $\phi$ is continuous and the domain of $Q$ is restricted to $[0,1/2]$, the maximum of $\Ex[\phi(Q)]$ under the constraint $\Ex[Q] \leq \epsilon$ is achieved\footnote{Prokhorov's theorem.}. Since $\phi$ is increasing, the maximum value of $\Ex[\phi(Q)]$ is achieved when $\Ex[Q] = \epsilon$. Since $\phi$ is strictly convex, the maximum value of $\Ex[\phi(Q)]$ is achieved on a measure supported on the endpoints $0,1/2$. Thus this measure must be $\Pr[Q=1/2] = 2\epsilon$ and $\Pr[Q=0] = 1 - 2\epsilon$. So
\[
 \Ex[h(P)]= \Ex[\phi(Q)] \leq (1-2\epsilon) \phi(0) + 2\epsilon \phi(1/2) = 2\epsilon.   \qedhere
\]
\end{proof}
%
%
%
%

\subsection{Information complexity with distributional error}\label{sec:proof:distributional-error}

\restate{Theorem~\ref{thm:bd-IC-mu--distributional}}{
Let $\mu$ be a probability measure on $\cX \times \cY$, and let $f\colon \cX \times \cY \to \cZ$ satisfy $\IC_\mu(f,\mu,0)>0$.  We have 
\[
 \IC_\mu(f,\mu,0) - 4 |\cX| |\cY| \hc(\sqrt{\epsilon/\alpha}) \le \IC_\mu(f,\mu,\epsilon) \le  \IC_\mu(f,\mu,0) - \frac{\alpha^2}{4} \entf\left(\epsilon \alpha/4 \right) + 3\epsilon \log |\cX \times \cY|,
\]
where $\alpha = \min_{xy \in \supp \mu} \mu(x,y)$.
}

\begin{proof}[Proof of Theorem \ref{thm:bd-IC-mu--distributional}] 

\textbf{Lower bound:} The proof is almost identical to the proof of Theorem~\ref{thm:non-distri-error-lowerbd}, however now we start from a distribution $\mu$ that possibly does not have  full support.  Consider a protocol $\pi$ that performs $[f,\mu, \epsilon]$, and define  $z_\ell$ and $\mu_\ell$ as in the proof of Theorem~\ref{thm:non-distri-error-lowerbd}. Now the new protocol  $\pi'$ that performs $[f, \mu, 0]$, is defined similar to the one in the proof of  Theorem \ref{thm:non-distri-error-lowerbd} with the only difference that the verification is only performed on the set
\[ \Omega'_\ell \defeq \{(x,y) : f(x,y)\neq z_\ell\} \cap \supp\mu. \]
Obviously $\pi'$ solves $[f,\mu,0]$. Note that $\pi$ has point-wise error at most $\eps/\alpha$ on every point in $\supp\mu$. Thus the same analysis of  Theorem~\ref{thm:non-distri-error-lowerbd} shows
\[
\IC_\mu(f,\mu,0) - \IC_\mu(\pi)  \le \IC_\mu(\pi') - \IC_\mu(\pi) \le  4|\cX \times \cY| \hc(\sqrt{\epsilon/\alpha}).   
\]

\textbf{Upper bound:}  For every $z \in \cZ$, let $\cX_z$ denote the set of all $x \in \cX$ such that for some $xy \in \supp\mu$, we have $f(x,y)=z$. Similarly let $\cY_z$ denote the set of all $y \in \cY$ such that for some $xy \in \supp\mu$, we have $f(x,y)=z$. The assumption $\IC_\mu(f,\mu,0)>0$ implies the existence of  distinct $z_1,z_2 \in \cZ$ such that either  $\cX_{z_1} \cap \cX_{z_2} \neq \emptyset$ or  $\cY_{z_1} \cap \cY_{z_2} \neq \emptyset$, otherwise,  Alice and Bob can exchange  the unique values of $z$ determined by their inputs, and since with probability $1$, these two values coincide, they can perform $[f,\mu,0]$ with zero information cost. Hence without loss of generality assume  there exists $x_0\oy, x_1\oy \in \supp\mu$  such that $f(x_0,\oy) \neq f(x_1,\oy)$ and $\mu(x_0\oy) \ge \mu(x_1\oy)$. We will apply Lemma~\ref{lem:Given-good-transcripts}. Consider a protocol $\pi$ with transcript $\Pi$ that performs $[f,\mu,0]$, and define the set of transcripts 
$$\cL \defeq \{t \ | \ \Pr[x_0\oy|t] \ge \Pr[x_0\oy]/2\},$$
and note that 
$$\Pr[x_0\oy]=\sum_t  \Pr[x_0\oy|t]\Pr[\Pi=t] \le \Pr[\Pi \in \cL] + \Pr[\Pi \not\in \cL] \frac{\Pr[x_0\oy]}{2},$$
which implies  $\Pr[\Pi \in \cL]  \ge \frac{\Pr[x_0\oy]}{2} \ge \frac{\alpha}{2}$. Note that the protocol $\pi'$ defined in Figure \ref{fig:modified_protocol} performs $[f,\mu,\eps]$. Furthermore we can set $C_1=C_2=\alpha/2$ and $\delta=0$, to obtain 
$$\IC_\mu(\pi') \le\IC_\mu(\pi) - \frac{\alpha^2}{4} \entf\left( \frac{\epsilon \alpha}{4} \right) + 3\epsilon \log |X \times Y|,$$
for $\eps \le 1/2$. As $-\frac{\alpha^2}{4} \entf\left(\epsilon \alpha/4 \right) + 3\epsilon \log |\cX \times \cY| \ge 0$ for $\epsilon \ge 1/2$, this finishes the proof for all $0\le \epsilon \le 1$.
\end{proof}

%
%

\subsection{Non-distributional prior-free information cost}\label{sec:proof:non-distributional-ub}

In this section we  prove Theorem~\ref{thm:non-distributional-ub}, that is
\[ \IC(f,\epsilon) \leq \IC(f,0) - \Omega(h(\epsilon)).\]
First we present some lemmas, and the proof of Theorem~\ref{thm:non-distributional-ub} will appear at the end of this section.

While Theorem~\ref{thm:upper-bd-IC-mu-eps} does not give a uniform bound on the parameters $C,\epsilon_0$ for every distribution $\mu$, it does for distributions in which there exist two elements with different outputs, that are in the same row or column and whose probabilities are $\Omega(1)$. We will show that for any non-constant function, the worst distribution is of this form; this might be of independent interest.

We start with the following simple lemma.

\begin{lemma} \label{lem:blocks}
 Let $f\colon \cX \times \cY \to \cZ$. Suppose that $\supp \mu \subseteq \bigcup_i \cX_i \times \cY_i$, where the $\cX_i$ and the $\cY_i$ are disjoint. Then
\[
 \IC_\mu(f,0) = \sum_i \mu(\cX_i \times \cY_i) \IC_{\mu|_{\cX_i \times \cY_i}}(f|_{\cX_i \times \cY_i}).
\]	
\end{lemma}
\begin{proof}
 The upper bound is easy to see: the players exchange which block they are in, and assuming that they are in the same block, they run an almost optimal protocol for that block. If they are not in the same block, then they exchange inputs, but this happens with probability zero.

 In the other direction, let $J$ be the block in which Alice's input lies. Since the value of $J$ is determined by the value of $X$, for a protocol $\pi$ with transcript $\Pi$, we have
\[
I(Y;\Pi|X) =    I(Y;\Pi|XJ) = \sum_j \Pr[J=j] I(Y;\Pi|X,J=j) = \sum_j \mu(\cX_j \times \cY_j) I(Y;\Pi|X,J=j).
\]
With probability~$1$, $J$ is also the block in which Bob's input lies, and so
\[
 \IC_\mu(\pi) = \sum_j \mu(\cX_j \times \cY_j) [I(X;\Pi|Y,J=j) + I(Y;\Pi|X,J=j)] \geq \sum_j \mu(\cX_j \times \cY_j) \IC_{\mu|_{\cX_j \times \cY_j}}(f|_{\cX_j \times \cY_j}). \qedhere
\]
\end{proof}

We can therefore restrict our attention (for now) to distributions based on a single block. 
The crucial observation is the following.

\begin{lemma} \label{lem:improvement}
 Let $f\colon \cX \times \cY \to \cZ$, and let $\mu$ be a distribution such that $f$ is constant on its support, each atom in the support has probability at least $\alpha$, and the marginals of the support are $\cX,\cY$. If $f$ is not constant then there is a distribution $\nu$ such that $\IC_\nu(f,0) \geq \IC_\mu(f,0) + C(\alpha)$, where $C(\alpha) > 0$ depends only on $\alpha,|\cX|,|\cY|$.
\end{lemma}

\begin{proof}
Let $(x_0,y_0)$ be any point not in the support of $\mu$ such that $f(x_0y_0)$ is different from the constant value of $f$ on $\supp\mu$. Since the marginals of the support are $\cX,\cY$ and every atom in the support has probability at least $\alpha$, we see that $\Pr[X=x_0],\Pr[Y=y_0] \geq \alpha$.

Let $\nu=\eps \delta_{x_0y_0}+(1-\eps)\mu$, where $\eps$ is a parameter to be determined later, and $\delta_{x_0y_0}$ denotes the Dirac measure concentrated on the point $(x_0,y_0)$. Note that $X'Y' \sim\nu$ can be sampled in the following manner. First we pick $XY \sim \mu$ and an independent Bernoulli random variable $B$ with  $\Pr[B=1] = \epsilon$. Then
\[
 X'Y' =
 \begin{cases}
 	XY & \text{if } B = 0, \\
 	x_0y_0 & \text{if } B = 1.
 \end{cases}
\]	
Let $\pi$ be a protocol that performs the task $[f,0]$, and let $\Pi_{xy}$ denote the transcript of this protocol when it is run on the input $xy$. Note that with probability $1$, the value of $B$ is determined by the value of $X'Y'$, and thus
\begin{align*}
  I(X';\Pi_{X'Y'} | Y') &=  I(X'B;\Pi_{X'Y'} | Y')=I(B;\Pi_{X'Y'} | Y')+I(X';\Pi_{X'Y'} | Y'B)\\ 
  &=I(B;\Pi_{X'Y'} | Y')+(1-\eps) I(X;\Pi_{XY} | Y).
\end{align*}
Moreover, since $f(x_0,y_0)$ is different from the constant value of $f$ on the support of $\mu$, the value of $B$ is determined by $\Pi_{X'Y'}$. Thus $I(B;\Pi_{X'Y'} | Y')=H(B|Y')$, and
\[ I(X';\Pi_{X'Y'} | Y')= H(B | Y')+(1-\eps) I(X;\Pi_{XY} | Y). \]

To lower-bound $H(B | Y')$, note that
\[
 \Pr[B=1|Y'=y_0] = \frac{\Pr[B=1,Y'=y_0]}{\Pr[Y'=y_0]} = \frac{\epsilon}{(1-\epsilon)\Pr[Y=y_0]+\epsilon} \geq \epsilon,
\]
 and on the other hand,
\[
 \Pr[B=1|Y'=y_0] \leq \frac{\epsilon}{(1-\epsilon)\alpha+\epsilon},
\]
 which for $\epsilon \le \sqrt{\alpha}/2$ will be at most $1-\epsilon$. Since $\Pr[Y'=y_0] = (1-\epsilon) \Pr[Y=y_0] + \epsilon \geq \alpha$, we conclude that $H(B|Y') \ge  \alpha h(\epsilon)$. We deduce that
\[ I(X';\Pi_{X'Y'} | Y') \ge \alpha h(\eps) +(1-\eps) I(X;\Pi_{XY} | Y) \ge I(X;\Pi_{XY} | Y) + \alpha h(\eps) - \eps \log |\cX \times \cY|. \]
 The gain is
\[
 I(X';\Pi_{X'Y'} | Y') - I(X;\Pi_{XY} | Y) \geq \alpha \eps \log \frac{1}{\eps} - \eps \log |\cX \times \cY|  = \left(\alpha \log \frac{1}{\eps} - \log |\cX \times \cY|\right) \eps,
\]
 and so when $\eps \leq \eps_0 \defeq |\cX \times \cY|^{-2/\alpha}$, the gain is at least $\eps \log |\cX \times \cY|$. Taking $\eps = \min(\eps_0,\sqrt{\alpha}/2)$, we obtain a constant $C(\alpha) > 0$, depending on $|\cX \times \cY|$, such that
\[
 I(X';\Pi_{X'Y'} | Y') \geq I(X;\Pi_{XY} | Y) + C(\alpha),
\]
 and similarly $I(Y';\Pi_{X'Y'} | X') \geq I(X;\Pi_{XY} | Y) + C(\alpha)$. This shows that
\[
 \IC_\nu(f,0) \geq \IC_\mu(f,0) + 2C(\alpha). \qedhere
\]
\end{proof}

We obtain the following important consequence.

\begin{lemma} \label{lem:good-distributions}
 Let $f\colon \cX \times \cY \to \cZ$ be a non-constant function. There exist constants $c, \delta > 0$, depending only on the function $f$ and $|\cX|,|\cY|$, such that if $\IC_\mu(f,0) \geq \IC(f,0) - \delta$ then there exist points $P,Q$, on the same row or column, such that $\mu(P),\mu(Q) \geq c$ and $f(P) \neq f(Q)$.
\end{lemma}
\begin{proof}

 Call a distribution $\nu$ on $\cX \times \cY$ \emph{optimal} if $\IC(f,0) = \IC_\nu(f,0)$.  Braverman et al.~\cite{SelfRed} showed that $\IC_\nu(f,0)$ is continuous in $\nu$, and this implies that optimal distributions exist, and moreover the set of optimal distributions is closed. It is also convex, due to the concavity of $\IC_\nu(f,0)$ (see~\cite{MR3210776}).

 For a distribution $\nu$, let $\beta(\nu)$ be the maximal value $\beta$ such that there exist two points $P,Q$, on the same row or column, such that $\nu(P),\nu(Q) \geq \beta$ and $f(P) \neq f(Q)$.
 Note that $\beta(\nu)$ is continuous in $\nu$.

 Suppose that $\beta(\nu) = 0$. For $z \in \cZ$, let $\cX_z$ be the set of rows on which some point $P \in \supp \nu$ satisfies $f(P) = z$, and define $\cY_z$ analogously. We claim that the sets $\cX_z$ for $z \in \cZ$ are disjoint, similarly $\cY_z$  are disjoint. Indeed, if $x \in \cX_{z_1} \cap \cX_{z_2}$, then the row $x$ contains two points $P,Q$ in the support such that $f(P) \neq f(Q)$, and so $\beta(\nu) > 0$. Next we show that  $\supp \nu \subseteq \bigcup_z \cX_z \times \cY_z$. Indeed  if $P \in \cX_{z_1} \times \cY_{z_2}$ is in the support of $\nu$, and $f(P) \neq z_1$, then there exists some point $Q$ on the same row as $P$ is in the support and satisfies $f(Q) = z_1$, showing that $\beta(\nu) > 0$; a similar conclusion is reached if $f(P) \neq z_2$.

 Consider now one of the blocks $\cX_z \times \cY_z$. Lemma~\ref{lem:improvement} shows that we can modify the component of $\nu$ on that block so as to increase the information complexity, and Lemma~\ref{lem:blocks} shows that this increases the information complexity over the entire domain. We conclude that $\nu$ is not optimal.

 For $\rho \geq 0$, let $O_\rho = \{ \nu : \IC_\nu(f,0) \geq \IC(f,0) - \rho \}$. Continuity of $\IC_\nu(f,0)$ shows that $O_\rho$ is closed. We define $b(\rho) = \inf \{ \beta(\nu) : \nu \in O_\rho \}$; since $\beta$ is continuous and $O_\rho$ is closed, the infimum is achieved. In view of the preceding paragraph, $b(0) > 0$. Continuity of $\beta(\nu)$ and $\IC_\nu(f,0)$ shows that $b(\rho)$ is continuous as well, and so $b(\delta) > 0$ for some $\delta > 0$. The proof is complete by taking $c = b(\delta)$.
\end{proof}

We can now apply Theorem~\ref{thm:upper-bd-IC-mu-eps} to deduce that $\IC(f,\epsilon) \leq \IC(f,0) - \Omega(h(\epsilon))$.

\restate{Theorem~\ref{thm:non-distributional-ub}}{
 If $f\colon \cX \times \cY \to \cZ$ is non-constant then
\[ \IC(f,\epsilon) \leq \IC(f,0) - \Omega(h(\epsilon)), \]
 where the hidden constant depends on $f$.
}
\begin{proof}
 Let $c,\delta$ be the parameters from Lemma~\ref{lem:good-distributions}. For a distribution $\mu$, either $\IC_\mu(f,0) \leq \IC(f,0) - \delta$ or Theorem~\ref{thm:upper-bd-IC-mu-eps} shows that $\IC_\mu(f,\epsilon) \leq \IC_\mu(f,0) - (c^3/64) h(\epsilon) \leq \IC(f,0) - (c^3/64) h(\epsilon)$ for all $\epsilon \le \epsilon_0$ where $\epsilon_0$ depends only on $c$ and $|\cX \times \cY|$. Choose $\epsilon$ sufficiently enough such that $(c^3/64) h(\epsilon) \le \delta$ and $\epsilon \le \epsilon_0$, we conclude in both cases that $\IC_\mu(f,\epsilon) \leq \IC(f,0) - \Omega(h(\epsilon))$.
\end{proof}

\subsection{A characterization of trivial measures}   \label{sec:proof-trivial-measure}
First we present the proof of the external case, i.e. Theorem~\ref{thm:external-trivial}, as it is simpler.

\restate{Theorem~\ref{thm:external-trivial}}{
 Let $f\colon \cX \times \cY \to \cZ$ be an arbitrary function, and $\mu$ a distribution on $\cX \times \cY$.
 The distribution $\mu$ is external-trivial iff it is strongly external-trivial iff it is structurally external-trivial.	
}
\begin{proof}[Proof of Theorem~\ref{thm:external-trivial}]
 \proofpar{If $\mu$ is external-trivial then $\mu$ is structurally external-trivial.} Suppose that $\mu$ is external-trivial but not structurally external-trivial. We will reach a contradiction.

 We start by showing that if $\mu$ is external-trivial then $f$ has to be constant on the support of $\mu$. Indeed, suppose that the protocol $\pi$ computes $f$ correctly, and denote by $\Pi$ the transcript of $\pi$. The data processing inequality shows that
\[
 I(\Pi;XY) \geq I(\Pi;f(XY)) = H(f(XY)) - H(f(XY)|\Pi) = H(f(XY)).
\]
 This shows that $\mu$ can only be external-trivial if $H(f(XY)) = 0$, that is, if $f$ is constant on the support of $\mu$. From now, we assume that this is indeed the case.

 Let $ab$ be an arbitrary point in the support of $\mu$, and let $c = f(ab)$. Since $\mu$ is not structurally external-trivial, there must be some input $x_0y_0 \in S_A \times S_B$ for which $f(x_0y_0) \neq c$. Note that $x_0y_0$ is not in the support of $\mu$. Since $x_0 \in S_A$, $x_0y_1$ is in the support of $\mu$ for some $y_1 \in S_B$. Similarly, $x_1y_0$ is in the support of $\mu$ for some $x_1 \in S_A$.

 Since $\mu$ is external-trivial, there is a sequence $\pi_n$ of protocols computing $f$ correctly on every input such that $I(XY;\Pi_n) \to 0$, where $XY \sim \mu$. We think of $\pi_n$ also as a distribution over transcripts $t$. Since $f(XY)=c$ with probability~$1$, if $\pi_n(t) > 0$ then the transcript~$t$ indicates that the output is~$c$.
 Let $p_n$ be the joint distribution of $X,Y,t$. Recall that $D(p_n(x,y,t)\|\mu(x,y) \pi_n(t)) = I(XY;\Pi_n)$, hence $D(p_n(x,y,t)\|\mu(x,y) \pi_n(t)) \to 0$.

 For two distributions $\mu$ and $\nu$ on a finite space, Pinsker's inequality states that $D(\mu || \nu) \ge \frac{1}{2} \|\mu - \nu\|_1^2$. This implies that $\|p_n(x,y,t)-\mu(x,y)\pi_n(t)\|_1 \to 0$. On the other hand, for every transcript $t$ appearing with positive probability, either $p_n(x_0,y_1,t) = 0$ or $p_n(x_1,y_0,t) = 0$: otherwise $p_n(x_0,y_0,t) > 0$ (due to the rectangular property of protocols), contradicting the correctness of $\pi_n$ (since $f(x_0y_0) \neq c$). Therefore
\begin{equation*}
 |\mu(x_0,y_1) \pi_n(t) - p_n(x_0,y_1,t)| + |\mu(x_1,y_0) \pi_n(t) - p_n(x_1,y_0,t)| \geq \pi_n(t) \min(\mu(x_0,y_1),\mu(x_1,y_0)).
\end{equation*}
 Summing over all transcripts having positive probability, we deduce that
\begin{equation*}
 \|p_n(x,y,t)-\mu(x,y)\pi_n(t)\|_1 \geq \sum_t \pi_n(t) \min(\mu(x_0,y_1),\mu(x_1,y_0)) = \min(\mu(x_0,y_1),\mu(x_1,y_0)),
\end{equation*}
 contradicting our assumption that $\|p_n(x,y,t)-\mu(x,y)\pi_n(t)\|_1 \to 0$.

 \proofpar{If $\mu$ is structurally external-trivial then $\mu$ is strongly external-trivial.} Consider the following protocol. Alice tells Bob whether her input is in $S_A$. Bob tells Alice whether his input is in $S_B$. If the input is in $S_A \times S_B$, then the output is known. Otherwise, the players reveal their inputs (but this happens with probability zero). It's not difficult to check that this protocol has zero external information cost.

 \proofpar{If $\mu$ is strongly external-trivial then $\mu$ is external-trivial.} This is obvious.
\end{proof}

We comment that our proof gives an explicit lower bound on $\IC^\ext_\mu(f,0)$ whenever $\mu$ is not external-trivial.

Next we present the proof of  Theorem~\ref{thm:internal-trivial}, showing that all our definitions of internal triviality are equivalent. As before, we can get an explicit lower bound on $\IC_\mu(f,0)$ whenever $\mu$ is not internal-trivial.

\restate{Theorem~\ref{thm:internal-trivial}}{
Let $f\colon \cX \times \cY \to \cZ$ be an arbitrary function, and $\mu$ a distribution on $\cX \times \cY$. The distribution $\mu$ is internal-trivial iff it is strongly internal-trivial iff it is structurally internal-trivial.	
}
\begin{proof}[Proof of Theorem~\ref{thm:internal-trivial}]
 \proofpar{If $\mu$ is internal-trivial then $\mu$ is structurally internal-trivial.} Suppose that $\mu$ is internal-trivial but not structurally internal-trivial. We will reach a contradiction.

 Since $\mu$ is internal-trivial, there is a sequence of protocols $\pi_n$ such that $I(X;\Pi_n|Y) + I(Y;\Pi_n|X) \to 0$. In particular, $I(X;\Pi_n|Y),I(Y;\Pi_n|X) \to 0$. Moreover, for every $x \in S_A$ and for every $y \in S_B$, $I(X;\Pi_n|Y=y),I(Y;\Pi_n|X=x) \to 0$.

 Let $p_n(x,y,t)$ be the joint probability of the input and of the transcript of $\pi_n$ being~$t$. We also think of $\pi_n$ as a distribution over transcripts.
 As in the proof of Theorem~\ref{thm:external-trivial}, using Pinsker's inequality we deduce that for all $y \in S_B$, $\|p_n(x,t|y) - \mu(x|y) \pi_n(t|y)\|_1 \to 0$, and so for all $y \in S_B$,
\begin{equation*}
 B_y \defeq \sum_{x,t} |p_n(x,y,t) - \mu(x,y) \pi_n(t|y)| \to 0.
\end{equation*}
 Similarly, for all $x \in S_A$ we have
\begin{equation*}
 A_x \defeq \sum_{y,t} |p_n(x,y,t) - \mu(x,y) \pi_n(t|x)| \to 0.
\end{equation*}

 According to Lemma~\ref{lem:structurally-internal-trivial}, there exists a connected component $C$ of $G_\mu$ such that $f$ is not constant on $C_A \times C_B$. Suppose first that there is an edge $(P,Q)$ on which $f$ is not constant. Without loss of generality, assume $P = (a,y_0)$ and $Q = (a,y_1)$. Thus
\begin{equation*}
 \sum_t |p_n(a,y_0,t) - \mu(a,y_0) \pi_n(t|a)| + |p_n(a,y_1,t) - \mu(a,y_1) \pi_n(t|a)| \to 0.
\end{equation*}
 On the other hand, for each transcript $t$ either $p_n(a,y_0,t) = 0$ or $p_n(a,y_1,t) = 0$, since $f(ay_0) \neq f(ay_1)$. Thus
\begin{multline*}
 \sum_t |p_n(a,y_0,t) - \mu(a,y_0) \pi_n(t|a)| + |p_n(a,y_1,t) - \mu(a,y_1) \pi_n(t|a)| \geq \\ \sum_t \pi_n(t|a) \min(\mu(a,y_0),\mu(a,y_1)) = \min(\mu(a,y_0),\mu(a,y_1)),
\end{multline*}
 contradicting the assumption that the left-hand side tends to zero.

 Suppose next that $f$ is constant across all edges (and so on the entire connected component), say $f(x,y) = c$ for all $(x,y) \in C$. Since $f$ is not monochromatic on $C_A \times C_B$, there must exist a point $P \in C_A \times C_B$ such that $f(P) \neq c$. There must be points $P_A,P_B \in \supp \mu$ with the same row and column (respectively) as $P$. Since $P_A,P_B$ are in the same connected component, there is some path $P_A = Q_0, Q_1, \ldots, Q_m = P_B$ connecting them: for every $i < m$, $Q_i,Q_{i+1}$ are either in the same row or in the same column. We can assume that $m \leq M \defeq |\cX| + |\cY|$. No transcript can have positive probability for both $Q_0$ and $Q_m$, since otherwise it would have positive probability for $P$ as well, and this cannot happen since $f(Q_0) = f(Q_m) = c$ while $f(P) \neq c$.

 Let $t$ be any transcript satisfying $p_n(Q_0,t) > 0$. Since $p_n(Q_m,t) = 0$, there must be an index $i$ such that $p_n(t|Q_i) - p_n(t|Q_{i+1}) \geq p_n(t|Q_0)/m \geq p_n(t|Q_0)/M$. Assume without loss of generality that $Q_i = (a,y_0)$ and $Q_{i+1} = (a,y_1)$. The contribution of $t$ to $A_a$ is
\begin{multline*}
 |\mu(a,y_0) \pi_n(t|a) - p_n(a,y_0,t)| + |\mu(a,y_1) \pi_n(t|a) - p_n(a,y_1,t)| = \\
 \mu(a,y_0) |\pi_n(t|a) - p_n(t|a,y_0)| + \mu(a,y_1) |\pi_n(t|a) - p_n(t|a,y_1)| \geq \\
 \frac{\min(\mu(a,y_0),\mu(a,y_1))}{M} p_n(t|Q_0) \geq
 \frac{\min(\mu(a,y_0),\mu(a,y_1))}{M} p_n(Q_0,t),
\end{multline*}
 using the triangle inequality in the form $|\alpha - \gamma| + |\gamma - \beta| \geq |\alpha - \beta|$.

 Denoting by $\delta$ the minimum of $\mu(x,y)$ over the support of $\mu$, we conclude that $\sum_x A_x + \sum_y B_y$ is at least
\begin{equation*}
 \sum_t \frac{\delta}{M} p_n(Q_0,t) = \frac{\delta}{M} \mu(Q_0) \geq \frac{\delta^2}{M},
\end{equation*}
 contradicting our assumption that $\sum_x A_x + \sum_y B_y \to 0$.

 \proofpar{If $\mu$ is structurally internal-trivial then $\mu$ is strongly internal-trivial.} Consider the following protocol. Alice tells Bob which block $\cX_i$ her input belongs to. Bob tells Alice which block $\cY_i$ his input belongs to. If the input is in $\cX_i \times \cY_i$, then the output is known. Otherwise, the players reveal their inputs (but this happens with probability zero). It's not difficult to check that this protocol has zero internal information cost.

 \proofpar{If $\mu$ is strongly internal-trivial then $\mu$ is internal-trivial.} This is obvious.	
\end{proof}

\section{Parametrization of all distributions as product distributions}\label{sec:parametrization}

In Section~\ref{sec:randomwalk} we discussed how a communication protocol can be interpreted as a random walk on the set of distributions on $\cX \times \cY$. Every time a player sends a signal, we update the underlying distribution based on the information provided by the sent signal. These updates are by scaling either the $\cX$ marginal or the $\cY$ marginal of the distribution. This restricted way in which the underling distribution can be updated will allow us to parametrize  the set of all reachable distributions from a specific distribution $\omu$ in such a way that  the changes are captured by product measures. First note that each reachable distribution $\omu'$ can be identified by the constants that multiplied $\omu$ to obtain $\omu'$.

To formalize this intuition, we have the following definition.
\begin{definition}
	For two  distributions $\mu, \nu \in \Delta(\cX, \cY)$, define
	\begin{equation}    \label{eq:def-HadamardProd}
	\mu \odot \nu := \frac{\mu \cdot \nu}{\langle \mu, \nu \rangle},
	\end{equation}
	where $\mu \cdot \nu$ is the usual point-wise product of the two measures.
\end{definition}
Clearly,  $\mu \odot \nu \in \Delta(\cX, \cY)$ unless $\langle \mu, \nu \rangle = 0$, in which case the product is undefined.
For our purposes, we will consider decompositions of the form $\omu= \nu \odot \mu$, where $\mu$ is a \emph{product measure}. The statement ``$\omu$ is a distribution obtained from $\nu$ by scaling its rows and columns'' is equivalent to ``there exists a product measure $\mu$ such that $\omu = \nu \odot \mu$''. Note that if $\mu$ is the uniform distribution, then $\nu = \mu \odot \nu$ for all distributions $\nu$.

Let $\omu$ be  the prior distribution on $\cX \times \cY$ in a communication protocol.
We fix a decomposition $\omu=\nu \odot \mu$, where $\mu$ is a product distribution.
For every distribution $\omu'$ reachable from $\omu$ there is a product distribution $\mu'$ such that $\omu'=\nu \odot \mu'$, for the same distribution $\nu$.
This follows from the fact that $\omu'$ is obtained from $\omu$ by scaling its rows and columns; therefore if we scale the rows and columns of $\mu$ by the same constants and then normalize it, we obtain the desired $\mu'$.
In such a $\omu=\nu\odot\mu$ decomposition,
$\omu$ is called the \emph{real} distribution, $\nu$ the \emph{reference distribution} and $\mu$ the \emph{pretend distribution}.

We would like to work with product distributions since they are simpler, and easier to analyze, as we will demonstrate in Section~\ref{sec:AND}. 
Therefore, we define a \emph{pretend random walk}, which is a random walk on pretend distributions, as opposed to the normal random walk presented in Section \ref{sec:randomwalk}, which we call the \emph{real random walk} to distinguish it from the pretend one.
It start from a product measure $\mu=(\mu^{\cX}, \mu^{\cY})$, where $\mu^{\cX}$ and $\mu^{\cY}$ are the $\cX$ and $\cY$ marginals of $\mu$.
At each step we either move by scaling  the $\Delta(\cX)$ marginal or the $\Delta(\cY)$ marginal. The transition in $\Delta(\cX)$ is performed by moving with probability $\lambda_0$ to $(\mu_0,\mu^{\cY})$ and with probability $\lambda_1$ to $(\mu_1, \mu^{\cY})$,
where $0 < \lambda_0, \lambda_1 < 1$, $\lambda_0 + \lambda_1 = 1$ and $\sum_{b=0,1} \lambda_i \mu_i = \mu^{\cX}$.
A step in the $\Delta(\cY)$ direction is performed similarly.

Every pretend random walk corresponds to a real random walk performed by some protocol. Given such a pretend random walk, and a reference distribution $\nu$, if we replace every distribution $\mu$ encountered in the random walk by $\nu \odot \mu$, and scale the transition probabilities, we obtain a real random walk performed by some protocol. Here $\nu$ can be any distribution such that $\nu \odot \mu$ is defined for every $\mu$ encountered in the protocol (e.g. if $\supp \nu$ includes the support of the initial distribution). The inverse transformation is also possible.

To formalize this idea, consider a pretend random walk step, from $\mu$ to $\mu_0$ and $\mu_1$ with transition probabilities $\lambda_0$ and $\lambda_1$, respectively. Fix a reference distribution $\nu$.
Then
\[
\nu \odot \mu =
\frac{\nu \cdot \mu}{\langle \nu, \mu \rangle} =
\sum_{b=0,1} \lambda_b \frac{\nu \cdot \mu_b}{\langle \nu, \mu \rangle} =
\sum_{b=0,1} \frac{\langle \nu, \mu_b \rangle}{\langle \nu, \mu \rangle} \lambda_b (\nu \odot \mu_b) =
\sum_{b=0,1} \overline{\lambda_b} (\nu \odot \mu_b)
\]
for the values 
\begin{equation} \label{eq:pretend-prob-conversion}
\overline{\lambda_b} = \frac{\langle \nu, \mu_b \rangle}{\langle \nu, \mu \rangle} \lambda_b.
\end{equation}
A calculation shows
\[
\sum_{b=0,1} \overline{\lambda_b} =
\sum_{b=0,1} \frac{\langle \nu, \mu_b \rangle}{\langle \nu, \mu \rangle} \lambda_b =
\frac{\langle \nu, \sum_{b=0,1} \lambda_b \mu_b \rangle}{\langle \nu, \mu \rangle} =
\frac{\langle \nu, \mu \rangle}{\langle \nu, \mu \rangle} =
1.
\]
Furthermore, if the pretend random walk step is performed in the $\Delta(\cX)$ direction, then $\nu \odot \mu_b$ is obtained by scaling the rows of $\mu$, and if in the $\Delta(\cY)$ direction, then by scaling the columns.
Therefore, there exists a real random walk step where we move from $\nu \odot \mu$ to $\nu \odot \mu_0$ and $\nu \odot \mu_1$ with probabilities $\overline{\lambda_0}$ and $\overline{\lambda_1}$ respectively.
The  conversion in the opposite direction, from the real world to the pretend world, is possible due to essentially the same calculations.

Let $\pi_0$ and $\pi_1$ be the two branches of the protocol $\pi$ corresponding to the value of the first bit that was sent.
Let $\omu$ be an input distribution that moves either to $\omu_0$ or to $\omu_1$ with probabilities $\overline{\lambda_0}$ and $\overline{\lambda_1}$, respectively.
The following equation regarding the concealed information,
\[
\CI_{\omu}(\pi) = \sum_{b=0,1} \overline{\lambda_b} \CI_{\omu_b}(\pi_b)
\]
translates to
\[
\CI_{\nu  \odot \mu}(\pi) 
= \sum_{b=0,1} \frac{\langle \nu, \mu_b \rangle}{\langle \nu, \mu \rangle} \lambda_b \CI_{\nu  \odot \mu_b}(\pi_b).
\]
Multiplying by ${\langle \nu, \mu \rangle}$ we get
\[
\CI_{\nu  \odot \mu}(\pi) {\langle \nu, \mu \rangle}
= \sum_{b=0,1} \lambda_b \langle \nu, \mu_b \rangle  \CI_{\nu  \odot \mu_b}(\pi_b),
\]

This motivates the following definition.

\begin{definition}   \label{def:SIM}
	Let $\nu$ be a fixed reference distribution. Define the \emph{scaled information} of a protocol $\pi$ with respect to a product distribution $\mu$ as
	\begin{equation}    \label{eq:def-SIM}
	\SIM_{\mu}(\pi) := \langle \nu, \mu \rangle \CI_{\nu \odot \mu}(\pi).
	\end{equation}
\end{definition}

Equation \eqref{eq:def-SIM} allows us to write
\begin{equation} \label{eq:pretend-dist-rand-step}
\SIM_{\mu}(\pi) = \lambda_0 \SIM_{\mu_0}(\pi_0) + \lambda_1 \SIM_{\mu_1}(\pi_1).
\end{equation}

Recall that $\CI$ is the expected amount of entropy that the players have concealed from each other by the end of the protocol.
To formally state this, let $\omu$ be a distribution over the inputs, $\pi$ some protocol and $\Pi$ the random variable representing the transcript of the protocol.
Let $\omu_{\Pi}$ be the random variable that represents the distribution over the inputs given the transcript $\Pi$, as defined in Section \ref{sec:randomwalk}. Then

\begin{equation} \label{eq:rand-walk-expectancy}
\CI_\omu(\pi) = \Ex_{\Pi}\left[H_{\omu_{\Pi}}(X|Y) + H_{\omu_{\Pi}}(Y|X)\right].
\end{equation}

We will translate \eqref{eq:rand-walk-expectancy} to a formula involving the pretend random walk.
Let $\omu = \nu \odot \mu$, and denote by $\mu_{\Pi}$ the pretend distribution where the pretend random walk ends if its associated protocol has the transcript $\Pi$. Or, in a more formal way, $\mu_{\Pi}$ is the distribution such that $\nu \odot \mu_{\Pi} = \omu_{\Pi}$.
Equation~\eqref{eq:def-SIM} implies
\begin{equation} \label{eq:SIM-Ex}
\SIM_\mu(\pi) 
=  \Ex_{\Pi} \langle \nu, \mu_{\Pi} \rangle \left[ H_{(\nu \odot \mu)_{\Pi}}(X|Y) + H_{(\nu \odot \mu)_{\Pi}}(Y|X) \right],
\end{equation}
where the probability for each transcript $\Pi$ is according to the pretend random walk and not to the real one.

One should ask: What is the probability of a transcript $t$ in the pretend random walk, given its probability $\overline{\lambda}$ in the real world?
The answer turns out to be very simple. Let $\omu^0,\dots, \omu^k$ be the real distributions encountered in the real random walk, where $\omu^0$ is the input distribution and $\omu^k=\omu_{t}$ is the last distribution encountered.
For all $1 \leq i \leq k$, let $\overline{\lambda^i}$ be the transition probability from $\omu^{i-1}$ to $\omu^i$ in the real random walk, so that $\overline{\lambda} = \overline{\lambda^1} \cdots \overline{\lambda^k}$.
Let $\mu^i$ be the pretend distribution associated with $\omu^i$ such that $\omu^i = \nu \odot \mu^i$ for all $i$ .
Then, the transition probability from $\mu^{i-1}$ to $\mu^i$ in the pretend world equals 
\[
\lambda^i = \frac{\langle \nu, \mu^{i-1} \rangle}{\langle \nu, \mu^i \rangle} \overline{\lambda^i},
\]
using the conversion in \eqref{eq:pretend-prob-conversion}.
Multiplying all together, we get that the probability of $t$ in the pretend world is
\[
\lambda 
= \prod_{i=1}^k \lambda^i
= \prod_{i=1}^k \frac{\langle \nu, \mu^{i-1} \rangle}{\langle \nu, \mu^i \rangle} \overline{\lambda^i}
= \frac{\langle \nu, \mu^0 \rangle}{\langle \nu, \mu^k \rangle} \overline{\lambda}.
\]
This equation also shows how one can derive \eqref{eq:SIM-Ex} from \eqref{eq:rand-walk-expectancy} by multiplying the equation by
$\langle \nu, \mu^0 \rangle$.

For more discussion on the parametrization by product distributions and its applications, see~\cite{DaganFilmus}.

\section{The analysis of the AND function} \label{sec:AND}

This section is mainly devoted to proving the only remaining case of Theorem~\ref{thm:AND-gap}, i.e. the lower bound on $\IC_\mu(\AND,\epsilon)$. This is presented below separately as Theorem~\ref{thm:LB-D}. Our general strategy for this proof was sketched in Section~\ref{sec:ANDfunction} following Theorem~\ref{thm:AND-gap}.

%

\paragraph{Preliminaries and notations.}

The section relies strongly on the parametrization of distributions as product distributions, as presented in Section~\ref{sec:parametrization}.
A real distribution is usually denoted as $\omu$, and it is usually decomposed as $\omu = \nu \odot \mu$, where $\nu$ is a symmetric reference distribution and $\mu$ a pretend distribution.
Pretend distributions are always product ones. We will use the shorthand notation $\mu = (p,q)$ for the product distribution in which $p = \mu(1,0) + \mu(1,1)$ and $q=\mu(0,1) + \mu(1,1)$.
The distribution $\omu$ will usually be assumed to be of full support, which in turn forces $\nu$ and $\mu$ to be so too.

We are usually going to be working in a pretend world, dealing with the pretend distributions, and keeping the reference distributions in the background. Furthermore, reference distributions are usually kept fixed. We regard protocols as pretend random walks, as presented in Section~\ref{sec:parametrization}.

Suppose that we run a protocol $\pi$ starting at a distribution $\omu = \nu \odot \mu$. As we explained in Section~\ref{sec:parametrization}, for each transcript $t$ of the protocol, there is a product distribution $\mu_t$ such that $\nu \odot \mu_t$ is the distribution of the players' inputs conditioned on the protocol terminating at the leaf $t$.
Let $\Pi$ be the random transcript of the pretend random walk associated with an execution of $\pi$ on input distribution $\omu$. Therefore, for any transcript $t$, $\Pr[\Pi=t]$ is the probability for the transcript $t$ in the pretend random walk, which might be different than the corresponding probability in the real random walk. Throughout this section our view of the protocol is only by the pretend random walk, therefore all random variable that correspond to $\Pi$ are assumed to be distributed according to the pretend random walk. Since $\mu_{\Pi}$, the pretend distribution on the random transcript $\Pi$, is a product distribution, it can be written as $\mu_{\Pi} = (\leafp, \leafq)$, where $\leafp,\leafq$ are random variables. We call $(\leafp,\leafq)$ the \emph{leaf distribution} of $\pi$. We define a crucial random variable, $\leafl = \max(\leafp,\leafq)$.

If $\pi$ is a zero-error protocol, then the leaf distribution is supported on product distributions of the form $(p,0)$, $(0,q)$ or $(1,1)$, since in order to know the AND of the two players' inputs we need to know that one of the players has input~$0$, or that both  inputs are~$1$.

Since we are concerned with almost-optimal protocol,  we would like to quantify optimality.
Given a protocol $\pi$, define its \emph{wastage} with respect to a distribution $\omu$ by
\[
\IW_\omu(\pi) = \IC_{\omu}(\pi) - \IC_{\omu}(\AND,0)
=  \CI_{\omu}(\AND,0)- \CI_{\omu}(\pi).
\]

\subsection{Stability results} \label{sec:AND-stability}

Braverman et al. \cite{MR3210776}, studying the complexity of the $\AND$ function, suggested a continuous protocol whose information complexity equals $\IC_\omu(\AND,0)$, called the \emph{buzzer protocol}. This protocol is defined differently for any input distribution $\omu$. Here we denote this protocol by $\pi^*$. The buzzer protocol is not a conventional communication protocol as it has access to a continuous clock, however, it can be viewed as a limit of a sequence of genuine protocols. The information complexity of the protocols in that sequence converges to that of the buzzer protocol, and their leaf distribution converges in distribution.

We start by presenting the leaf distribution of the buzzer protocol. We assume that the input reference distribution is symmetric; its importance will become apparent later on.

\begin{table}[!h]
	\caption{The leaf distribution of the buzzer protocol starting from $(p,q)$, where $p \geq q$.}
	\label{table:leaf-distri-offdiag-cont-protocol}
	\centering
	\begin{tabular}{c|c|c|c}
		\hline
		Distribution $\mu_\Pi$
		& $(p,0)$
		&
		\begin{tabular}{c}
			$(\ell,0)$, $(0,\ell)$ \\
			($p < \ell < 1$)
		\end{tabular}
		& $(1,1)$
		\\ \hline
		The probability to reach that distribution
		& $1-q/p$
		& $pq/\ell^3 \, \rmd\ell$
		& $pq$
		\\  \hline
	\end{tabular}
\end{table}

As it can be seen in Table~\ref{table:leaf-distri-offdiag-cont-protocol}, this is a mix of discrete probabilities and a continuous density. To verify that the above formulas are correct, we can convert the leaf distribution of the buzzer protocol as it is calculated in \cite{MR3210776} for the real random walk to its corresponding leaf distribution in the pretend random walk. The formulas that are discussed in  Section~\ref{sec:parametrization} can be used to calculate the appropriate scaling of the  probabilities as we convert the  real random walk to the pretend one.

There is also a second and more intuitive way to obtain these formulas.  This is done by considering a sequence of protocols that converges to the buzzer protocol. We describe the protocols in that sequence by their pretend random walk. The initial distribution in the pretend world of a protocol in that sequence is $(p,q)$, where $p,q \in \{ 0,\frac 1 n, \frac 2 n, \dots, 1 \}$. In each step, the pretend random walk moves to one of two adjacent grid points, each with probability half. If we are currently in a distribution $(\frac a n, \frac b n)$ where $a \geq b$, then the step moves to one of $(\frac a n, \frac{b+1}{n})$ and $(\frac a n, \frac{b-1}{n})$. Otherwise, the protocol moves to one of $(\frac{a+1}{n}, \frac{b}{n})$ and $(\frac{a-1}{n}, \frac{b}{n})$.

Therefore, starting at the point $(\frac a n, \frac b n)$ where $a \geq b$, the random walk moves in the $y$ axis, until it ends up either at $(\frac a n,0)$ or at $(\frac a n,\frac{a+1}{n})$. Since this walk is balanced, the probabilities to get to these points are $1-\frac{b}{a+1}$ and $\frac{b}{a+1}$, respectively.
Then, from that point the random walk moves in the $x$ axis, until it either gets to the point $(0, \frac{a+1}{n})$ or to $(\frac{a+1}{n},\frac{a+1}{n})$, with probabilities $\frac{1}{a+1}$ and $\frac{a}{a+1}$ respectively. Then again, it ends up either at $(\frac{a+1}{n},0)$ or at $(\frac{a+1}{n},\frac{a+2}{n})$, then at $(0,\frac{a+2}{n})$ or $(\frac{a+2}{n},\frac{a+2}{n})$ and continues this way, until it either gets to the point $(1,1)$, or to a point of the form $(0,\frac i n)$ or $(\frac i  n,0)$. Calculating the leaf distribution of each pretend random walk in that sequence, and taking the limit as $n \to \infty$, results in a leaf distribution, which equals that of the buzzer protocol, as will be explained below.

The buzzer protocol can also be defined similarly as a sequence of converging protocols, where for each protocol in the sequence, the real-world analogue of moving in the $y$ direction is performed whenever $\Pr[X=1] \geq \Pr[Y=1]$, while the analogue of moving in the $x$ direction is performed otherwise. In order for our limit protocol to behave identical to the buzzer protocol, we would like the region $\Pr[X=1] \geq \Pr[Y=1]$ to correspond to the region $p \geq q$. This is done by using a symmetric reference distribution. 

Next, we would like to show a stability result, proving that every protocol performing the task $[\AND,0]$ with nearly optimal information complexity is similar to the  buzzer protocol. We measure similarity in terms of the leaf distribution, and define the following potential function:
\begin{definition}\label{def:potential}
	Given a protocol $\pi$ for $[\AND,0]$, a constant $0 < c < 1$, and a pretend distribution $\mu$, let
	\[
	\Phi_{c,\mu}(\pi) = \Ex\left[((c-\leafl)_+)^2\right],
	\]
	where $(\cdot)_+ = \max \{ \cdot, 0 \}$.
	Denote $\Phi_{c,\mu} = \Phi_{c,\mu}(\pi^*)$,
	where $\pi^*$ is the buzzer protocol.   
\end{definition}

The following theorem shows that the value of the potential function is small for nearly optimal protocols. 

\begin{theorem} \label{thm:stability}
	Let $\omu$ be a full support distribution, and $\omu = \nu \odot \mu$ be its decomposition, where $\nu$ is a symmetric reference distribution and $\mu=(p,q)$ is the product pretend distribution.
	Assume that $c \le \max \{ p,q \}$.
	Let $\pi$ be a protocol performing $[\AND,0]$. Then
	\[
	\Phi_{c,\mu}(\pi) = O(\IC_\omu(\pi) - \IC_\omu(\AND,0)) = O(\IW_\omu(\pi)).
	\]
	The constant in the $O(\cdot)$ is uniform whenever $\nu(0,0),\nu(0,1),\nu(1,0),p,q$ are bounded away from $0$ and $1$. 
\end{theorem}

In order to prove this theorem, we measure how each performed step contributes both to the wastage and to the potential function. To measure the wastage, we work with $\SIM$ instead of $\IC$, as it is a more natural measure for this task.

\begin{lemma} \label{lem:stability-one-bit}
	Let $\omu$ be a full support distribution, and $\omu = \nu \odot \mu$ be its decomposition, where $\nu$ is a symmetric reference distribution and $\mu$ is the pretend distribution. Let $0 < c < 1$, and let $\pi$ be the protocol which behaves as follows:
	\begin{enumerate}
		\item
			One step of a pretend random walk is performed, which corresponds to one bit that is sent in the protocol.
		\item
			The pretend random walk that corresponds to the buzzer protocol is simulated from that point: assuming that after the first bit was sent the pretend distribution is $(p,q)$, let $\pi^*_{(p,q)}$ be the buzzer protocol for the input distribution $\nu \odot (p,q)$. Then, the pretend random walk that corresponds to $\pi^*_{(p,q)}$ is simulated (the value of $(p,q)$ is different for the case that the first bit equals 1, and when it equals 0).
	\end{enumerate}
	Then
	\[
	\Phi_{c,\mu}(\pi) - \Phi_{c,\mu} = O_{\nu}(\SIM_{\mu}(\AND,0) - \SIM_{\mu}(\pi)).
	\]
	The constant in the $O(\cdot)$ is uniform whenever $\nu(0,0),\nu(0,1),\nu(1,0),c$ are bounded away from 0 and 1.
\end{lemma}

The potential function of Definition~\ref{def:potential} is defined in that manner so that Lemma~\ref{lem:stability-one-bit} holds. Let us elaborate on this: assume that a protocol $\pi$ is defined as in this lemma, with a pretend input distribution of $(p,q)$. Assume that the first step moves from $(p,q)$ either to $(p+\delta)$ or to $(p-\delta)$ with equal probability.
Then 
\begin{align*}
\SIM_{(p,q)}(\pi) - \SIM_{(p,q)}(\AND,0) 
&= \frac 1 2  \SIM_{(p+\delta,q)}(\AND,0) + \frac 1 2  \SIM_{(p+\delta,q)}(\AND,0) - \SIM_{(p,q)}(\AND,0) \\
&\approx \frac{\delta^2}{2} \frac{\partial^2}{\partial p^2} \SIM_{(p,q)}(\AND,0).
\end{align*}
Thus, this difference has the same order of magnitude as $\delta^2$.
We would like the change in the potential function to have the same order.
Looking at the function $x^2$, it holds that
\[
\frac 1 2 (x+\delta)^2 + \frac 1 2 (x-\delta)^2 - x^2 = \frac {\delta^2}{2}.
\]
If a protocol $\pi$ moves according to the direction of the buzzer protocol, then $\pi$ is the same as $\pi^*$ and both differences are zero. Therefore, 
assume that $p > q$, and $\pi$ moves in the $x$ direction, whereas the buzzer protocol would have moved in the $y$ direction. 
Roughly speaking, the leaf distribution of $\pi$ is obtained from the leaf distribution of $\pi^*$ by splitting some of the mass around $\leafl \approx p$ between $\leafl \approx p - \delta$ and $\leafl \approx p + \delta$. Thus, 
$\Phi_{c,\mu}(\pi) - \Phi_{c,\mu}$ approximately has the order of magnitude of 
\[
\frac 1 2 (c - p - \delta)^2 + \frac 1 2 (c - p+\delta)^2 - (c-p)^2 = \frac {\delta^2}{2}.
\]
We chose $(c-p)_+^2$ instead of $(c-p)^2$ since Lemma~\ref{lem:completion-distribution} requires the buzzer protocol to have a value of zero. Indeed, by choosing $c$ carefully we can achieve this. 

We will prove Lemma~\ref{lem:stability-one-bit} using the following criterion.

\begin{lemma} \label{lem:concave}
	Let $\nu$ be a symmetric reference distribution, 
	and $C > 0$ a constant.
	Define $F(p,q) = C \SIM_{(p,q)}(\AND,0) + \Phi_{c,(p,q)}$.
	If for every $q$, $F(p,q)$ is concave as a function of $p$, and for every $p$,
	$F(p,q)$ is concave as a function of $q$, then Lemma~\ref{lem:stability-one-bit}
	holds, and the constant in the $O(\cdot)$ can be taken to be equal to $C$.
\end{lemma}

\begin{proof}
	Let $\pi$ be the protocol defined in Lemma~\ref{lem:stability-one-bit}, and let
	$\mu$ be its pretend input distribution.
	Assume that the pretend random walk of $\pi$ first moves from $\mu$ either to $\mu_0$ or to $\mu_1$, with probabilities $\lambda_0$ and $\lambda_1$. We assume this step is on the $x$-direction, thus, the first step is from $(p,q)$ to $(p_0,q)$ or $(p_1,q)$. The analysis for the case that this step it in the $y$-direction is similar.
	Let $0<c<1$.
	Then
	$\SIM_{(p,q)}(\pi) = \sum_b \lambda_b \SIM_{(p_b,q)}(\AND,0)$, and 
	$\Phi_{c,{(p,q)}} = \sum_b \lambda_b \Phi_{c,(p_b,q)}$.
	From concavity,
	\begin{align*}
	C \SIM_{{(p,q)}}(\AND,0) + \Phi_{c,{(p,q)}} &= F({p,q}) \geq \sum_b \lambda_b F(p_b,q) = \sum_b \lambda_b (C \SIM_{(p_b,q)}(\AND,0) + \Phi_{c,(p_b,q)}) \\
	&= C \SIM_{{(p,q)}}(\pi) + \Phi_{c,{(p,q)}}(\pi). \qedhere
	\end{align*}
\end{proof}

Thus, our focus would be proving that these concavity conditions hold for some value $C$. 
We proceed by calculating $\Phi_{c,(p,q)}$, assuming without loss of generality that $p \geq q$. One can see that whenever $p \geq c$, with probability 1 the leaf distribution of the buzzer protocol satisfies $\leafl \geq p \geq c$, and thus the potential function evaluates to $0$.
Consider the case $p < c$.
Using the leaf distribution, we obtain the formula
\[
\Phi_{c,(p,q)} = (1-q/p)(c-p)^2 + 2 \int_{\ell = p}^c \frac{pq}{\ell^3} (c-\ell)^2 d \ell.
\]
Thus, the general definition is as follows:
\[
\Phi_{c,(p,q)} = \begin{cases}
0 & \text{if } \max \{ p, q \} \geq c, \\
(1-q/p)(c-p)^2 + 2 \int_{\ell = p}^c \frac{pq}{\ell^3} (c-\ell)^2 d \ell & \text{if } q \leq p \leq c, \\
(1-p/q)(c-q)^2 + 2 \int_{\ell = q}^c \frac{pq}{\ell^3} (c-\ell)^2 d \ell & \text{if } p \leq q \leq c.
\end{cases}
\]

In order to apply Lemma~\ref{lem:concave}, we start by showing that the function $\Phi_{c,(p,q)}$ is differentiable for all $p$ (in the direction of $p$) given a fixed value of $q$, and for all $q$ given a fixed value of $p$. This is done by calculating the two one-sided derivatives in the points suspected of non-differentiability: $p=q$ and $\max \{ p,q\} = c$. 
To state it into more detail, for any fixed $q$, we calculate both
\[
	\frac{\partial \Phi_{c,(p,q)}}{\partial p}_+ = \lim_{h \rightarrow 0^+} \frac{\Phi_{c,(p+h,q)} - \Phi_{c,(p,q)}}{h},
\]
and
\[
\frac{\partial \Phi_{c,(p,q)}}{\partial p}_- = \lim_{h \rightarrow 0^-} \frac{\Phi_{c,(p+h,q)} - \Phi_{c,(p,q)}}{h},
\]
and verify that both values are equal in all suspected points. We do the same switching the roles of $p$ and $q$. (though it is not required as this potential function is symmetric, since we assume the reference distribution to be symmetric)
Additionally, we calculate its second derivatives whenever they are defined. If $\max \{ p,q \} > c$, then they are trivially zero. For $q < p < c$, we get:
\[  \frac{\partial^2 \Phi_{c,(p,q)}}{\partial p^2} = 2(1 - q/p) \]
and
\[  \frac{\partial^2 \Phi_{c,(p,q)}}{\partial q^2} = 0. \]

Actually, there is a reason why this second derivative with respect to $q$ is zero. For any $0 < \delta \leq \min\{ p - q,q\}$, consider a protocol $\pi$ that first moves to $(p,q-\delta)$ or to $(p,q+\delta)$, each with probability $1/2$, and then simulates the buzzer protocol. It has the same leaf distribution as the buzzer protocol (in the pretend world). 
Both the buzzer protocol and $\pi$ either get to the point $(p,0)$ or to the point $(p,p)$, with probabilities $1-q/p$ and $q/p$, respectively. From that point on, both continue the same way, resulting in the same leaf distribution. 
This validates the equality
\[ \Phi_{c,(p,q)} = \frac 1 2 \Phi_{c, (p,q+\delta)} + \frac 1 2 \Phi_{c, (p,q-\delta)} \]
for all $q$ and $\delta$ sufficiently small, which implies linearity in the region $q \in [0,p]$ (given a fixed $p$).

Similar calculations will now be performed with regard to $\SIM_{p,q}(\AND,0)$.
Denote  $x = \nu(0,0), y=\nu(1,0)=\nu(0,1), z=\nu(1,1)$.
It is possible to extract the value of this function from the equations in \cite{MR3210776}, using the conversion from $\SIM$ to $\CI$ \eqref{eq:def-SIM} and from $\CI$ to $\IC$~\eqref{eq:def-CI}. Nevertheless, we calculate it using the formula \eqref{eq:SIM-Ex}, which is an expectation over a value obtained in the leafs of the protocol. Let $p \geq q$, and let $\Pi$ correspond to the buzzer protocol, which starts at distribution $(p,q)$. Then,

\begin{align*}
\SIM_{p,q}(\AND,0)
&= \Ex_{\Pi}[\langle \nu, \mu_{\Pi} \rangle(H_{\mu_{\Pi}}(X|Y) + H_{\mu_{\Pi}}(Y|X))]\\
&=
\left(1 - \frac{q}{p}\right) ((1-p)x + py) h \left(\frac{py}{(1-p)x+py}\right) + \\
&\quad \int_p^1 \frac{2pq}{\ell^3} ((1-\ell)x + \ell y) h\left(\frac{y\ell} {x(1-\ell)+y\ell}\right) \, \rmd \ell \\
&= -\left[
q(1-p)y + (1-p)(1-q)x \log \frac{(1-p)x}{(1-p)x+py} + \right. \\
&\qquad \left. \left(\frac{pqy^2}{x} + (p+q-2pq)y \right) \log \frac{py}{(1-p)x+py}
\right].
\end{align*}

Calculating the second derivative, we get for $p > q$,
\[ \frac{\partial^2 \SIM_{(p,q)}(\AND,0)}{\partial p^2} = -2(1 - q/p) \frac{xy}{2(1-p)p^2((1-p)x + py)}, \]
and 
\[ \frac{\partial^2 \SIM_{(p,q)}(\AND,0)}{\partial q^2} = 0. \]
The reason that the second derivative is zero is the same as explained for the potential function. 
For proving differentiability (on each direction separately), the only suspected point is $p=q$. Comparing the two one-sided derivatives implies the result.

Now we are almost ready to apply Lemma~\ref{lem:concave}.
Define
\[ C = \max_{0 \leq p \leq 1} \frac{2(1-p)p^2((1-p)x + py)}{xy}, \]
and $F(p,q) = C \SIM_\mu(\pi^*) + \Phi_{c,\mu}$.
For any fixed $q$, $\frac{\partial F(p,q)}{\partial p}$ is continuous, piecewise differentiable, and its derivative, $\frac{\partial}{\partial p} \frac{\partial F(p,q)}{\partial p}$ is non-positive wherever it is defined. 
Thus, $\frac{\partial F(p,q)}{\partial p}$ is non-increasing, and $F(p,q)$ is concave as a function of $p$. 
The same holds when switching the roles of $p$ and $q$, thus the conditions in Lemma~\ref{lem:concave} are satisfied,
which concludes the proof of Lemma~\ref{lem:stability-one-bit}.
Finally, we are able to prove Theorem~\ref{thm:stability}.

\begin{proof}[Proof of Theorem~\ref{thm:stability}]
	Let $T$ be the protocol tree of $\pi$. This is a directed binary tree with two children for each internal node. Each node corresponds to a state of the protocol when some communication has taken place, and its children are the two consecutive states, chosen according to the bit sent by the player owning the node.
	
	We can construct $T$ using a sequence of trees, $T_1, T_2, \dots, T_k = T$. 
	The tree $T_1$ contains only the root of $T$, and for all $i$, $T_i$ is obtained from $T_{i-1}$ by adding the children of a leaf of $T_{i-1}$ which is not a leaf of $T$.
	
	Given a tree $T_i$, construct a protocol $\pi_i$, that whenever it reaches  a state represented by node $v$ which is not a leaf of $T_i$, the protocol behaves as $\pi$ for the next bit sent, and if the state is represented by a leaf of $T_i$, then the buzzer protocol is simulated from that point on.
	Let $D$ be the constant in the $O(\cdot)$ guaranteed from Lemma~\ref{lem:stability-one-bit}. 
	The lemma implies that for all $i$,
	$\Phi_{c,\mu}(\pi_i) - \Phi_{c,\mu}(\pi_{i-1}) \leq D (\SIM_{\mu}(\pi_{i-1}) - \SIM_{\mu}(\pi_i))$.
	Summing over $i$, we get a telescopic summation that results in
	\[
	\Phi_{c,\mu}(\pi) 
	= \Phi_{c,\mu}(\pi_k) - \Phi_{c,\mu}(\pi_1) 
	\leq D(\SIM_{\mu}(\pi_1) - \SIM_{\mu}(\pi_k))
	= D(\SIM_{\mu}(\AND,0) - \SIM_{\mu}(\pi)).
	\]
	We used the fact that $\Phi_{c,\mu}(\pi_1) = \Phi_{c,\mu} = 0$, which hold since we assumed that $c \le \max \{ p,q \}$, and the leaf distribution of the buzzer protocol has zero mass on $\ell < \max \{ p,q \}$, therefore its potential cost is zero.
	This finishes the proof as
	\[
	\SIM_{\mu}(\AND,0) - \SIM_{\mu}(\pi) = \langle \nu, \mu \rangle (\CI_{\omu}(\AND,0) - \CI_{\omu}(\pi)) = \langle \nu, \mu \rangle \IW_\omu(\pi) \leq \IW_\omu(\pi). \qedhere
	\]
\end{proof}

\subsection{Lower bound on the information complexity of $\IC_\mu(\AND,\epsilon)$}

%

In this section, we prove Theorem~\ref{thm:AND-gap} by showing that  every distribution $\omu$ which is of full support, except perhaps for $\omu(1,1)$, satisfies  $\IC_\omu(\AND, \epsilon) \geq \IC_\omu(\AND,0) - O(\hc(\epsilon))$. Recall that  Theorem~\ref{thm:AND-gap}~(ii) follows from Part (i) and we have already established the upper bound of Theorem~\ref{thm:AND-gap}~(i) in Theorem~\ref{thm:upper-bd-IC-mu-eps}. Hence it remains to prove the following theorem.

\begin{theorem}[The remaining case of Theorem~\ref{thm:AND-gap}] \label{thm:LB-D}
	Let $\omu$ be a full-support distribution, except perhaps for $\omu(1,1)$. For all $\epsilon \ge 0$,
	\[
	\IC_{\omu}(\AND,\epsilon) \ge \IC_{\omu}(\AND,0) - O_\omu(\hc(\epsilon)).
	\]
	The hidden constant can be fixed if  $\omu(0,0), \omu(0,1), \omu(1,0)$ are bounded away from $0$.
\end{theorem}

The proof uses the idea of \emph{protocol completion}: given a protocol $\pi$ performing $[\AND,\epsilon]$, we can create a protocol
$\pi_0$, which we call the zero-error \emph{completion} of $\pi$. Such a protocol $\pi_0$ takes the following steps:
\begin{itemize}
	\item
	First Alice and Bob simulate $\pi$ until it terminates.
	\item
	Afterwards they run a protocol that solves the $\AND$ function with zero error.
\end{itemize}

The \emph{cost of completion} is the amount of information revealed in the second step, and it is equal to $\IC_\omu(\pi_0) - \IC_\omu(\pi)$. We have shown in the proof of Theorem~\ref{thm:non-distri-error-lowerbd} that  for general functions, this cost is bounded by $O(\hc(\sqrt \epsilon))$, but here we would like to prove a stronger bound of  $O(\hc(\epsilon))$ for protocols that are almost optimal for the AND function. This obviously would yield the desired lower bound, and prove Theorem~\ref{thm:LB-D}.
This completion cost can be arbitrarily close to $\Ex_\Pi[\IC_{\omu_\Pi}(\AND, 0)]$. In order to bound this quantity, we first bound the information complexity of the $\AND$ function.

\begin{lemma} \label{lem:completion-cost-non-product}
	Consider a reference distribution $\nu$ with $\nu(0,0)=x, \nu(1,0)=\nu(0,1)=y,\nu(1,1)=z$, such that $x,y,z > 0$. Let $\mu = (p,q)$ be a pretend distribution.
	Let $\omu = \nu \odot \mu$, and $\omu(1,1)=\delta$.
	Let $0 < C < 1$ be an arbitrary constant.
	
	Firstly $\IC_\omu(\AND,0) \leq 2 \hc(1-\delta)$. Secondly
	\[
	\IC_\omu(\AND,0) \leq
	\begin{cases}
	O(\hc(\delta/z)) & \text{if $\max(p,q) \geq C$}, \\
	O(\hc(\sqrt{\delta/z})) & \text{if $p,q < C$}.
	\end{cases}
	\]
	The hidden constants can be fixed if  $x,y,C$ are bounded away from both $0$ and $1$.
\end{lemma}
\begin{proof}
	First we prove that $\IC_\omu(\AND,0) \leq 2 \hc(1-\delta)$. Assume that $\delta \geq 1/2$, as otherwise the inequality trivially follows. The information complexity is achieved by a protocol where both Alice and Bob send their inputs. The cost of that protocol is at most $H(XY) \leq H(X)+H(Y) \leq 2h(\delta)$.
	
	For proving the other bounds, assume that $\delta < 1/2$, since otherwise the lemma trivially follows. If $p,q > 1/2$, then $\delta = \frac{\nu(1,1) pq}{\langle \nu, \mu \rangle} \geq \nu(1,1) = z$, as
	\[
	\langle \nu,\mu \rangle = (1-p)(1-q)x + [p(1-q)+(1-p)q]y + pqz \leq (x+2y+z)pq = pq.
	\]
	In this case, the lemma follows.
	
	Assume that either $p \leq 1/2$ or $q \leq 1/2$. Without loss of generality, $p \leq q$. We will analyze the protocol in which Alice first sends her input to Bob, and if $X=1$ then Bob sends his input to Alice. This protocol has a cost of
	\[ 
	H(X|Y) + \Pr[X=1] H(Y | X=1) 
	\leq H(X) + \Pr[X=1]
	\leq \hc(\Pr[X=1]) + \Pr[X=1].
	\]
	The obtained bound is monotonic in $\Pr[X=1]$, a fact that we will use.
	
	Now
	\[
	\Pr[X=1]
	= \frac{p(1-q)y+pqz}{\langle \nu, \mu \rangle}
	\leq \frac{p(y+z)}{\langle \nu, \mu \rangle}
	= \frac{\delta (y+z)}{zq}.
	\]
	
	Thus, if $q \geq C$, then the cost of completion is at most
	\begin{equation} \label{eq:bound-ic-small-delta}
	\hc\left( \frac{\delta (y+z)}{z C} \right) + \frac{\delta (y+z)}{z C}
	\leq \frac{(y+z)\delta}{Cz} +
	\begin{cases}
	\hc(\delta/z) & \text{if } \frac{y+z}{C} < 1, \\
	\frac{y+z}{C} 2\hc(\delta/z) & \text{otherwise,}
	\end{cases}
	\end{equation}
	using the bound $\hc(cx) \leq 2c\hc(x)$ for all $c > 1$, from~\eqref{eq:conv-factor}.
	
	If $q \leq C$, $\Pr[X=1]$ is maximized at $q=p$. Assume indeed that $p=q$.
	We will bound its value from below.
	The equation $\frac {q^2z}{\langle \nu, \mu \rangle} = \frac {q^2z}{\langle \nu, \mu \rangle} = \delta$ implies
	\[
	q = \sqrt{\frac{\delta \langle \nu, \mu \rangle}{z}}.
	\]
	Now since
	\[
	\langle \nu, \mu \rangle
	\geq \nu(0,0) \mu(0,0)
	= (1-p) (1-q) x
	\geq (1-C)^2 x,
	\]
	we have
	\[
	\Pr[X=1]
	\leq \frac{\delta (y+z)}{zq}
	\leq \sqrt{\frac{\delta}{z}} \frac{y+z}{(1-C)\sqrt{x}}.
	\]
	The proof concludes applying similar calculations as in \eqref{eq:bound-ic-small-delta}.
\end{proof}

Next, we use this bound to show that if the probability that $\max \{ \leafp, \leafq \}$ does not exceed some constant is very small,
then one can get an improvement over $\hc(\sqrt \epsilon)$ for the completion cost.

\begin{lemma} \label{lem:completion-cost-2}
	Let $\nu$ be a symmetric reference distribution with $\nu(0,0)=x$, $\nu(0,1)=\nu(1,0)=y$ and $\nu(1,1)=z  > 0$. Let $\mu=(p,q)$ be a pretend distribution, and let $\omu = \nu \odot \mu = \nu$.
	
	Let $\pi$ be a protocol performing $[\AND,\epsilon]$.
	Let $0 < C < 1$ be an arbitrary constant, $\kappa = \Pr[\max \{ \leafp,\leafq \} \leq C]$.
	
	The protocol $\pi$ can be completed to a zero-error protocol using an additional information cost of
	\[
	O\left( \kappa \hc(\sqrt{\epsilon/\kappa}) + (1-\kappa) \hc(\tfrac{\epsilon}{1-\kappa}) \right),
	\]
	where the cost is according to the distribution $\omu$, and the hidden constant in $O(\cdot)$ can be fixed if $x,y,p,q,C$ are all bounded away from both $0$ and $1$.
\end{lemma}
\begin{proof}
	First, note that 
	\[
	\omu(1,1)
	= \frac{zpq}{\langle \nu, \mu \rangle}
	\leq \frac{z pq}{x(1-p)(1-q)}
	= O(z).
	\]
	Let $\boldsymbol{\psi}$ be the random variable denoting the completion cost  as a function of $\Pi$. Let $\boldsymbol{1}_{o=b}$ be the indicator of whether $\pi$ outputs $b$ given the transcript $\Pi$, for $b=0,1$. The total completion cost is
	\[
	\Ex[\boldsymbol{\psi}]
	= \sum_{b=0,1} \Ex[\boldsymbol{\psi} \boldsymbol{1}_{o=b}].
	\]
	
	We start by bounding $\Ex[\boldsymbol{\psi} \boldsymbol{1}_{o=1}]$.
	Let $\boldsymbol{\delta}$ be the random variable which equals $\omu_{\Pi}(1,1)$.
	\[
	\Ex[(1-\boldsymbol{\delta}) \boldsymbol{1}_{o=1}]
	= \Pr[(X,Y)\neq (1,1), \pi \text{ outputs 1}]
	\leq \epsilon.
	\]
	From Lemma~\ref{lem:completion-cost-non-product},
	the completion cost $\boldsymbol{\psi}$ is at most $2 \hc(1 - \boldsymbol{\delta})$.
	From the concavity of $\hc$,
	\[
	\Ex[\boldsymbol{\psi} \boldsymbol{1}_{o=1}]
	= \Ex O(\hc(1 - \boldsymbol{\delta})) \boldsymbol{1}_{o=1}
	= \Ex O(\hc((1-\boldsymbol{\delta})\boldsymbol{1}_{o=1}))
	\leq O(\hc(\Ex[(1-\boldsymbol{\delta})\boldsymbol{1}_{o=1}]))
	\leq O(\hc(\epsilon)).
	\]
	This can be bounded as desired since in both cases of $\kappa > 1/2$ and $\kappa \leq 1/2$, we have
	\[  \hc(\epsilon) = O\left( \kappa \hc(\sqrt{\epsilon/\kappa}) + (1-\kappa) \hc(\tfrac{\epsilon}{1-\kappa}) \right). \]
	
	Next we bound $\Ex[\boldsymbol{\psi} \boldsymbol{1}_{o=0}]$.
	\[
	\Ex[\boldsymbol{\delta} \boldsymbol{1}_{o=0}]
	= \Pr[(X,Y)= (1,1), \pi \text{ outputs 0}]
	\leq \epsilon \omu(1,1)
	\leq \epsilon O(z).
	\]
	Let $S$ be the event that $\max \{ \leafp, \leafq \} \leq C$. Then,
	\[
	\Ex[\boldsymbol{\delta} \boldsymbol{1}_{o=0} | S]
	\leq \epsilon O(z) / \Pr[S]
	= \epsilon O(z) / \kappa.
	\]
	\[
	\Ex[\boldsymbol{\delta} \boldsymbol{1}_{o=0} | \overline{S}]
	\leq \epsilon O(z) / (1 - \kappa).
	\]
	
	From Lemma~\ref{lem:completion-cost-non-product}, the completion cost is of order of $\hc \left( \sqrt{\boldsymbol{\delta}/z} \right)$ when $S$ happens, and $\hc(\boldsymbol{\delta}/z)$ otherwise.
	
	\begin{align}
	\Ex[\boldsymbol{\psi} \boldsymbol{1}_{o=0}]
	&= \Pr[S] \Ex[\boldsymbol{\psi} \boldsymbol{1}_{o=0} | S] + \Pr[\overline{S}] \Ex[\boldsymbol{\psi} \boldsymbol{1}_{o=0} | \overline{S}] \notag \\
	&= O\left( \kappa \Ex\left[\hc \left(  \sqrt{\boldsymbol{\delta} \boldsymbol{1}_{o=0}/z} \right) | S\right]
	+ (1-\kappa) \Ex[\hc(\boldsymbol{\delta} \boldsymbol{1}_{o=0}/z) | \overline{S}] \right) \notag \\
	&\leq O\left( \kappa \hc \left( \sqrt{\Ex[\boldsymbol{\delta} \boldsymbol{1}_{o=0} | S] / z} \right)
	+ (1-\kappa) \hc(\Ex[\boldsymbol{\delta} \boldsymbol{1}_{o=0} | \overline{S}]/z) \right) \label{eq:compl-2-concavity} \\
	&\leq O\left( \kappa \hc \left( \sqrt{O(\epsilon)/\kappa} \right)
	+ (1-\kappa) \hc(O(\epsilon)/(1-\kappa)) \right) \notag \\
	&\leq O \left( \kappa \hc\left( \sqrt{\epsilon / \kappa}\right) + (1-\kappa) \hc \left(\frac {\epsilon} {1-\kappa} \right) \right), \label{eq:compl-2-entropy}
	\end{align}
	where \eqref{eq:compl-2-concavity} follows from the concavity of $\hc(\cdot/z)$ and $\hc(\sqrt{\cdot/z})$, and \eqref{eq:compl-2-entropy} follows from~\eqref{eq:conv-factor}.
\end{proof}

Consider an almost optimal protocol $\pi_0$ so that $\IC_{\omu}(\pi_0) - \IC_\omu(\AND,0)$ is small. Our stability result, Theorem~\ref{thm:stability}, translates this to a bound on the potential function introduced in Definition~\ref{def:potential}. The next lemma uses this to show that for such a protocol $\pi_0$, one can obtain a strong bound on  the value of $\kappa$ in  Lemma~\ref{lem:completion-cost-2}.

\begin{lemma} \label{lem:completion-distribution}
	Let $\omu$ be full-support distribution and let $\omu = \nu \odot \mu$ be its decomposition, where $\nu$ is a symmetric reference distribution, and $\mu$ is the pretend distribution.
	Let $c=\max\left\{ \Pr_\mu[X=1], \Pr_\mu[Y=1] \right\}$. Let $\pi$ be an arbitrary protocol, and $\pi_0$ be the completion of $\pi$ to a protocol performing $[\AND,0]$. Then
	\[
	\Pr[\max\{ \leafp,\leafq \} \leq \frac{c}{4}] = O_{c,\mu,\nu}(\IC_{\omu}(\pi_0) - \IC_\omu(\AND,0)),
	\]
	The hidden constant can be fixed if $p,q,\mu(0,0),\mu(0,1),\mu(1,0)$ are all bounded away from both $0$ and $1$, where $\mu=(p,q)$.
\end{lemma}
\begin{proof}
	Let $\leafl_{p,q}$ be the distribution of $\leafl$ that corresponds to the buzzer protocol when it is invoked from a pretend distribution parametrized by $(p,q)$.
	
	We start by showing that for any $0<p,q<1$,
	\[
	\Pr[\leafl_{p,q} \le 2 \max\{ p, q \}] \ge \frac 3 4.
	\]
	Assume without loss of generality that $p \ge q$. Using the leaf distribution from Section~\ref{sec:AND-stability},
	\[
	\Pr[p \le \leafl \le 2p]
	= 2\int_p^{2p} \frac{pq}{\ell^3} \, \rmd\ell + \left(1 - \frac{q}{p}\right) > \frac{3}{4}.
	\]
	
	This implies
	\begin{align*}
	\Pr[\leafl_{\pi_0} \leq \frac{c}{2}] 
	&= \Pr\left[\leafl_{\pi_0} \leq 2 \frac{c}{4}\right] \\
	&\ge \Pr\left[\max \{ \leafp,\leafq \} \leq \frac{c}{4}\right] \Pr\left[\ell_{\leafp,\leafq} \le 2 \max\{ \leafp,\leafq \}\right]  \\
	&\ge \frac{3}{4} \Pr\left[\max \{ \leafp,\leafq \} \leq \frac{c}{4}\right].
	\end{align*}
	Markov's inequality and Theorem~\ref{thm:stability} imply
	\[
	\Pr[\leafl_{\pi_0} \leq \frac{c}{2}] 
	= \Pr[(c - \leafl_{\pi_0})_+^2 \geq \frac{c^2}{4}]
	\leq \frac {\Ex[(c - \leafl_{\pi_0})_+^2]}{c^2/4}
	= \frac {\Phi_{c,\mu}(\pi_0)}{c^2/4}
	= O(\IC_\omu(\pi_0) - \IC(\AND,0)). \qedhere
	\]
\end{proof}

Now we are ready to prove Theorem~\ref{thm:LB-D}, and thus complete the proof of Theorem~\ref{thm:AND-gap}.

\begin{proof}[Proof of Theorem~\ref{thm:LB-D}]
	We first prove the theorem for the full-support distributions. Consider such a distribution $\omu$.
	Let $\pi$ be a protocol performing $[\AND, \epsilon]$.
	We can assume that $\IC_\omu(\pi) \leq \IC_\omu(\AND,0)$, and let $C = \max \{ \Pr_\mu[X=1], \Pr_\mu[Y=1] \} / 4$, $\kappa = \Pr[\max \{ \leafp, \leafq \} \leq C ]$. Lemma~\ref{lem:completion-cost-2} constructs a zero-error protocol $\pi_0$ whose wastage $w$ is at most
	\[
	w = O\left(\kappa \hc\left(\sqrt{\frac{\epsilon}{\kappa}}\right) + (1-\kappa) \hc\left(\frac{\epsilon}{1-\kappa}\right) \right).
	\]
	Lemma~\ref{lem:completion-distribution} states that $\kappa = O(w)$, and so
	\[
	\kappa = O\left(\kappa \hc\left(\sqrt{\frac{\epsilon}{\kappa}}\right) + (1-\kappa) \hc\left(\frac{\epsilon}{1-\kappa}\right)\right).
	\]
	If $\frac{\epsilon}{1-\kappa} \leq 1/2$, then~\eqref{eq:conv-factor} shows that
	\begin{equation} \label{eq:completion-1}
	\kappa = O\left(\kappa \hc\left(\sqrt{\frac{\epsilon}{\kappa}}\right) + \hc(\epsilon)\right).
	\end{equation}
	Otherwise, $\kappa \geq 1-2\epsilon \geq 1/2$ (assuming $\epsilon \leq 1/4$), and so
	\[
	\kappa = O(h(\sqrt{\epsilon}) + (1-\kappa)) = O(h(\sqrt{\epsilon}) + \epsilon),
	\]
	which contradicts $\kappa \geq 1/2$ for small enough $\epsilon$.
	
	Denoting the hidden constant in~\eqref{eq:completion-1} by $M$, we get
	\[
	\left(1 - M h\left(\sqrt{\frac{\epsilon}{\kappa}}\right) \right) \kappa \leq M h(\epsilon).
	\]
	We will show that for small $\epsilon$, this forces $\kappa \leq 2M h(\epsilon)$. Indeed, suppose that $\kappa > 2M h(\epsilon)$, which implies that $\kappa > 2M \epsilon \log (1/\epsilon)$. Then
	\[
	\frac{\epsilon}{\kappa} < \frac{1}{2M \log (1/\epsilon)},
	\]
	and so for small enough $\epsilon$, $Mh(\sqrt{\epsilon/\kappa}) < 1/2$. This shows that
	\[
	\left(1 - M h\left(\sqrt{\frac{\epsilon}{\kappa}}\right) \right) \kappa > \frac{\kappa}{2} > M h(\epsilon),
	\]
	contradicting the inequality above. We conclude that for small $\epsilon$ we have $\kappa = O(h(\epsilon))$.
	
	Applying Lemma~\ref{lem:completion-cost-2} again, we see that
	\[
	\IC_\omu(\pi_0) - \IC_\omu(\pi) \leq
	\kappa O\left( \hc\left(\sqrt{\frac{\epsilon}{\kappa}}\right) \right) + O(\hc(\epsilon)) \leq
	O(\kappa) + O(\hc(\epsilon)) = O(\hc(\epsilon)).
	\]
	Since $\IC_\omu(\pi_0) \geq \IC_\omu(\AND, 0)$, we conclude that $\IC_\omu(\pi) \geq \IC_\omu - O(\hc(\epsilon))$.
	
	Next consider a distribution $\omu$ with $\omu(1,1)=0$, that assigns a strictly positive probability for every other input.
	There is a series of full support distributions, $\omu_1,\omu_2,\dots$ that converge to $\omu$, and assume without loss of generality that for every input $a \in \{0,1\}^2$ and for every $n \in \mathbb{N}$, $\omu_n(a) \ge \omu(a)/2$. 
	From the continuity of information complexity with respect to the tasks $[\AND,0]$ and $[\AND,\epsilon]$,
	\[
	\lim_{n\rightarrow \infty} \IC_{\omu_n}(\AND,0) = \IC_{\omu}(\AND,0),
	\]
	and
	\[
	\lim_{n\rightarrow \infty} \IC_{\omu_n}(\AND,0) = \IC_{\omu}(\AND,0).
	\]
	Assume that $\omu(0,0),\omu(0,1),\omu(1,0)$ are bounded from below.
	It is possible to decompose $\omu$ into $\nu \odot (p,q)$, 
	where $\nu$ is symmetric and $p,q,\nu(0,0),\nu(0,1)$ and $\nu(1,0)$ are bounded.
	This is done by considering a decomposition where $p=1/2$ and $q$ is chosen such that $\nu$ is symmetric.
	Therefore, there is a constant $C > 0$ such that
	\[
	\IC_{\mu_n}(\AND,\epsilon) \ge \IC_{\mu_n}(\AND,\epsilon) - C \hc(\epsilon).
	\]
	Thus,
	\[
	\IC_\mu(\AND,\epsilon) \ge \IC_\mu(\AND,\epsilon) - C \hc(\epsilon).
	\]
\end{proof}

\section{The set disjointness function with error}   \label{sec:DISJ-proof}
In this section we present the proofs of the results concerning  the set disjointness function. It will be  convenient to switch the roles of $0$ and $1$ in the range of the function, and redefine  $\DISJ_n$ as $\DISJ_n(X,Y)= \vee_{i=1}^n (X_i \wedge Y_i)$, i.e. $\DISJ_n(X,Y)=0$ if the inputs are disjoint and it is equal to $1$ otherwise. Obviously, this will not affect the correctness of our results. 

\subsection{Proof of Theorem~\ref{thm:set_disj_cc}}\label{sec:proof:set_disj_cc}

\restate{Theorem~\ref{thm:set_disj_cc}}
{
For the set disjointness function $\DISJ_n$ on inputs of length $n$, we have
\[
 R_\epsilon(\DISJ_n) = n[\IC^0(\AND,0) - \Theta(h(\epsilon))].
\]
}

As discussed in Section \ref{sec:DISJ-result}, we only need to prove the upper bound. In fact, we will prove the following lemma, from which Theorem~\ref{thm:set_disj_cc} follows using  Corollary~\ref{cor:AND-gap}.

\begin{lemma}   \label{lem:amortize-DISJ-upperbd}
For every $\epsilon > 0$ and sufficiently large $n$,
\[
\frac{R_\epsilon(\DISJ_n)}{n}  \le \IC^0(\AND,\epsilon,1\to 0) + o_{n \to \infty}(1).
\]
\end{lemma}

Intuitively, an upper bound like Lemma \ref{lem:amortize-DISJ-upperbd} is essentially a compression result. Besides, as $\DISJ_n$ has a self-reducible structure (see~\cite{SelfRed}), one can make use of this fact together with the Braverman--Rao~\cite{MR3265014} compression. A difficulty is that what we want to solve is $[\DISJ_n, \epsilon]$, that is, the error allowed is non-distributional, while the error unavoidably introduced in the compression phase is distributional. Fortunately, this can be salvaged by a minimax argument introduced in Section 6.2 of~\cite{MR2961528}.

In order to use self-reducibility and compression, one first needs to have a control on the information cost of solving $[\DISJ_n, \epsilon]$.


\begin{lemma}  \label{lem:DISJ-ICn}
For every  $\epsilon > 0$ and sufficiently large $n$,
\[
\IC(\DISJ_n, \epsilon, 1 \to 0) \le n \IC^0(\AND,\epsilon,1\to 0) + o(n),
\]
where $\IC(\DISJ_n, \epsilon, 1 \to 0) \defeq \max_\mu \IC_\mu(\DISJ_n, \epsilon, 1 \to 0)$.
\end{lemma}

The proof is a direct adaptation of the proof for Lemma 8.5 in \cite{MR3210776}.

\begin{proof}
Let $\Omega_0$ denote the set of all measures $\mu$ on $\{0,1\}^2$ with $\mu(1,1)=0$.  Let $\pi$ be a protocol that computes $[\AND, \epsilon, 1 \to 0]$ and satisfies  $\max_{\mu \in \Omega_0} \IC_\mu(\pi) \le \IC^0(\AND,\epsilon,1\to 0) + \delta$ for some small $\delta > 0$. Consider the following protocol $\tau$ that computes $\DISJ_n$ with error.

\begin{framed}
\begin{itemize}
\item Alice and Bob exchange (with replacement using public randomness) $n^{2/3}$ random coordinates. Denote this set of random coordinates by $J$. If for some $j \in J$, $x_j = 1$ and $y_j = 1$, then they output $1$ and terminate.
\item For each coordinate outside $J$, Alice and Bob run the protocol $\pi$ and output $1$ if $\pi$ outputs $1$ on some coordinate. Otherwise they output $0$.
\end{itemize}
\end{framed}

As $\pi$ has one-sided $1 \to 0$ error, obviously $\tau$ has only one-sided $1 \to 0$ error too, and this error happens with probability at most $\epsilon^d \le \epsilon$, where $d$ is the number of coordinates outside $J$ which satisfy $x_j = y_j = 1$ (if $x_j = y_j = 1$ for some coordinate in $J$, there is no error). In particular, $\tau$ computes $[\DISJ_n, \epsilon, 1 \to 0]$.

A direct inspection shows that the remaining proof of Lemma 8.5 in \cite{MR3210776}  depends only on the protocol but not on the specific problem, hence the proof works for our problem too, and the lemma can be proved similarly.
\end{proof}

Next we prove an amortized upper bound for $\DISJ_n$.

\begin{lemma}  \label{lem:DISJ-CCnN}
For every $\epsilon, \delta > 0$, there exists a constant $C > 0$ that depends on $n, \epsilon, \delta$, such that as long as $N \ge C(n, \epsilon, \delta)$, we have
\[
\frac{R_{\epsilon}(\DISJ_{n \times N})}{N} \le (1 + \delta) \IC(\DISJ_n, \epsilon, 1 \to 0).
\]
\end{lemma}

\begin{proof}
We sketch the proof below. More details can be found in Section 6.2 of~\cite{MR2961528}.

\begin{itemize}
\item Step 1. Choose a good protocol for $[\DISJ_n, \epsilon - \xi, 1 \to 0]$ for an appropriate $\xi > 0$.

Denote $I \defeq \IC(\DISJ_n, \epsilon, 1 \to 0)$. By continuity of information complexity (Lemma~\ref{lem:continuity}, which holds for one-sided error with the same proof), there exists $\xi>0$ such that
\[
 \IC(\DISJ_n, \epsilon-\xi, 1 \to 0) \leq \left(1 + \frac{\delta}{6}\right) I.
\]
A minimax argument along the lines of Theorem~3.5 and Theorem~3.6 of~\cite{MR2961528} (but simpler) shows that there exists a protocol $\pi$ that computes $[\DISJ_n, \epsilon - \xi, 1 \to 0]$, and for every distribution $\mu$, its information cost satisfies
\[
\IC_\mu(\pi) \le \left(1 + \frac{\delta}{3}\right) I.
\]
Denote by $r$ the number of rounds in $\pi$.

\item Step 2. Parallel computing.

Let $M = \sqrt[3]{N}$. For an arbitrary distribution $\mu$ on $\{0,1\}^{n \times M} \times \{0,1\}^{n \times M}$, let $\mu_1, \ldots, \mu_M$ be the marginals of $\mu$ restricted to each block of size $n$. Consider $\pi^M$, that is, the execution of $M$ copies of $\pi$ in parallel.
The protocol $\pi^M$ has information cost
\[
\IC_\mu(\pi^M) \le \sum_{i=1}^M \IC_{\mu_i} (\pi) \le \left(1 + \frac{\delta}{3}\right) M \cdot I.
\]
Clearly, $\pi^M$ is still an $r$-round protocol (this is required in order to apply Braverman--Rao compression).

\item Step 3. Compression (with the aid of a minimax argument), and truncation.

By Braverman--Rao compression~\cite{MR3265014} one can find another protocol with communication cost roughly equal to $M \cdot I$, and with an extra small error. However, this extra error is distributional according to the distribution $\mu$. What we want is to solve $[\DISJ_{n\times M}, \epsilon]$, that is, the protocol is only allowed to err with probability at most $\epsilon$ on \emph{every} input.

Fortunately, one can fix this by applying a minimax argument, presented as Claim~6.10 in~\cite{MR2961528}, followed by an extra parallel computation step, presented as Claim~6.11 in~\cite{MR2961528}.

The analog of Claim~6.10 comes up with a protocol $\tau$ with the following properties:
\begin{itemize}
\item For every input in $\{0,1\}^{n \times M} \times \{0,1\}^{n \times M}$, the statistical distance between the output of $\tau$ and the output of $\pi^M$ is $O(1/M^3)$.
\item The expected communication cost of $\tau$ is at most $\left(1 + \frac{\delta}{2}\right) M \cdot I$.
\item The worst-case communication cost of $\tau$ is at most $O(Mn/\delta_1)$.
\end{itemize}
(The statement of Claim~6.10 has $1/M^2$ instead of $1/M^3$, but the proof of Claim~6.10 works for any constant exponent; this can be traced to the fact that the dependence on the error in Braverman--Rao compression is logarithmic.)

The idea now is to run $M^2$ copies of $\tau$ in parallel, truncating the result, as in Claim~6.11 of~\cite{MR2961528}. For large enough $M$ (depending on $n,\epsilon,\delta$), the resulting protocol $\tau'$ satisfies the following properties:
\begin{itemize}
\item For every input in $\{0,1\}^{n \times M \times M^2} \times \{0,1\}^{n \times M \times M^2}$, the statistical distance between the output of $\tau'$ and the output of $\tau^{M^2}$ is at most $\eta$, where $\eta$ tends to zero as $M\to\infty$.
\item The worst-case communication complexity of $\tau'$ is at most $(1+\delta) M^3 \cdot I$.
\end{itemize}

In particular, the statistical distance between $\tau'$ and $\pi^{M^3} = \pi^N$ is at most $\eta + O(1/M)$ on every input, which tends to zero as $M\to\infty$. Choose $M$ large enough to guarantee that the statistical distance between the output of $\tau'$ and the output of $\pi^N$ is at most $\xi$. The protocol $\tau'$ can be used to compute $[\DISJ_{n \times N},\epsilon]$, as in the proof of Lemma~\ref{lem:DISJ-ICn}.
 This completes the proof. \qedhere
\end{itemize}
\end{proof}

Now we prove the upper bound.

\begin{proof}[Proof of Lemma \ref{lem:amortize-DISJ-upperbd}]
Fix $\epsilon > 0$. By Lemma \ref{lem:DISJ-ICn}, there exists $T(\epsilon)$ depending on $\epsilon$ such that
\[
\IC(\DISJ_n, \epsilon, 1 \to 0) \le n \IC^0(\AND,\epsilon,1\to 0) + o(n)
\]
whenever $n \ge T(\epsilon)$. For every such sufficiently large $n$, choose $\delta = \frac{1}{n}$. Lemma \ref{lem:DISJ-CCnN} states that
\[
\frac{R_{\epsilon}(\DISJ_{n \times N})}{N} \le \left(1 + \frac{1}{n}\right) \IC(\DISJ_n, \epsilon, 1 \to 0)
\]
whenever $N \ge C(n, \epsilon)$ for some constant $C(n, \epsilon)$. Since $\IC(\DISJ_n,\epsilon,1\to0) \le n$,
\[
\frac{R_{\epsilon}(\DISJ_{n \times N})}{n \times N}
\le
\IC^0(\AND, \epsilon, 1 \to 0) +\frac{1}{n}+o(1)
\]
for $N \geq C(n, \epsilon)$.
It follows that
\[
\frac{R_{\epsilon}(\DISJ_{M})}{M} \le \IC^0(\AND, \epsilon, 1 \to 0) + o(1)
\]
where $o(1) \to 0$ as $M \to \infty$, completing the proof.
\end{proof}

\subsection{A protocol for Set-Disjointness}\label{sec:proof:setDisj_distrib}

\restate{Theorem~\ref{thm:setDisj_distrib}}{
For the set-disjointness function $\DISJ_n$ on inputs of length $n$, we have
\[
\ICD(\DISJ_n,\epsilon) =n[\IC^0(\AND,0) - \Theta(\sqrt{h(\epsilon)})] + O(\log n).
\]}
\begin{proof}
We already established the lower bound  in \eqref{eq:disj_distrib_lowerbnd}, it remains to prove the upper bound.

Let $\mu$ be an input distribution for $\DISJ_n$, and let $p = \Pr_{\mu}[\DISJ_n(X,Y) = 1]$. We can assume that $p \geq \epsilon$ as otherwise $\IC_\mu(\DISJ_n,\mu,\eps)=0$, and the upper bound trivially holds. Below we introduce a protocol $\pi$ in Figure \ref{fig:disj_protocol} that solves $[\DISJ_n, \mu, \eps]$ and has the desired information cost. In fact, our protocol is stronger in the sense that it has only one-sided error: the protocol $\pi$ always outputs $0$ correctly if the correct output is $0$, and on the other hand, if there are $t \geq 1$ coordinates satisfying $X_i = Y_i = 1$, then  $\pi$ will erroneously output $0$ with probability at most $(\epsilon/2p)^t \le \epsilon/2p$. Thus the distributional error of $\pi$ is at most $p \cdot \frac{\epsilon}{2p} < \epsilon$, and $\pi$ indeed solves $[\DISJ_n,\mu,\eps, 1 \to 0]$.


\begin{figure}[ht]
\begin{framed}
On input $(X,Y)$:
\begin{itemize}
\item Alice and Bob, using public randomness, jointly sample a permutation $\sigma$ on the set $\{1, 2, \ldots, n\}$ uniformly at random; and they run the following sub-protocol $\pi^\sigma$:
\item For $i=1,2, \ldots, n$ repeat: 
      \begin{itemize}
      \item Alice and Bob run a protocol $\pi^\sigma_i$ that is (almost) optimal for $\IC_{\nu_i}(\AND,\eps/2p,1 \to 0)$ on input $(X_{\sigma(i)}, Y_{\sigma(i)})$, where $\nu_i$ is the distribution of $(X_{\sigma(i)}, Y_{\sigma(i)})$ conditioned on the event that the protocol has not yet terminated;
      \item if the protocol $\pi^\sigma_i$ outputs $1$, then terminate and output $1$;
      \end{itemize}
\item  If the ``for-loop'' ends without outputting $1$,  output $0$ and terminate.
\end{itemize}
\end{framed}
\caption{The protocol $\pi$ that solves $[\DISJ_n,\mu,\eps, 1 \to 0]$. \label{fig:disj_protocol}}
\end{figure}


We now analyze the information cost. We start by analyzing the information cost of the sub-protocol $\pi^\sigma$. Let $\Pi^\sigma$ be the transcript of $\pi^\sigma$, and write $\Pi^\sigma=\Pi^\sigma_1 \ldots \Pi^\sigma_n$ where $\Pi^\sigma_i$ denotes the transcript of the protocol $\pi^\sigma_i$ for $i=1,\ldots,n$. As usual let $\Pi^\sigma_{<i} = \Pi^\sigma_1 \ldots \Pi^\sigma_{i-1}$ be the partial transcript. Let $\mu_i$ denote the distribution of $X_{\sigma(i)}Y_{\sigma(i)}$, and $\nu_i$ denote the  distribution of $X_{\sigma(i)}Y_{\sigma(i)}$ conditioned on $\Pi^\sigma_{<i}$. Corollary~\ref{cor:ANDzero-gap}~(iii) gives a bound on the information exchanged in each round: there exist constants $C_1, C_2 > 0$ such that for any distribution $\nu$, 
\[ 
\IC_{\nu}(\AND,\eps/2p,1 \to 0) \le  \IC^0(\AND,0) + C_1 \hc(\nu(1,1)) - C_2 \hc(\epsilon/p). 
\]
Note that $(\Pi^\sigma_i | XY\Pi^\sigma_{<i})$ has the same distribution as $(\Pi^\sigma_i | X_{\sigma(i)}Y_{\sigma(i)}\Pi^\sigma_{<i})$, and thus
\begin{align*}
I(Y ; \Pi^\sigma | X)
&= \sum_{i=1}^n I(Y ;\Pi^\sigma_i | X, \Pi^\sigma_{<i}) = \sum_{i=1}^n [H(\Pi^\sigma_i | X, \Pi^\sigma_{<i}) - H(\Pi^\sigma_i | XY, \Pi^\sigma_{<i})] \\
&\leq \sum_{i=1}^n [H(\Pi^\sigma_i | X_{\sigma(i)},  \Pi^\sigma_{<i}) - H(\Pi^\sigma_i | X_{\sigma(i)} Y_{\sigma(i)}, \Pi^\sigma_{<i})] \\
&= \sum_{i=1}^n I(Y_{\sigma(i)} ;\Pi^\sigma_i | X_{\sigma(i)}, \Pi^\sigma_{<i}).
\end{align*}
Thus, denoting by $T^\sigma$ the number of AND protocols executed before the termination of $\pi^\sigma$, the above inequality implies (note that $\nu_i$ is a random variable, and $\pi^\sigma_i$ depends on $\nu_i$)
\begin{align*}
\IC_\mu(\pi^\sigma)
&\leq \sum_{i=1}^n \Ex  \IC_{\nu_i}(\pi^\sigma_i) \leq \sum_{i=1}^n \Pr[T^\sigma \ge i]\Ex \left[\IC_{\nu_i}(\pi^\sigma_i) \ | \ T^\sigma \ge i\right]\\
&\leq \sum_{i=1}^n \Pr[T^\sigma \geq i] \Ex \left[\IC^0(\AND,0) + C_1 \hc(\nu_i(1,1)) - C_2 \hc(\epsilon/p) \ | \ T^\sigma \ge i\right] \\
&\leq \left(\IC^0(\AND,0) - C_2 \hc(\epsilon/p) \right) \Ex[T^\sigma]
		+C_1 \sum_{i=1}^n \Pr[T^\sigma \geq i] \Ex \left[ \hc(\nu_i(1,1)) | T^\sigma \geq i \right].
\end{align*}
We want to bound the second term. Note since $p \geq \epsilon$,
\[
	\Pr[T^\sigma=i | T^\sigma \ge i,X_{\sigma(i)}=Y_{\sigma(i)}=1] 
	= \Pr[\pi^\sigma_i(X_{\sigma(i)}Y_{\sigma(i)}) = 1 | T^\sigma \ge i,X_{\sigma(i)}=Y_{\sigma(i)}=1]
	\geq 1 - \frac{\epsilon}{2p}
	\geq 1/2.
\] 
Hence, applying~\eqref{eq:conv-factor} twice and using the concavity of $\hc$, we get
\begin{align*}
\Pr[T^\sigma \geq i] \Ex \left[ \hc(\nu_i(1,1)) | T^\sigma \geq i \right]
&\le \Pr[T^\sigma \geq i] \hc \left( \Ex\left[ \nu_i(1,1) | T^\sigma \geq i \right] \right)\\
&= \Pr[T^\sigma \geq i] \hc(\Pr[X_{\sigma(i)} = Y_{\sigma(i)} = 1|T^\sigma \geq i])\\
&\le \hc(\Pr[X_{\sigma(i)} = Y_{\sigma(i)} = 1|T^\sigma \geq i] \Pr[T^\sigma \geq i])\\
&= \hc(\Pr[T^\sigma \ge i,X_{\sigma(i)}=Y_{\sigma(i)}=1]) \\ 
&\le 2\Pr[T^\sigma=i | T^\sigma\geq i,X_{\sigma(i)}=Y_{\sigma(i)}=1] 
		\hc(\Pr[T^\sigma \ge i,X_{\sigma(i)}=Y_{\sigma(i)}=1]) \\
&\le 2  \hc(\Pr[T^\sigma=i, X_{\sigma(i)}=Y_{\sigma(i)}=1]) \\
&\le 2  \hc(\Pr[T^\sigma=i, \pi(X,Y)=1]).
\end{align*}
Using concavity of $\hc$ again,
\[
\frac{1}{n} \sum_{i=1}^n \hc(\Pr[T^\sigma=i,\pi(X,Y)=1]) \leq \hc(\Pr[(X,Y)=1]/n) = \hc(p/n).
\]
Therefore 
\[ \sum_{i=1}^n \Pr[T^\sigma \geq i] \Ex \left[ \hc(\nu_i(1,1)) | T^\sigma \geq i \right] \leq 2n\hc(p/n). \]
That is, we have shown
\begin{equation}   \label{eq:pi-sigma-bound}
\IC_\mu(\pi^\sigma)
\le
\left(\IC^0(\AND,0) - C_2 \hc(\epsilon/p) \right) \Ex[T^\sigma] + 2C_1 n\hc(p/n).
\end{equation}
Taking the expectation with respect to $\sigma$, we obtain 
\begin{equation}   \label{eq:pi-bound}
\IC_\mu(\pi) 
= \Ex_\sigma \IC_\mu(\pi^\sigma) = \left(\IC^0(\AND,0) - C_2 \hc(\epsilon/p) \right) \Ex_{\sigma,XY} [T^\sigma] + 2C_1 n\hc(p/n).
\end{equation}
Hence it remains to bound $\Ex[T^\sigma]$ where the expectation is over $\sigma$ and the input $XY$. 


Let $x,y$ be such that $\DISJ(x,y)=1$, and let $j$ be an index such that $\AND(x_j,y_j)=1$. Then
\begin{align*}
	&\Ex_{\sigma,XY} [T^\sigma|XY=xy] =\sum_{i=1}^n \Pr[\sigma(i)=j]\Ex[T^\sigma|XY=xy,\sigma(i)=j] =\frac{1}{n} \sum_{i=1}^n \Ex[T^\sigma|XY=xy,\sigma(i)=j] \\
	&=\frac{1}{n} \sum_{i=1}^n \sum_{b=0,1} \Ex[T^\sigma|XY=xy,\sigma(i)=j, \pi^\sigma_i(X,Y) = b] 
			\Pr[\pi^\sigma_i(X,Y) = b | XY=xy,\sigma(i)=j] \\
	&\le \frac{1}{n} \sum_{i=1}^n \Big( i \Pr[\pi^\sigma_i(X,Y) = 1 | XY=xy,\sigma(i)=j] 
			+ n \Pr[\pi^\sigma_i(X,Y) = 0 | XY=xy,\sigma(i)=j]  \Big) \\
	&\le \frac{1}{n} \sum_{i=1}^n  \left( i (1-\frac{\epsilon}{2p})  + n \frac{\epsilon}{2p} \right) 
	  = (1-\frac{\eps}{2p}) \frac{n+1}{2} + \frac{\eps}{2p} n 
	  \le \frac{n+1}{2} + \frac{\eps}{4p} n.	
\end{align*}

This allows us the next bound:
\begin{align}
	\Ex_{\sigma, XY}[T^\sigma]
	&= \Pr[\DISJ(X,Y)=1] \Ex[T|\DISJ(X,Y)=1] + \Pr[\DISJ(X,Y)=0] \Ex[T|\DISJ(X,Y)=0]  \nonumber \\
	&\leq p \left(  \frac{n+1}{2} + \frac{\epsilon}{4p}n \right) + (1-p) n  
	   \leq \frac{2p}{3}n + \frac{\epsilon}{4}n + (1-p) n = (1-p/3 + \eps/4) n.  \label{eq:expect-T}
\end{align}
Combine \eqref{eq:pi-bound} and \eqref{eq:expect-T} we get
\begin{align*}
	\IC_\mu(\pi)
	&\leq  n(1-p/3 + \eps/4) \left( \IC^0(\AND,0) - C_2 \hc(\epsilon/p) \right)
			+ C_1 2 n \hc(p/n)\\					
	&= n (\IC_0(\AND,0) - \Omega(\hc(\epsilon/p) + p)) + O(n \hc(p/n)).
\end{align*}

It remains to optimize over $p$. We start by minimizing $p + \hc(\epsilon/p)$. Up to a constant multiple, the minimum is attained at the point satisfying $p = \hc(\epsilon/p)$. A simple calculation shows that $p \approx \sqrt{h(\epsilon)}$, and so $p + \hc(\epsilon/p) = \Omega(\sqrt{h(\epsilon)})$. Thus
\[
	\IC_\mu(\pi) \leq n[\IC^0(\AND,0) - \Omega(\sqrt{h(\epsilon)})] + O(n\hc(p/n)).
\]
The value of the error term $O(n\hc(p/n))$ is at most $O(n\hc(1/n)) = O(n\frac{\log n}{n})=O(\log n)$, and the theorem follows.
\end{proof}

\section{Open problems and concluding remarks}

\begin{itemize}
 \item In Conjecture~\ref{conj:SetDisjointExact} we speculated that the exact asymptotics of $R_\epsilon(\DISJ_n)$ is given by the information complexity of the AND function when only one-sided error is allowed:
 ֿ\[ R_\epsilon(\DISJ_n) =  n \IC^0(\AND,\epsilon,1\to 0) \pm  o(n). \]

 The set disjointness function has a ``self-reducible'' structure in the sense that it is possible to solve an instance of the corresponding communication problem by dividing the input into blocks and solving the  same problem on each block separately. This structure allows us to relate  the communication complexity of the problem to its  amortized communication complexity, and thus to its information complexity via the fundamental result of Braverman and Rao~\cite{MR3265014}.  Applying such ideas we showed (the lower bound is obvious)
\[ \IC(\DISJ_n,\eps)   \le R_\epsilon(\DISJ_n) \le  m \IC(\DISJ_{\frac{n}{m}},\eps,1 \to 0) + o(n), \]
for an appropriate choice of $m=m(n)$ that tends to infinity as $n \to \infty$.
 In Theorem~\ref{thm:set_disj_cc} we combined this with our analysis of the information complexity of the set disjointness to prove $R_\epsilon(\DISJ_n) = n[\IC^0(\AND,0) - \Theta(h(\epsilon))]$.  More precisely we showed
\[
 n \IC^0(\AND,\epsilon) \le \IC(\DISJ_n,\eps) \le \IC(\DISJ_n,\eps,1 \to 0)  \le n \IC^0(\AND,\epsilon,1\to 0) + o(n),
\]
and combined it with our results regarding the information complexity of the $\AND$ function. We believe that the upper bound is the truth; that is
\[ \IC(\DISJ_n,\eps) \ge n \IC^0(\AND,\epsilon,1\to 0) - o(n), \]
which would imply Conjecture~\ref{conj:SetDisjointExact}.

 \item The example of the AND function shows that the $\Omega(h(\eps))$ gain in the information cost, appearing in our upper bounds in Theorems~\ref{thm:upper-bd-IC-mu-eps},~\ref{thm:bd-IC-mu--distributional},~\ref{thm:non-distributional-ub} and~\ref{thm:ICD-eps-bound} is tight. However we do not know whether the $O(h(\sqrt{\eps}))$ gain appearing in the lower bounds in Theorems~\ref{thm:non-distri-error-lowerbd}~and~\ref{thm:bd-IC-mu--distributional}, Corollary~\ref{cor:non-distri-error-lowerbd-Prior-Free} and Theorem~\ref{thm:ICD-eps-bound} is sharp. In fact we are not aware of any example that exhibits a gain  that is not $\Theta(h(\eps))$.   Is it true that for every function $f\colon \cX \times \cY \to \cZ$, and  measure $\mu$ on $\cX \times \cY$ with $\IC_\mu(f,0)>0$, we have $\IC_\mu(f,\eps)=\IC_\mu(f,0)-\Theta(h(\eps))$?  One can ask a similar question for  $\IC_\mu(f,\mu,\eps)$, $\IC(f,\eps)$, and $\ICD(f,\eps)$.

 \item Recall that the \emph{inner product function} $\mathrm{IP}_n\colon \{0,1\}^n \times \{0,1\}^n \to \{0,1\}$ is defined as
 \[ \mathrm{IP}_n\colon (x,y) \mapsto \sum_{i=1}^n x_i y_i \mod 2. \]
 Let $\nu$ denote the uniform probability measure on $\{0,1\}^n \times \{0,1\}^n$. It is easy to see that $\IC_\nu(\mathrm{IP}_n, \nu, \eps) \le (1 - 2 \eps)n$. In \cite[Theorem~1.3]{SelfRed}, Braverman et al.\ exploited the self-reducibility properties of the inner product function to showed that for every $\delta>0$, there exists  an $\eps>0$ and $n_0>0$ such that for every $n>n_0$, $\IC(\mathrm{IP}_n,\eps)>(1-\delta)n$.

 In \cite[Problem~1.4]{SelfRed} they ask whether the dependency of $\delta$ on $\eps$ is linear. In other words, is there  a constant $\alpha>0$ such that for every sufficiently small  $\eps>0$  and sufficiently large $n$,   $\IC_\nu(\mathrm{IP}_n, \nu,\eps) \ge (1 - \alpha \eps)n$?  If yes, then can we take $\alpha \approx  2$, or more precisely,  is it true that  $\IC_\nu(\mathrm{IP}_n, \nu, \eps) = (1 - 2 \eps - o(\eps))n$? Note that the bound $\IC_\nu(f,\nu,\eps)<\IC_\nu(f,\nu,0)-\Omega(h(\eps))$ of Theorem~\ref{thm:bd-IC-mu--distributional} does not refute these possibilities as in these questions $\eps$ is fixed, and asymptotics are as $n \to \infty$.   

 \item The  focus of this paper has been on the internal information complexity, and except for few results such as Proposition~\ref{prop:ext_XOR}, we have not studied the external information complexity analogues. However considering that external information complexity is typically simpler than internal information complexity, we believe that the analogues of many of our results, specially those about the AND function, can be proven for this case as well. We defer this to future research.

\end{itemize}


\newpage

\bibliographystyle{amsalpha}
\bibliography{InformationError}
\end{document}

%% file: stdinc.tex
\usepackage{enumerate}

\usepackage{amsmath}
\usepackage{amstext,amssymb,amsfonts}
\usepackage{fullpage}
\usepackage{bm}

\usepackage{todonotes}

\newcounter{mynotes}
\usepackage{nameref}
\usepackage[pagebackref,colorlinks,linkcolor=blue,filecolor = blue,
citecolor = blue, urlcolor = blue,pdfstartview=FitH]{hyperref}
\usepackage[capitalize]{cleveref}

\usepackage{amsthm}
\usepackage{thmtools}

\declaretheorem[within=section]{theorem}
\declaretheorem[sibling=theorem]{corollary}

\declaretheorem[sibling=theorem]{lemma}

\declaretheorem[sibling=theorem]{definition}
\declaretheorem[sibling=theorem]{Lemma+Definition}
\declaretheorem[sibling=theorem]{remark}

\declaretheorem[sibling=theorem]{conjecture}
\declaretheorem[sibling=theorem]{proposition}

\crefname{proposition}{Proposition}{Propositions}
\crefname{conjecture}{Conjecture}{Conjectures}
\crefname{claim}{Claim}{Claims}
\crefname{remark}{Remark}{Remarks}

\newcounter{termcounter}
\renewcommand{\thetermcounter}{\Alph{termcounter}}
\crefname{term}{term}{terms}
\creflabelformat{term}{\textup{(#2#1#3)}}

\makeatletter
\def\term{\@ifnextchar[\term@optarg\term@noarg}
\def\term@optarg[#1]#2{%
  \textup{(#1)}%
  \def\@currentlabel{#1}%
  \def\cref@currentlabel{[][2147483647][]#1}%
  \cref@label[term]{#2}}
\def\term@noarg#1{%
  \refstepcounter{termcounter}%
  \textup{(\thetermcounter)}%
  \cref@label[term]{#1}}
\makeatother

\newcommand{\ignore}[1]{}

\newcommand{\defeq}{\stackrel{\mathrm{def}}=}
\newcommand{\supp}{\mathrm{supp}}

\definecolor{DSred}{rgb}{1,0,0}

\renewcommand{\leq}{\leqslant}
\renewcommand{\geq}{\geqslant}
\renewcommand{\ge}{\geqslant}
\renewcommand{\le}{\leqslant}
\renewcommand{\epsilon}{\varepsilon}
\newcommand{\eps}{\epsilon}

\newcommand{\cL}{\mathcal L}

\newcommand{\cX}{\mathcal X}
\newcommand{\cY}{\mathcal Y}
\newcommand{\cZ}{\mathcal Z}

\newcommand{\Psymb}{{\bf Pr}}

\DeclareMathOperator*{\ProbOp}{\Psymb}

\renewcommand{\Pr}{\ProbOp}

\newcommand{\Ex}[1]{\E\Brac{#1}}



%% file: header.tex
\usepackage{url}
\usepackage{latexsym}
\usepackage{amsmath}
\usepackage{amsfonts}
\usepackage{amssymb}
\usepackage{amsthm}
\usepackage{mathrsfs}
\usepackage{enumerate}

\allowdisplaybreaks    

\newcommand{\restate}[2]{\medskip
\noindent{\bf #1 (restated).}{\sl #2}}

\definecolor{Blue}{rgb}{0,0,1}